\documentclass[a4paper,11pt]{amsart} 
\usepackage{amsmath,amsxtra,amssymb,latexsym, amscd,amsthm}
\usepackage[mathscr]{eucal}
\usepackage{mathrsfs}
\usepackage{bbm}
\usepackage{enumerate}
\usepackage{pict2e}
\usepackage{tikz}
\usepackage{graphicx}
\usepackage[a4paper]{geometry}
\usepackage{color}
\usepackage{hyperref}
\numberwithin{equation}{section}

\theoremstyle{plain}
\newtheorem{thm}{Theorem}[section]

\newtheorem{prp}[thm]{Proposition}
\newtheorem{cor}[thm]{Corollary}
\newtheorem{lem}[thm]{Lemma}
\newtheorem*{conj*}{Conjecture}
\newtheorem*{fact*}{Fact}
\theoremstyle{definition}

\newtheorem{rem}[thm]{Remark}

\newtheorem{exa}[thm]{Example}
\newtheorem*{rem*}{Remark}

\newtheorem*{q1}{Problem 1}
\newtheorem*{q2}{Problem 2}
\newcommand{\dd}{\mathrm{d}}
\newcommand{\ee}{\mathrm{e}}
\newcommand{\ii}{\mathrm{i}}

\newcommand{\N}{\mathbb{N}}
\newcommand{\Z}{\mathbb{Z}}

\newcommand{\R}{\mathbb{R}}
\newcommand{\C}{\mathbb{C}}

\newcommand{\T}{\mathbb{T}}

\newcommand{\cT}{\mathcal{T}}
\newcommand{\cN}{\mathcal{N}}
\newcommand{\cA}{\mathcal{A}}

\newcommand{\cH}{\mathcal{H}}

\newcommand{\cF}{\mathcal{F}}
\newcommand{\cP}{\mathcal{P}}

\newcommand{\cI}{\mathcal{I}}
\newcommand{\fa}{\mathfrak{a}}

\newcommand{\fh}{\mathfrak{h}}

\DeclareMathOperator{\ran}{Ran}

\DeclareMathOperator{\supp}{supp}

\DeclareMathOperator{\id}{id}

\DeclareSymbolFont{extraup}{U}{zavm}{m}{n}
\DeclareMathSymbol{\varheart}{\mathalpha}{extraup}{86}
\DeclareMathSymbol{\vardiamond}{\mathalpha}{extraup}{87}

\title{Flat bands of periodic graphs}
\author{Mostafa Sabri, Pierre Youssef}
\address{Department of Mathematics, Faculty of Science, Cairo University, Giza 12613, Egypt.}
\address{Science Division, New York University Abu Dhabi, Saadiyat Island, Abu Dhabi, UAE.}
\email{mostafa.sabri@nyu.edu}
\address{Science Division, New York University Abu Dhabi, Saadiyat Island, Abu Dhabi, UAE \& Courant Institute of Mathematical Sciences, New York University, 251 Mercer st, New York, NY 10012, USA.}
\email{yp27@nyu.edu}
\usepackage{calc}
\usepackage{graphicx}

\makeatletter
\newlength{\temp@wc@width}
\newlength{\temp@wc@height}
\newcommand{\widecheck}[1]{%
  \setlength{\temp@wc@width}{\widthof{$#1$}}%
  \setlength{\temp@wc@height}{\heightof{$#1$}}%
  #1\hspace{-\temp@wc@width}%
  \raisebox{\temp@wc@height+2pt}[\heightof{$\widehat{#1}$}]%
     {\rotatebox[origin=c]{180}{\vbox to 0pt{\hbox{$\widehat{\hphantom{#1}}$}}}}%
}
\makeatother

\begin{document}

\begin{abstract}
We study flat bands of periodic graphs in a Euclidean space. These are infinitely degenerate eigenvalues of the corresponding adjacency matrix, with eigenvectors of compact support. We provide some optimal recipes to generate desired bands, some sufficient conditions for a graph to have flat bands, we characterize the set of flat bands whose eigenvectors occupy a single cell and we compute the list of such bands for small cells. We next discuss stability and rarity of flat bands in special cases. Additional folklore results are proved and many questions are still open.
\end{abstract}

\maketitle

\section{Introduction}

Consider a connected, locally finite graph $\Gamma$ which is invariant under translation by some linearly independent vectors $\fa_1,\dots,\fa_d$ (we say $\Gamma$ is $\Z^d$-periodic). Using Floquet theory (Section~\ref{sec:gen}), one sees that $\sigma(\cA_{\Gamma})$, the spectrum of the $0/1$ adjacency matrix of $\Gamma$, consists of bands. At least one of these bands will be non-degenerate (absolutely continuous) but curiously, degenerate bands do occur here, in contrast to the classical periodic Schr\"odinger operators in $\R^d$. These degenerate bands are infinitely degenerate eigenvalues for $\cA_{\Gamma}$, and they always have corresponding eigenvectors of compact support. See Figure~\ref{fig:boxnover} for an example. They are called \emph{flat bands for $\Gamma$}, the terminology coming from the fact that the corresponding Floquet eigenvalue $E_j(\theta)$ is constant in the quasimomentum $\theta$ and thus traces a single point as $\theta$ varies in $[0,1)^d$, instead of a nontrivial interval $[a,b]$. The graph of $\theta\mapsto E_j(\theta)$ is thus completely flat, see e.g. \cite[Fig.6]{RMS} for an illustration.

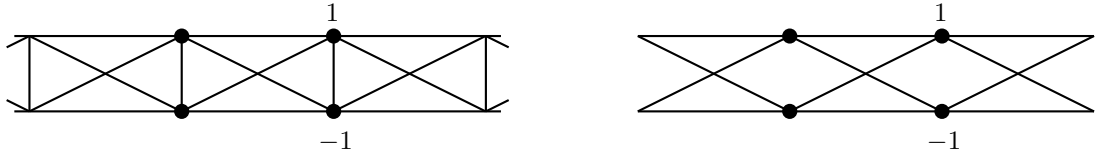
\begin{figure}[h!]
\begin{center}
\setlength{\unitlength}{1cm}
\thicklines
\begin{picture}(1.3,1.4)(-0.8,-1.2)
   \put(-7.2,0){\line(1,0){6.4}}
	 \put(-7.2,-1){\line(1,0){6.4}}
	 \put(-7,-1){\line(2,1){2}}
	 \put(-7,0){\line(2,-1){2}}
	 \put(-5,-1){\line(0,1){1}}
	 \put(-5,-1){\line(2,1){2}}
	 \put(-5,0){\line(2,-1){2}}
	 \put(-3,-1){\line(0,1){1}}
	 \put(-1,-1){\line(0,1){1}}
	 \put(-7,-1){\line(0,1){1}}
	 \put(-3,-1){\line(2,1){2}}
	 \put(-3,0){\line(2,-1){2}}
	 \put(-3,-1){\circle*{.2}}
	 \put(-3,0){\circle*{.2}}
	 \put(-5,-1){\circle*{.2}}
	 \put(-5,0){\circle*{.2}}
	 \put(-3.1,0.2){\small{$1$}}
	 \put(-3.2,-1.5){\small{$-1$}}
	 \put(-1,-1){\line(2,1){0.3}}
	 \put(-1,0){\line(2,-1){0.3}}
	 \put(-7,-1){\line(-2,1){0.3}}
	 \put(-7,0){\line(-2,-1){0.3}}
	 \put(1,0){\line(1,0){6}}
	 \put(1,-1){\line(1,0){6}}
	 \put(1,-1){\line(2,1){2}}
	 \put(1,0){\line(2,-1){2}}
	 \put(3,-1){\line(2,1){2}}
	 \put(3,0){\line(2,-1){2}}
	 \put(5,-1){\line(2,1){2}}
	 \put(5,0){\line(2,-1){2}}
	 \put(5,-1){\circle*{.2}}
	 \put(5,0){\circle*{.2}}
	 \put(3,-1){\circle*{.2}}
	 \put(3,0){\circle*{.2}}
	 \put(4.9,0.2){\small{$1$}}
	 \put(4.8,-1.5){\small{$-1$}}
\end{picture}
\caption{Flat bands $\lambda=-1$ (left) and $\lambda=0$ (right). The values of eigenvectors are given (they are extended by zero to the infinite $\Gamma$).}\label{fig:boxnover}
\end{center}
\end{figure}

There has been intensive activity in the physics community in recent years regarding these flat bands, as they have found applications in the contexts of superfluidity, topological phases of matter and many-body physics; see \cite{FLBMD,MRPS,RY19,RMS} and references therein.

Mathematically, the topic has been much less explored. The basic problem that one would want to consider is the following: ``Given a connected periodic $\Gamma$, by looking at the graph, decide whether or not $\Gamma$ has a flat band.''

More precisely, let $V_f$ be the vertices of the motif which is copied by integer translations. The translated motifs should not intersect one another and should span $\Gamma$. Let
\[
\nu = |V_f|\,,
\]
e.g. $\nu=2$ in Figure~\ref{fig:boxnover}. Several choices of $V_f$ are possible, fix one. Let $\cN_{v_i}$ be the set of neighbors of $v_i$ in the full graph $\Gamma$.

\begin{q1}
Given the data of $\nu$, $d$ and $\cN_{v_i}$ for $i=1,\dots,\nu$, conclude whether $\Gamma$ has a flat band or not.
\end{q1}

This problem is completely open, and it may be too ambitious as stated. Note that $v_i$ can hop to large lattice distances, at least from a mathematical standpoint. We reformulate this question algebraically in \S\ref{sec:alg}.

\subsection{Main results} 
In this paper we take some small steps towards understanding the emergence of flat bands. In Section~\ref{sec:gen}, we gather and prove some folklore results. In Section~\ref{sec:gens}, we provide recipes to generate graphs with or without flat bands. More precisely, we give a construction in Theorem~\ref{thm:eigenchar} to generate flat bands from eigenvalues of finite graphs, and we identify corresponding eigenvectors. Our generator is shown to be optimal in the size of $V_f$. In the special case where the eigenvalue comes from a regular graph, we provide a more efficient generator. On the other hand, we give simple operations to construct new graphs from old ones which preserve the lack of flat bands in \S\ref{sec:gennof}.

Next, we characterize the set of flat bands for which there exist eigenvectors localized within a fundamental cell in Theorem~\ref{thm:singlecellcriter}. These types of flat bands have been known to be more tractable in the physics community \cite{FLBMD}. We apply our criterion to deduce further structural results about these bands, and also to compute explicitly the set $\cF_\nu^s$ of such flat bands for $\nu\le 5$ in Section~\ref{scfb} (the superscript $s$ stands for ``single-cell''). In doing so, we notice a curious property about the growth of $\cF_\nu^s$ with $\nu$ and formulate an open question. 

Back to the main ``Problem 1'' of characterizing flat bands in general, we show in Section~\ref{sec:symm} that if $\Gamma$ is invariant under some symmetries, then it has flat bands. This approach appeared originally in \cite{RMS}, and we give a more general and precise statement, using a different proof. These symmetry conditions are not necessary however.

The converse, namely finding sufficient conditions for the absence of flat bands, seems a lot more challenging. There is the following:

\begin{fact*}
If $\nu=1$ and $d$ is arbitrary, then the spectrum consists of a single band of absolutely continuous spectrum.
\end{fact*}

Indeed, $\cA$ is then equivalent to multiplication by $E(\theta) = \sum_{p=1}^D 2\cos(2\pi\theta\cdot n^{(p)})$ in $L^2[0,1)^d$, for some $n^{(p)}\in\N^d\setminus\{0\}$. As $E(\theta)$ is analytic and nonconstant, the claim follows.

Even the simple case $\nu=2$ is not completely understood however. We investigate it in more details and give some partial results in Section~\ref{sec:2ver}. For $\nu>1$, the only general result we know concerning the absence of flat bands is the following one of \cite{HN}:

\begin{fact*}
Fix a finite multigraph $G$. If $G$ has a $2$-factor, then there is a recipe to create a $\Z$-periodic graph $\Gamma$ with quotient graph $G$, such that $\Gamma$ has no flat bands. In particular, if $G$ has a $2$-factor, then the maximal abelian cover of $G$ has no flat bands.
\end{fact*}

Recall that a graph $G$ is called \emph{Hamiltonian} if there is a cycle in $G$ which covers all vertices of $G$. The condition that $G$ has a $2$-factor is a generalization, which means that there is a collection of disjoint cycles which cover all vertices of $G$. This condition is satisfied in particular if $G$ is even-regular or bipartite odd-regular, and \cite{HN} conjecture that the maximal abelian cover of any regular multigraph should have no flat bands.

In the preceding discussion we were only interested in the adjacency matrix of $\Gamma$, so there are no weights on the edges and no potentials. This is in the hope of getting a completely geometric characterization, and this restriction makes some aspects of the mathematics more interesting. We do not expect the questions to be any different for the variant $D-\cA_\Gamma$ (the discrete Laplacian, where $D$ is the degree matrix). Physicists do endow their graphs with periodic weights. The Floquet theory is still applicable and it makes sense to study flat bands of $\cA_\Gamma+Q$ for instance, for a periodic potential $Q$. Looking at the physics literature, it is apparent that flat bands are ``rare'' and that they are a ``fine tuning'' phenomenon, but we only found examples rather than theorems in this respect. This motivates the following problem:

\begin{q2}
Show that for any periodic connected $\Gamma$, the set
\begin{equation}\label{e:pgamma}
\cP_\Gamma=\{(Q_1,\dots,Q_\nu)\in \R^\nu : \cA_\Gamma +Q \text{ has a flat band}\}
\end{equation}
has dimension at most $\nu-1$.
\end{q2}

We show in Section~\ref{sec:destro} that $\cP_\Gamma$ is a closed semialgebraic set. An answer to Problem 2 would imply that $\cP_{\Gamma}$ has Lebesgue measure zero, reflecting the rarity of flat bands, and is nowhere dense, reflecting the fine tuning (it would contain no balls: given $\cA_\Gamma+Q$, there would exist $Q'$ arbitrarily close to $Q$ such that $\cA_\Gamma+Q'$ has no flat bands). Showing these weaker properties would already be interesting for physical purposes. 

We also show in Section~\ref{sec:destro} that for any $\nu$ there exists some $\Gamma$ with $\mathrm{dim}(\cP_\Gamma) \ge \nu-1$, so the conjectured bound on the dimension cannot be improved in general. Moreover, we answer Problem~2 completely for $\nu=2$; it remains open for $\nu\ge 3$.

\medskip

For the reader's convenience, we summarize the results of this paper:
\begin{enumerate}[(1)]
\item For any strictly positive periodic weight $w$ on the edges and any periodic potential $Q$, the top band of $\cH=\cA_w+Q$ is never flat. There are examples of complex Hermitian weights $w$ such that the spectrum of $\cH$ is entirely flat.

$\cH$ has a localized eigenvector iff the Floquet matrix $H(\theta)$ has an eigenvector which is a trigonometric polynomial in each entry. There is a constructive recipe to derive one from the other. (Section~\ref{sec:gen}).
\item If $\Gamma$ is invariant under a permutation of $V_f$ consisting of $r$ cycles, which exchanges stars around vertices, then $\Gamma$ has at least $\nu-r$ flat bands. In particular, if $\cN_{v_i}\setminus\{v_j\}=\cN_{v_j}\setminus\{v_i\}$ for some $v_i,v_j\in V_f$, then $\Gamma$ has a flat band $\lambda\in \{0,-1\}$. (Section~\ref{sec:symm}).
\item We have a recipe to make all eigenvalues of any finite graph $G$ appear as flat bands, such that the eigenvectors of $\cA_\Gamma$ are entirely supported in $V_f$, with $|V_f|=2|G|$. We call these \emph{single-cell flat bands}. The recipe is optimal in the size of $V_f$. We also have recipes to create graphs with arbitrary $d,\nu$ which have no flat bands. (Section~\ref{sec:gens}).
\item We characterize the set $\cF_\nu^s$ of single-cell flat bands and apply our criterion to compute $\cF_\nu^s$ explicitly for $\nu\le 5$. In general $\cF_\nu^s\subseteq \cF_{\nu+1}^s$, we have a description of $\cF_{\nu+1}^s\setminus\cF_\nu^s$, and $\cup_\nu \cF_\nu^s$ is the set of all totally real algebraic integers  (Section~\ref{scfb}).
\item Flat bands for $\nu=2$ are necessarily integers. They are not necessarily in $\cF_2^s$. A sufficient condition is given for the absence of flat bands (Section~\ref{sec:2ver}).
\item We answer Problem 2 for $\nu=2$. For general $\nu$, we show that the set $\mathcal{P}_{\Gamma}$ in \eqref{e:pgamma} is closed, semialgebraic and can satisfy $\dim \cP_\Gamma\ge \nu-1$ for certain $\Gamma$. A stronger form of Problem~2 was conjectured in \cite{KorSa}, we show that version to be untrue (Section~\ref{sec:destro}).
\end{enumerate}

Point (1) is essentially folklore, while point (2) is inspired by \cite{RMS}; in both cases we provide more precise and general statements with simpler proofs. Flat band generators occupy a large part of physics literature, the one in (3), though very simple, is to our knowledge new. Points (4),(5),(6) are new.

We refer to \cite{KFSH,KS,RY19} for further results, in particular concerning spectral gaps and singularities in band crossings. Flat bands have also been investigated in the special framework of Archimedean tilings of the plane in \cite{PT}. There is a total of $11$ such tilings, with explicit Floquet matrices. By analyzing each of them, the authors show that only two of them have flat bands. The question of perturbation by edge weights was later studied in \cite{KTW} and it was found that for flat bands to survive, some explicit relations between the weights must be satisfied. This is in accord with Problem~2. We also mention the recent paper \cite{DEFM} where compactly supported eigenfunctions have been studied in the Penrose and Amman-Beenker aperiodic tilings of the plane.

To conclude, let us mention that ``flat bands'' also arise in the context of universal covers of finite undirected graphs. The spectra of these trees $\cT$ has a quite similar structure to those of the present $\Z^d$-periodic graphs, though the mathematical tools to analyze them are different (Floquet theory is not applicable). It is shown in \cite[\S 4.3]{AS} that the spectrum of $\cT$ consists of bands, some of which may be degenerate, but that absolutely continuous spectrum always exists if $\min \deg \cT \ge 2$. The papers \cite{BGM,S20} characterized the flat bands, completing earlier work of \cite{Ao}, and \cite{BGM} answered a variant of Problem 2 in that setting.

\subsection{Algebraic formulation}\label{sec:alg}
The question of flat band existence can be reduced to the following. Let $h_{ij}(z) = \sum_{k\in I_{ij}} z^k$ be Laurent polynomials on $\C^d\setminus\{0\}$ for $1\le i,j\le \nu$, where $I_{ij}\subset \Z^d$ is finite and $z^k = z_1^{k_1}\cdots z_d^{k_d}$. We assume $h_{ji}(z) = h_{ij}(z^{-1})$, so $I_{ji}=-I_{ij}$. We assume $0\notin I_{ii}$ for all $i$. Some $I_{ij}$ can be empty.

Now consider the $\nu\times \nu$ matrix $A(z) = (h_{ij}(z))$. What are all possible choices of $I_{ij}$ and $\lambda\in \R$ such that the characteristic polynomial $p_\lambda(z) = \det(A(z)-\lambda)$ vanishes identically as a Laurent polynomial? This is essentially Problem 1.

For example, if $\nu=2$, this becomes a question of factorization. Namely, what are all choices of $I_{ij}$ and $\lambda$ such that $(h_{11}(z)-\lambda)(h_{22}(z)-\lambda)=h_{12}(z)h_{12}(z^{-1})$? In the special case $h_{12}(z^{-1})=h_{12}(z)$, the question is, when can we have $P_\lambda(z)Q_\lambda(z)=H(z)^2$ for Laurent polynomials whose $z^k$ coefficients are $0$ or $1$, for $k\neq 0$? An obvious choice is $P_\lambda=Q_\lambda=H$, but there are other choices (see \S~\ref{sec:2bsc}), yet we do not know all of them.

Technically, not all choices of $I_{ij}$ and $\lambda$ will answer the question of flat band existence: one also needs to assume that the periodic graph generated by $A(z)$ is connected (see Section~\ref{sec:gen} for the connection to $\Gamma$). For example, while $(z^2+z^{-2}+2)(z^4+z^{-4}+2)=(z^3+z^{-3}+z+z^{-1})^2$ is a valid factorization, the corresponding $\Gamma$ is disconnected, so this choice of $I_{ij}$ and $\lambda$ has to be disregarded.

Another concept of algebraic flavor related to periodic operators is that of \emph{irreducibility}. The question is whether the characteristic polynomial $p_\lambda(z)$ is irreducible \emph{as a Laurent polynomial}, either for fixed $\lambda$ (Fermi irreducibility) or as a polynomial in both $\lambda$ and $z$ (Bloch irreducibility). In case of Bloch irreducibility, there are no flat bands (if $\lambda=c$ was flat, then $p_\lambda(z)=(\lambda-c)q_\lambda(z)$ would be a nontrivial factorization), but irreducibility is strictly stronger. So far irreducibility has been established only for graphs with $\nu=1$ like $\Z^d$ endowed with periodic potentials \cite{FLM,Wen} and planar graphs with $\nu=2$ \cite{LS}.

\section{Generalities}\label{sec:gen}
Let $\Gamma$ be an infinite graph in some Euclidean space, which is invariant under translation by linearly independent vectors $\fa_1,\dots,\fa_d$. The vertex set $V=V(\Gamma)$ satisfies
\begin{equation}\label{e:vertex}
V = V_f + \Z_\fa^d
\end{equation}
for some finite vertex set $V_f=\{v_1,\dots,v_\nu\}$ of a motif which is copied periodically under translations by $k_\fa:= \sum_{i=1}^dk_i\fa_i\in \Z_\fa^d$, where $\Z_\fa^d=\{k_\fa:k\in \Z^d\}$.

Endow $\Gamma$ with periodic potential $Q(v_p+k_\fa) = Q(v_p)$ and edge weights $w(v_p,v_q+k_\fa) = w(v_p-k_\fa,v_q)$ for $v_q+k_\fa\sim v_p$. We consider the Schr\"odinger operator
\[
\cH = \cA_w+Q
\]
on $\Gamma$, where $(\cA_w\psi)(v_i+r_\fa) = \sum_{u\sim v_i+r_\fa} w(v_i+r_\fa,u)\psi(u)$ for $v_i+r_\fa\in V_f+\Z_\fa^d$. Let 
\begin{equation}\label{e:index}
\mathcal{I}_{ij} = \{v_j+k_\fa:v_j+k_\fa\sim v_i\}\,, \qquad I_{ij} = \{k:v_j+k_\fa\in \cI_{ij}\}\,.
\end{equation}

The $I_{ij}$ dictate the hopping (if $k\in I_{ij}$, then there is an edge from $V_f$ to its $k$-th translate, $k\in \Z^d$). Note that $0\notin I_{ii}$ since $\Gamma$ has no self-loops. We have
\[
\cN_{v_i} = \cup_{j=1}^\nu \cI_{ij}
\]
and $(\cA_w\psi)(v_i+r_\fa) = \sum_{j=1}^\nu\sum_{k\in I_{ij}} w(v_i+r_\fa,v_j+r_\fa+k_\fa)\psi(v_j+r_\fa+k_\fa)$.

\smallskip

We henceforth identify $\ell^2(\Gamma) \equiv \ell^2(\Z^d)^\nu$ via $\psi(v_p+k_\fa) \mapsto \psi_p(k)$ for $p=1,\dots,\nu$ and $k\in \Z^d$. We similarly map $Q(v_p) \mapsto Q_p$ and $w(v_p,v_j+r_\fa)\mapsto w_{pj}(r)$. Then $\cH$ is now regarded as operating on vector functions $\psi\in\ell^2(\Z^d)^\nu$ by
\begin{equation}
(\cH\psi)_i(r) = \sum_{j=1}^\nu \sum_{k\in I_{ij}} w_{ij}(k)\psi_j(r+k) + Q_i\psi_i(r)\,.
\end{equation}

For example, in Figure~\ref{fig:boxnover} (left), we have $I_{11}=I_{22}=\{\pm 1\}$, $I_{12}=\{0,\pm 1\}=I_{21}$, $w\equiv 1$, $Q\equiv 0$, and the depicted eigenvector is $\psi(0) = \binom{1}{-1}$, $\psi(k)=\binom{0}{0}$ for $k\in\Z\setminus\{0\}$.

Much like the Fourier transform is used to check that $\cA_{\Z^d}$ is unitarily equivalent to $M_f$, the multiplication operator by $f(\theta)=\sum_{j=1}^d 2\cos 2\pi \theta_j$ on $L^2(\T_\ast^d)$, where $\T_\ast^d:=[0,1)^d$, the \emph{Floquet transform} $U:\ell^2(\Z^d)^\nu \to L^2(\T_\ast^d)^\nu$ defined by
\begin{equation}\label{e:utile}
(U\psi)_p(\theta) = \sum_{k\in \Z^d} \ee^{-2\pi \ii\theta \cdot k} \psi_p(k)
\end{equation}
for $p=1,\dots,\nu$ satisfies the following.

\begin{lem}
The operator $U$ is unitary and
\begin{equation}\label{e:dir}
U\cH U^{-1} = M_H\,,
\end{equation}
where $(M_Hf)(\theta)=H(\theta)f(\theta)$ is the multiplication operator on $L^2(\T_\ast^d)^\nu$ by the matrix
\begin{align}\label{e:hthe}
H(\theta) &= (h_{ij}(\theta))_{i,j=1}^\nu\,,\\
h_{ii}(\theta) = Q_i + \sum_{k\in I_{ii}} w_{ii}(k)\ee^{2\pi\ii\theta\cdot k} & \qquad\qquad h_{ij}(\theta) := \sum_{k\in I_{ij}} w_{ij}(k)\ee^{2\pi\ii\theta\cdot k}\,, \quad i\neq j\,. \nonumber
\end{align}
Here $0\notin I_{ii}$ as $\Gamma$ has no loops. If $Q\equiv 0$ and $w\equiv 1$, we denote $A(\theta)$ instead.
\end{lem}
This is quite standard \cite{BoSa,KorSa}, though our language is simpler and more general.
\begin{proof}
Being the Fourier transform in each coordinate, $U$ is clearly unitary. The pre-image of $f(\theta)$ is $\phi_p(k)= \int_{\T_\ast^d} f_p(\theta)\ee^{2\pi\ii k\cdot \theta}\,\dd\theta$. Next given $\psi\in \ell^2(\Z^d)^\nu$,
\begin{align*}
(U\cH \psi)_p(\theta) &= \sum_{k\in\Z^d} \ee^{-2\pi \ii\theta\cdot k} \Big(\sum_{j=1}^\nu\sum_{r\in I_{pj}} w_{pj}(r)\psi_j(k+r) + Q_p\psi_p(k)\Big) \\
&= \sum_{j=1}^\nu\sum_{r\in I_{pj}} w_{pj}(r)\ee^{2\pi\ii\theta\cdot r}\sum_{k\in \Z^d} \ee^{-2\pi\ii\theta\cdot (k+r)}\psi_j(k+r) + Q_p(U\psi)_p(\theta) \\
&= \sum_{j=1}^\nu h_{pj}(\theta)(U\psi)_j(\theta) = [H(\theta)(U\psi)(\theta)]_p \,. \qedhere
\end{align*}
\end{proof}

We next note that the map $\tau: \cI_{ij} \to \cI_{ji}$ given by $v_j+k_\fa\mapsto v_i-k_\fa$ is a bijection. In fact, by translation invariance, $v_j+k_\fa\sim v_i \iff v_j\sim v_i-k_\fa$. This implies the symmetries
\begin{equation}\label{e:indicessym}
 I_{ii} = - I_{ii}\,, \qquad I_{ji} = - I_{ij}
\end{equation}
as index sets. Note however that $I_{ij}$ can be distinct from $-I_{ij}$ for $i\neq j$. Assuming $w$ satisfies $w_{ji}(-k) = \overline{w_{ij}(k)}$, we get by \eqref{e:indicessym} that the matrix $H(\theta)$ is Hermitian:
\[
h_{ji}(\theta)=\overline{h_{ij}(\theta)}\,.
\]

\begin{lem}
There exist $\nu$ continuous functions $E_1(\theta)\le \dots \le E_\nu(\theta)$ on $\T_\ast^d$ consisting of the eigenvalues of $H(\theta)$.

There exists a real analytic variety $X\subset \T_\ast^d$ of dimension $\le d-1$, (so $X$ is a closed nullset), such that the eigenvalues $E_j(\theta)$ of $H(\theta)$ are analytic on $\T_\ast^d\setminus X$. Each eigenvalue has constant multiplicity on each of the finitely many connected components of $\T_\ast^d\setminus X$.
\end{lem}
If $d=1$, the eigenvalues and eigenvectors can be chosen to be analytic over all $\T_\ast$. This is Rellich's theorem, it holds more generally when $H(\theta)$ has a compact resolvent, see \cite[Th. 3.9, p. 392]{Kato}, as long as $H(\theta)$ is self-adjoint and $\theta\in \T_\ast$ is a real parameter.
\begin{proof}
The continuity of $E_j$ is shown in \cite[p.106--109]{Kato}. The argument there is for $d=1$, but holds without change for any $d$. In steps, one shows the resolvent is continuous in $\theta$, then deduces the continuity of the spectral projection near the distinct $E_j(\theta_0)$, from which one gets that $H(\theta)$ has $m$ eigenvalues $E_{j,k}(\theta)$ close to $E_j(\theta_0)$ counting multiplicity, if $E_j(\theta_0)$ has multiplicity $m$ and $\theta$ is close to $\theta_0$.

Analyticity cannot be directly deduced from the $1d$ arguments in \cite{Kato}. Still, the proof of our statement is the same as \cite[Lemmas 4.5--4.7]{Wilcox}, which consider periodic Schr\"odinger operators in $\R^3$. In our case the argument is simpler as we can define directly the corresponding set $E = \{(\theta,\lambda)\in \R^{d+1}:p(\theta,\lambda)=0\}$, where $p(\theta,\lambda) = \det(H(\theta)-\lambda I)$ is the characteristic polynomial, which is analytic, being a polynomial of the analytic entries of $H(\theta)-\lambda I$. We also fix $n:=\nu$ in Lemmas~4.6--4.7 of \cite{Wilcox}.
\end{proof}

It follows from \eqref{e:dir} and the continuity of $E_j(\theta)$ that
\[
\sigma(\cH) = \cup_{j=1}^\nu \sigma_j\,,
\]
where $\sigma_j = \ran E_j(\theta)$ are bands, as $M_H$ is equivalent to $M_D$, with $D(\theta)=\mathrm{Diag}(E_j(\theta))$. The operator $\cH$ has no singular continuous spectrum \cite[Prp. 4.5]{HN}. One can also see that from the analyticity of $E_j(\theta)$ on $\T_\ast^d\setminus X$, because all partial derivatives $\partial_{\theta_k}E_j(\theta)$ are analytic as well, hence either $E_j(\theta)$ is constant and $\sigma_j$ is reduced to a point, yielding point spectrum, or $\nabla_\theta E_j(\theta)\neq 0$ a.e. and $\sigma_j$ consists of absolutely continuous spectrum.

We say that $\lambda\in \R$ is a \emph{flat band} of $\cH$ if $\lambda$ is an eigenvalue of $\cH$.

\begin{lem}\label{lem:flatchar}
$\lambda$ is a flat band of $\cH$ iff $\lambda$ is an eigenvalue of $H(\theta)$ for all $\theta\in \T_\ast^d$.
\end{lem}
\begin{proof}
By \eqref{e:dir}, $\lambda$ is an eigenvalue of $\cH$ iff $\lambda$ is an eigenvalue of $H(\theta)$ for all $\theta$ in a set $\Omega\subset \T_\ast^d$ of positive measure, i.e. $p(\theta;\lambda):=\det(H(\theta)-\lambda I) =0$ for all $\theta\in\Omega$. As $p(\cdot;\lambda)$ is analytic on $\T_\ast^d$, 
then $p(\cdot;\lambda)\equiv 0$ on $\Omega$ iff it vanishes on all $\T_\ast^d$, see e.g. \cite[Prp. 4.3]{HN}.
\end{proof}

Define the maximal hopping range in the $i$-th coordinate by
\[
\fh_{i} = \max(k_i:k\in \cup_{p,j\le \nu} I_{pj})
\]
If $\fh_{i}=1$ for all $1\le i\le d$, we speak of \emph{nearest-neighbor hopping}. 

\begin{thm}\label{thm:loceigen}
If $\cH$ has a flat band, then it has a corresponding eigenvector $\psi$ on $\Gamma$ of compact support. We may choose $\psi_p$ to be supported inside $\times_{i=1}^d [\![-(\nu-1)\fh_{i},(\nu-1)\fh_{i}]\!]$ for each $p=1,\dots,\nu$.
\end{thm}
This result should be contrasted with the one of \cite{McKSa}, which shows that the wavefunctions corresponding to the \emph{absolutely continuous bands} are very delocalized. Theorem~\ref{thm:loceigen} appeared in \cite[Th. 3.2]{HN}. We give a different proof which is constructive and is of practical interest, as we illustrate on some graphs. This also yields the bound on the support, which seems new but is probably not sharp. The proof is based on two results:

\begin{lem}\label{lem:cayleyham}
If $\cH$ has a flat band, then $H(\theta)$ has a corresponding eigenvector $f(\theta)$ whose entries $f_p(\theta)$ are trigonometric polynomials of degree at most $(\nu-1)\fh_{i}$ in each $\theta_i$.
\end{lem}
A statement similar to the first part, with a different derivation, appears in \cite{RY19}.
\begin{proof}
We may assume the flat band is $\lambda_0=0$, by translating $\cH$ by $\lambda_0I$. Say $0$ has multiplicity $m$ for $H(\theta)$. Then $p_\theta(\lambda) := \det(H(\theta)-\lambda I) = \lambda^mq_\theta(\lambda)$, for some polynomial $q_\theta(\lambda)$ of degree $\nu-m$ in $\lambda$. Note that $q_\theta(H(\theta))\neq 0$ since $0\neq q_\theta(0)\in \sigma(q_\theta(H(\theta)))$.

By the Cayley-Hamilton theorem, $H(\theta)^mq_\theta(H(\theta))=0$, so there exists $0\le n\le m-1$ such that $H(\theta)^nq_\theta(H(\theta))\neq 0$ and $H(\theta)^{n+1}q_\theta(H(\theta))=0$. Hence, there is a nonzero column $f(\theta)$ in $H(\theta)^nq_\theta(H(\theta))$ lying in the kernel of $H(\theta)$. 

Finally, if $p_\theta(\lambda) = \sum_{r=0}^\nu c_r(\theta) \lambda^r$, then $c_r(\theta)=0$ for $r<m$ and $q_\theta(\lambda) = \sum_{r=m}^\nu c_r(\theta) \lambda^{r-m}$. Here $c_r(\theta)$ consists of minors of $H(\theta)$ of size $\nu-r$. Since each entry of $H(\theta)$ is a trigonometric polynomial of degree at most $\fh_{i}$ in $\theta_i$, then each entry of $H(\theta)^{n}q_\theta(H(\theta))$ is of degree at most $(n+\nu-m)\fh_{i}$ in each $\theta_i$.
\end{proof}

\begin{prp}\label{prp:trigiffcomp}
If $H(\theta)$ has an eigenvector $f(\theta)$ of entries $f_p(\theta) = \sum_{m\in \Lambda_p} \alpha_p(m)\ee^{-2\pi\ii m\cdot \theta}$ for some finite $\Lambda_p\subset\Z^d$, $p=1,\dots,\nu$, then defining $\psi_p(k) = \alpha_p(k)$ if $k\in \Lambda_p$, $\psi_p(k)=0$ otherwise, $\psi$ is a localized eigenvector for $\cH$.

Conversely, if $\cH$ has an eigenvector $\psi$ of finite support $\Lambda_p$ in each entry $p=1,\dots,\nu$, and we let $f_p(\theta) = \sum_{k\in \Lambda_p} \ee^{-2\pi\ii \theta\cdot k}\psi_p(k)$, then $f(\theta)$ is an eigenvector for $H(\theta)$.

In particular, $\cH$ has an eigenvector localized in $V_f$ iff $H(\theta)$ has an eigenvector independent of $\theta$.
\end{prp}
See the graphs in Figures~\ref{fig:exnonsy}-\ref{fig:counterconj} for an illustration of this construction. 
\begin{proof}
Let $\psi_p(k)= \int_{\T_\ast^d} f_p(\theta)\ee^{2\pi\ii k\cdot \theta}\,\dd\theta$, the pre-image of $f(\theta)$ under $U$. Then $(U \cH\psi)(\theta )= H(\theta)(U\psi)(\theta) = H(\theta)f(\theta) = \lambda f(\theta) = \lambda (U\psi)(\theta)$, so $\psi$ is an eigenvector of $\cH$, since $U$ is unitary. The first claim follows since $\int_{\T_\ast^d} \ee^{2\pi\ii(k-m)\cdot \theta}\,\dd\theta=\delta_{k,m}$.

Next, for $\psi$ in the converse statement, \eqref{e:utile} becomes $(U\psi)_p(\theta) = \sum_{k\in \Z^d} \ee^{-2\pi \ii\theta\cdot k}\psi_p(k) = f_p(\theta)$. Moreover, $(U\psi)(\theta)$ is an eigenvector of $H(\theta)$ since $H(\theta)(U\psi)(\theta) = (U\cH\psi)(\theta) = \lambda (U\psi)(\theta)$, where we used $\cH\psi=\lambda\psi$.

The special case follows by taking $\Lambda_p=\{0\}$ for each $p$.
\end{proof}

\begin{proof}[Proof of Theorem~\ref{thm:loceigen}]
This now follows from Lemma~\ref{lem:cayleyham} and Proposition~\ref{prp:trigiffcomp}.
\end{proof}

Note that if $\Gamma$ has an eigenvector of compact support, then it can be translated under the action of $\Z^d$ to produce infinitely many such vectors, which are mutually orthogonal. So a flat band has infinite multiplicity for $\cH$.

\begin{thm}\label{thm:notflat}
If the weight $w$ is symmetric, $w_{ji}(-k)=w_{ij}(k)$, and strictly positive, then the top band of $\cH=\cA_w+Q$ is never flat; it consists of absolutely continuous spectrum.
\end{thm}
To our knowledge, this result was only proved rigorously in \cite[Th. 2.1]{KorSa2} for $w(u,v) = \frac{1}{\sqrt{\deg(u)\deg(v)}}$. We give a much simpler argument which works for arbitrary $w>0$.
\begin{proof}
Suppose on the contrary that the top band is an eigenvalue $\lambda$. 

\emph{Step 1.} By Theorem~\ref{thm:loceigen}, we may find a corresponding eigenvector $\psi$ of finite support. Let $S=\supp \psi$, $S_1 = \{w\notin S:w\sim v,v\in S\}$ and $S_2 = \{w\notin (S\cup S_1):w\sim v,v\in S_1\}$ be the vertices at distances $1$ and $2$ from $S$. Let $B= S\cup S_1 \cup S_2$. Since $\Gamma$ is connected, for any $v,v'\in B$, we may find a path $P_{v,v'}$ connecting $v$ to $v'$. We take $P_{v,v'}$ to be the shortest one, and fix one if there are many. Let $C = B \cup \{P_{v,v'}\}_{v,v'\in B}$. Then the finite graph $C$ is connected. Let $\widetilde{\psi}$ be the restriction of $\psi$ to $C$. Then $0=(\cH-\lambda)\psi(v)  = (H_C-\lambda)\widetilde{\psi}(v)$ if $v\in S\cup S_1$ and $0 = (H_C-\lambda)\widetilde{\psi}(v)$ for $v\in C\setminus (S\cup S_1)$ by definition of $S,S_1$.

\emph{Step 2.} We may assume $Q$ has positive entries up to shifting $\cH$ by $cI$. Then the matrix $H_C$ is irreducible since $C$ is connected, so it has a Perron-Frobenius eigenvector $\phi$ with $\phi(v)>0$ on $C$. Here $\lambda$ is the top eigenvalue of $H_C$ since $\lambda\le \|H_C\|\le \|\cH\| = \lambda$, so $\lambda=\|H_C\|$ (indeed $\|H_C\|=\sup\limits_{\|f\|_{\ell^2(C)}=1}\langle f,H_C f\rangle \le \sup\limits_{\|f\|_{\ell^2(\Gamma)}=1} \langle f,\cH f\rangle=\|\cH\|$). Since $\lambda$ is simple for $H_C$ and $(H_C-\lambda)\widetilde{\psi}=0$, we get $\widetilde{\psi}=\phi$.

We showed that $\psi(v)$ has only nonnegative entries. Finally, take $v\in S_1$. Then $0=\lambda\psi(v) = \cH\psi(v)$, implying that $\psi(w)=0$ for all $w\sim v$. Progressively, as $\Gamma$ is connected, we deduce that $\psi\equiv 0$, a contradiction.
\end{proof}

Theorem~\ref{thm:notflat} is no longer true if the edges of $\Gamma$ are allowed to carry Hermitian weights: the spectrum can be entirely flat in that case, as the example taken from \cite{Creu} in Figure~\ref{fig:creu} shows. The complex weights are interpreted as magnetic fields.

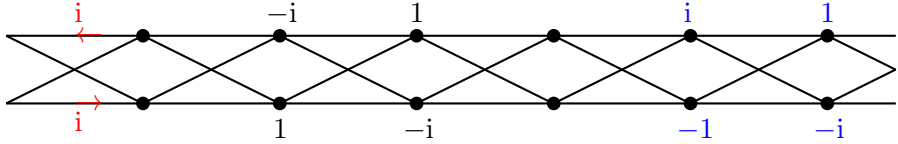
\begin{figure}[h!]
\begin{center}
\setlength{\unitlength}{0.9cm}
\thicklines
\begin{picture}(1.3,1.3)(-1.1,-1.15)
   \put(-7,0){\line(1,0){13}}
	 \put(-7,-1){\line(1,0){13}}
	 \put(-7,-1){\line(2,1){2}}
	 \put(-7,0){\line(2,-1){2}}
	 \put(-5,-1){\line(2,1){2}}
	 \put(-5,0){\line(2,-1){2}}
	 \put(-3,-1){\line(2,1){2}}
	 \put(-3,0){\line(2,-1){2}}
	 \put(-1,-1){\line(2,1){2}}
	 \put(-1,0){\line(2,-1){2}}
	  \put(1,-1){\line(2,1){2}}
	 \put(1,0){\line(2,-1){2}}
	 \put(3,-1){\line(2,1){2}}
	 \put(3,0){\line(2,-1){2}}
	 \put(5,-1){\line(2,1){1}}
	 \put(5,0){\line(2,-1){1}}
	 \put(-3,-1){\circle*{.2}}
	 \put(-3,0){\circle*{.2}}
	 \put(-5,-1){\circle*{.2}}
	 \put(-5,0){\circle*{.2}}
	 \put(-1,-1){\circle*{.2}}
	 \put(-1,0){\circle*{.2}}
	 \put(1,-1){\circle*{.2}}
	 \put(1,0){\circle*{.2}}
	 \put(3,-1){\circle*{.2}}
	 \put(3,0){\circle*{.2}}
	 \put(5,-1){\circle*{.2}}
	 \put(5,0){\circle*{.2}}
	 \put(-6,-1.1){\textcolor{red}{$\rightarrow$}}
	 \put(-6,-0.1){\textcolor{red}{$\leftarrow$}}
	 \put(-6,-1.4){\textcolor{red}{$\ii$}}
	 \put(-6,0.2){\textcolor{red}{$\ii$}}
	 \put(-1.1,0.2){$1$}
	 \put(-1.2,-1.5){$-\ii$}
	 \put(-3.2,0.2){$-\ii$}
	 \put(-3.1,-1.5){$1$}
	 \put(4.9,0.2){\textcolor{blue}{$1$}}
	 \put(4.8,-1.5){\textcolor{blue}{$-\ii$}}
	 \put(2.9,0.2){\textcolor{blue}{$\ii$}}
	 \put(2.8,-1.5){\textcolor{blue}{$-1$}}
\end{picture}
\caption{Each horizontal edge carries the symmetric weight $\pm\ii$, with the sign as shown (in red), while the diagonal edges carry the weight $1$. Two localized eigenfunctions are shown in black and blue, corresponding to $\lambda=2$ and $\lambda=-2$, respectively. Here $\sigma(\cH)=\{\pm 2\}$.}\label{fig:creu}
\end{center}
\end{figure}

Let $p(\theta;\lambda) = \det(H(\theta)-\lambda I)$. We see from \eqref{e:hthe} that if $z_j:=\ee^{2\pi\ii\theta_j}$, then $p(z;\lambda)$ is a Laurent polynomial in $z$ and polynomial in $\lambda$. If $\lambda$ is a flat band, then $p(\cdot;\lambda)$ vanishes identically on the sphere $\mathbb{S}^d = \{(z_1,\dots,z_d)\in \C^d:|z_1|=|z_2|=\dots=|z_d|=1\}$ by Lemma~\ref{lem:flatchar}. This implies the following fact that we frequently use.

\begin{lem}\label{lem:triv}
If $\lambda_0$ is a flat band for $\cH$, then for this $\lambda_0$, the characteristic polynomial $p(z;\lambda_0)$ is the zero Laurent polynomial in $z$. 
\end{lem}
In particular, each coefficient of $z^k$ must vanish in this case.
\begin{proof}
It suffices to show that if $p(z)$ is a Laurent polynomial vanishing identically on $\mathbb{S}^d$, then $p(z)\equiv 0$ on $\C^d\setminus\{0\}$. Take $Q\in \N^d$ large, so that $q(z) = z^Qp(z)$ is a polynomial in $z$. We show that $q(z)$ vanishes for all $z\in \C^d$.

If $d=1$, then $q(z)$ is identically zero on $\C$, as $q$ is holomorphic and $q\equiv 0$ on $\mathbb{S}$.

Suppose the claim holds in dimension $d-1$ and fix $z_d\in \mathbb{S}$. Then $(z_1,\dots,z_{d-1})\mapsto q(z_1,\dots,z_d)$ is a polynomial in $d-1$ variables vanishing on $\mathbb{S}^{d-1}$, so it vanishes for all $(z_1,\dots,z_{d-1})\in \C^{d-1}$. Say $q(z) = \sum_{k\in\Lambda} \alpha_k z^k$ for some $\Lambda\subset \N^d$. Since $q(z) = \sum\limits_{k_1,\dots,k_{d-1}} \beta_{(k_1,\dots,k_{d-1})}(z_d) z_1^{k_1}\cdots z_{d-1}^{k_{d-1}}$ is trivial, where $\beta_{(k_1,\dots,k_{d-1})}(z_d) = \sum\limits_{k_d:(k_1,\dots,k_d)\in \Lambda}\alpha_k z_d^{k_d}$, then $\beta_{(k_1,\dots,k_{d-1})}(z_d)=0$ for each $(k_1,\dots,k_{d-1})$. As $z_d\in \mathbb{S}$ is arbitrary, this holds for all $z_d\in\mathbb{S}$. But each such $\beta$ is a polynomial in $z_d$. Since $\beta$ vanishes on $\mathbb{S}$, then it vanishes on $\C$, so each coefficient $\alpha_k$ vanishes. Thus, $q(z)\equiv 0$ on $\C^d$.
\end{proof}

\begin{rem}
If $\lambda$ is an eigenvalue of $H(\theta)$, then the characteristic polynomial $p(z;\lambda)$ must be symmetric in $z^{\pm k}$, because $p(z;\lambda)=0=\overline{p(z;\lambda)}$ on $\mathbb{S}^d$.
\end{rem}

Lemma~\ref{lem:triv} is the set of ``necessary and sufficient conditions'' given in \cite{TA} for $\cH$ to have a flat band. The authors compute the coefficients of $z^k$ in $p(z;\lambda_0)$ explicitly when ($d=\fh_{1}=1,\nu\le 4$), or ($d\in \{2,3\},\fh_{i}=1,\nu\le 3$), or ($d=1,\fh_{1}\in \{2,3\},\nu\le 3$) and give examples for a range of parameters that make these coefficients vanish, so that a flat band appears. An example is also given for $d=\fh_{1}=1,\nu=5$.

\section{Generators}\label{sec:gens}

The aim of \S\ref{sec:genfla} and \S\ref{sec:gennof} is to generate graphs with and without flat bands, respectively.

\subsection{Flat band generators}\label{sec:genfla} 
We are interested in $Q\equiv 0$ and $w\equiv 1$ here, so $\cH=\cA_\Gamma$.

\begin{thm}\label{thm:eigenchar}
Any eigenvalue of a finite graph $G_F$ can be made to appear as a flat band of a periodic graph $\Gamma$ with $\nu=2\left|G_F\right|$, with eigenvector $\psi$ supported in $V_f$. This choice is optimal, in the sense that there exists $G_F$ and $\lambda\in\sigma(G_F)$ such that $\lambda$ does not arise in any periodic $\Gamma$ having both $\nu<2\left|G_F\right|$ and $\psi$ supported in $V_f$.
\end{thm}

The eigenvector is explicit in the proof. Of course, this theorem is interesting because we consider $\cA_{\Gamma}$. If we allow potentials, then trivially any $\lambda\in \R$ can be made flat by shifting a flat band $0$ with $Q=\lambda I$. Before we give the proof, let us mention a consequence.

\begin{cor}
A flat band is necessarily a totally real algebraic integer. Conversely, any totally real algebraic integer can be made to appear as a flat band of some periodic $\Gamma$.
\end{cor}
Recall that a totally real algebraic integer is a root of a monic polynomial with integer coefficients and real roots.
\begin{proof}
The first part follows from Theorem~\ref{thm:loceigen}, as an eigenvector of compact support is also an eigenvector of a finite subgraph. The more interesting converse follows from Theorem~\ref{thm:eigenchar} and \cite{Sal}.
\end{proof}

\begin{proof}[Proof of Theorem~\ref{thm:eigenchar}]
We construct $\Gamma$ as shown in Figure~\ref{fig:cartgf}. Each vertex of the graph of Figure~\ref{fig:boxnover} (right) is replaced by a copy of $G_F$. The edges of this new $\Gamma$ consist of the old edges of each $G_F$, plus an edge from each vertex $v_p\in G_F$ to its copy $v_p'$ in each neighboring $G_F$. So $\deg_{\Gamma}(v_p)=\deg_{G_F}(v_p)+4$.
\begin{figure}[h!]
\begin{center}
\setlength{\unitlength}{1cm}
\thicklines
\begin{picture}(1.3,1.3)(-0.8,-1.1)
   \put(-6,0){\line(1,0){1.7}}
	 \put(-6,-1){\line(1,0){1.7}}
	 \put(-3.6,0){\line(1,0){1.4}}
	 \put(-3.6,-1){\line(1,0){1.4}}
	 \put(-1.6,0){\line(1,0){1.4}}
	 \put(-1.6,-1){\line(1,0){1.4}}
	 \put(0.4,0){\line(1,0){1.4}}
	 \put(0.4,-1){\line(1,0){1.4}}
	  \put(2.4,0){\line(1,0){1.4}}
	 \put(2.4,-1){\line(1,0){1.4}}
	 \put(-3.8,-0.9){\line(2,1){1.6}}
	 \put(-3.7,-0.15){\line(2,-1){1.5}}
	 \put(-1.8,-0.9){\line(2,1){1.6}}
	 \put(-1.7,-0.15){\line(2,-1){1.5}}
	 \put(0.2,-0.9){\line(2,1){1.6}}
	 \put(0.3,-0.15){\line(2,-1){1.5}}
	 \put(2.2,-0.9){\line(2,1){1.6}}
	 \put(2.3,-0.15){\line(2,-1){1.5}}
	 \put(-0.2,-1.1){$G_F$}
	 \put(-0.2,-0.1){$G_F$}
	 \put(-2.2,-1.1){$G_F$}
	 \put(-2.2,-0.1){$G_F$}
	  \put(-4.2,-0.1){$G_F$}
	   \put(1.8,-0.1){$G_F$}
	    \put(-4.2,-1.1){$G_F$}
	     \put(1.8,-1.1){$G_F$}
	 \put(-4.2,-0.9){\line(-2,1){1.6}}
	 \put(-4.2,-0.1){\line(-2,-1){1.6}}
\end{picture}
\caption{The Cartesian product of the graph in Figure~\ref{fig:choicevf} and $G_F$.}\label{fig:cartgf}
\end{center}
\end{figure}
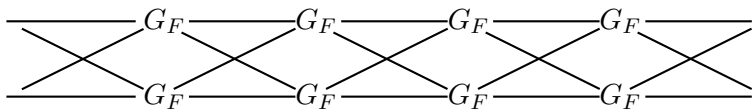
More precisely, $\Gamma$ is the \emph{Cartesian product} $\Gamma=\Gamma_0\mathop\square G_F$, where $\Gamma_0$ is the infinite graph of Figure~\ref{fig:boxnover}, so $\cA_\Gamma=\cA_{\Gamma_0}\mathop\otimes I+I\mathop\otimes \cA_{G_F}$. If $\lambda\in \sigma(\cA_{G_F})$ with eigenvector $\phi$, and if $f(0)=\binom{1}{-1}$, $f(k)=\binom{0}{0}$, $k\neq 0$, is the eigenvector of $\cA_{\Gamma_0}$ corresponding to $0$, let $\psi = f\otimes \phi$, i.e. $\psi_i(0)=\phi_i$ on the upper $G_F$, $\psi_i(0)=-\phi_i$ on the lower $G_F$ and $\psi(k)=0$ for $k\neq 0$. Then $\cA_\Gamma\psi = (\cA_{\Gamma_0}f)\otimes \phi + f\otimes(\cA_{G_F} \phi) = 0 + \lambda f\otimes \phi = \lambda \psi$.

Optimality follows from the examples in Section~\ref{scfb} below (which show much more). In fact, the eigenvalues $\pm \sqrt{2}$ of the path on $3$ vertices do not appear as single-cell flat bands for any $\Gamma$ with $\nu<6$, as we see from the explicit lists.
\end{proof}

Though this generator is optimal for general $G_F$, it can be improved for regular $G_F$.

\begin{prp}\label{prp:reggen}
All eigenvalues of a regular connected graph $G_F$ except the top one can be generated in some $\Gamma$ with $\nu=|G_F|+1$, with eigenvectors localized on $G_F\subset V_f$.
\end{prp}
\begin{proof}
We add a new vertex $o$ and attach it to each of the vertices of $G_F$, then $\Z$-periodize along the vertices $o$. See Figure~\ref{fig:deco}. Then $A(\theta) = \begin{pmatrix} 2\cos 2\pi\theta& 1& \cdots& 1\\ 1\\ \vdots& & \cA_{G_F} \\ 1&&\end{pmatrix}$. 
\begin{figure}[h!]
\begin{center}
\setlength{\unitlength}{0.9cm}
\thicklines
\begin{picture}(1.3,1.3)(-0.7,-0.1)
   \put(-7,0){\line(1,0){6}}
   \put(-6.2,1){$G_F$}
   \put(-6,0){\line(0,1){0.8}}
   \put(-6,0){\line(1,5){0.17}}
   \put(-6,0){\line(-1,5){0.17}}
   \put(-4.2,1){$G_F$}
   \put(-4,0){\line(0,1){0.8}}
   \put(-4,0){\line(1,5){0.17}}
   \put(-4,0){\line(-1,5){0.17}}
   \put(-2.2,1){$G_F$}
   \put(-2,0){\line(0,1){0.8}}
   \put(-2,0){\line(1,5){0.17}}
   \put(-2,0){\line(-1,5){0.17}}
   \put(-6,0){\circle*{.2}}
   \put(-4,0){\circle*{.2}}
   \put(-2,0){\circle*{.2}}
   \put(-6.1,-0.35){\small{$o$}}
    \put(-4.1,-0.35){\small{$o$}}
     \put(-2.1,-0.35){\small{$o$}}
   \put(1,0){\line(1,0){6}}
   \put(5,0){\line(1,1){1}}
   \put(5,0){\line(-1,1){1}}
   \put(2,0){\line(1,1){1}}
   \put(2,0){\line(-1,1){1}}
	 \put(4,1){\line(1,0){2}}
	 \put(1,1){\line(1,0){2}}
	 \put(5,0){\circle*{.2}}
	 \put(2,0){\circle*{.2}}
	 \put(6,1){\circle*{.2}}
	 \put(4,1){\circle*{.2}}
	 \put(3,1){\circle*{.2}}
	 \put(1,1){\circle*{.2}}
	 \put(3.4,1){\small{$-1$}}
	 \put(6.2,1){\small{$1$}}
	 \put(4.9,-0.4){\small{$0$}}
\end{picture}
\caption{General procedure (left), example $G_F=P_2$, $\lambda=-1$ (right).}\label{fig:deco}
\end{center}
\end{figure}
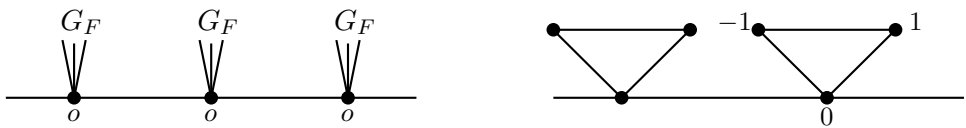 
As $G_F$ is regular, its top eigenvector is $f_1=(1,\dots,1)$, the remaining ones $f_i$, $1<i\le |G_F|$ corresponding to $\lambda_i$ are orthogonal to it, i.e. satisfy $\sum_{j=1}^{|G_F|} f_{ij}=0$. So the vector $\widetilde{f}_i = (0,f_i)^\intercal$ satisfies $A(\theta) \widetilde{f}_i = \lambda_i \widetilde{f}_i$ for each $i=2,\dots,|G_F|$ and all $\theta$.
\end{proof}

\begin{rem}
In the proof of Theorem~\ref{thm:eigenchar}, we can of course generalize the construction by replacing $\Gamma_0$ by any periodic graph having the flat band zero.
In particular, we can take the Lieb lattice or the graph in Figure~\ref{fig:curious} (right) with the depicted eigenvector $f$, and if $\phi$ is an eigenvector of $G_F$ for $\lambda_0$, then $\psi = f\otimes \phi$ will correspond to the flat band $\lambda_0$ but $\psi$ is not supported on a single-cell. Generating such $\psi$ seems to be of interest to physicists. Similarly, long-range hopping can be included by changing $\Gamma_0$.
\end{rem}

\subsection{No flat bands}\label{sec:gennof}
Starting from a finite $G_F$ and any $\Gamma_0$ with $\nu=1$, like $\Z^d$ or the triangular lattice, we can construct infinitely many new graphs $\Gamma$ which have no flat bands by considering the Cartesian product $\Gamma = \Gamma_0\mathop\square G_F$, as was shown in \cite[\S 3.2]{McKSa}.  
\begin{figure}[h!]
\begin{center}
\setlength{\unitlength}{0.9cm}
\thicklines
\begin{picture}(1.3,1.4)(-0.8,-0.8)
   \put(-7.2,0.3){\line(1,0){6.4}}
	 \put(-7.2,-0.7){\line(1,0){6.4}}
	 \put(-5,-0.7){\line(0,1){1}}
	 \put(-3,-0.7){\line(0,1){1}}
	 \put(-1,-0.7){\line(0,1){1}}
	 \put(-7,-0.7){\line(0,1){1}}
	 \put(-3,-0.7){\circle*{.2}}
	 \put(-3,0.3){\circle*{.2}}
	 \put(-5,-0.7){\circle*{.2}}
	 \put(-5,0.3){\circle*{.2}}
	 \put(1,0){\line(1,0){6}}
	 \put(1,-1){\line(1,0){6}}
	 \put(2.5,-1){\line(0,1){1}}
	 \put(4.5,-1){\line(0,1){1}}
	 \put(6.5,-1){\line(0,1){1}}
	 \put(4.5,-1){\circle*{.2}}
	 \put(4.5,0){\circle*{.2}}
	 \put(2.5,-1){\circle*{.2}}
	 \put(2.5,0){\circle*{.2}}
	 \put(6.5,-1){\circle*{.2}}
	 \put(6.5,0){\circle*{.2}}
	 \put(2.5,0){\line(1,1){0.7}}
	 \put(4.5,0){\line(1,1){0.7}}
	 \put(6.5,0){\line(1,1){0.7}}
	 \put(1.5,0.7){\line(1,0){6}}
	 \multiput(1.5,-0.3)(0.3,0){20}{\line(1,0){0.1}}
	 \put(3.2,0.7){\circle*{.2}}
	 \put(5.2,0.7){\circle*{.2}}
	 \put(7.2,0.7){\circle*{.2}}
 	 \multiput(3.2,0.7)(0,-0.2){5}{\line(0,-1){0.1}}
  \multiput(5.2,0.7)(0,-0.2){5}{\line(0,-1){0.1}}
	 \multiput(7.2,0.7)(0,-0.2){5}{\line(0,-1){0.1}}
	 \multiput(2.5,-1)(0.2,0.2){4}{\line(1,1){0.1}}
	\multiput(4.5,-1)(0.2,0.2){4}{\line(1,1){0.1}}
	\multiput(6.5,-1)(0.2,0.2){4}{\line(1,1){0.1}}
	\put(3.2,-0.3){\circle{.2}}
	\put(5.2,-0.3){\circle{.2}}
	\put(7.2,-0.3){\circle{.2}}
\end{picture}
\caption{The graphs $\Z\mathop\square P_2$ (left) and $\Z\mathop\square C_4$ (right).}\label{fig:carte}
\end{center}
\end{figure}
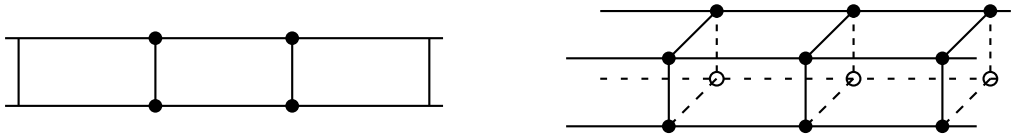

Another operation is to take the tensor product. Here, starting again from a graph $\Gamma_0$ with $\nu=1$, we consider $\Gamma = \Gamma_0\times G_F$ for finite $G_F$. However $G_F$ must now be chosen more carefully. It suffices to have $G_F$ non-bipartite with $0\notin \sigma(\cA_{G_F})$. In that case $\Gamma$ will be connected with no flat bands. See \cite[\S 3.4.2]{McKSa} for more details.

\section{Single-cell flat bands}\label{scfb}

The simplest kind of flat bands one can try to characterize are those where the corresponding eigenvector is entirely supported on a single cell, cf. \cite{FLBMD}. We may assume this cell to be $V_f$ by shifting $V_f$ if necessary. We will use a quite abusive terminology and describe these as \emph{single-cell flat bands}. The set of all possible single-cell flat bands of $\cA_\Gamma$ that can appear in connected periodic graphs $\Gamma$ with $|V_f|=\nu$ is denoted by $\cF_\nu^s$.

\begin{rem}\label{rem:choicevf}
The notion of single-cell flat band depends on the choice of $V_f$. It is a property of $(\Gamma,V_f)$, rather than $\Gamma$. There is an obvious non-uniqueness in choosing $V_f$ by shifting to $V_f+k_\fa$. But it can be more significant, as in Figure~\ref{fig:choicevf}. Note that changing $V_f$ also changes $I_{ij}$, as one can check from the figure.
\begin{figure}[h!]
\begin{center}
\setlength{\unitlength}{0.9cm}
\thicklines
\begin{picture}(1.3,1.3)(-0.8,-1.1)
   \put(-7.2,0){\line(1,0){6.4}}
	 \put(-7.2,-1){\line(1,0){6.4}}
	 \put(-7,-1){\line(2,1){2}}
	 \put(-7,0){\line(2,-1){2}}
	 \put(-5,-1){\line(2,1){2}}
	 \put(-5,0){\line(2,-1){2}}
	 \put(-3,-1){\line(2,1){2}}
	 \put(-3,0){\line(2,-1){2}}
	 \put(-3,-1){\textcolor{red}{\circle*{.2}}}
	 \put(-3,0){\circle*{.2}}
	 \put(-5,-1){\circle*{.2}}
	 \put(-5,0){\textcolor{red}{\circle*{.2}}}
	 \put(-3.1,0.2){\small{$1$}}
	 \put(-3.2,-1.5){\small{$-1$}}
	 \put(-1,-1){\line(2,1){0.3}}
	 \put(-1,0){\line(2,-1){0.3}}
	 \put(-7,-1){\line(-2,1){0.3}}
	 \put(-7,0){\line(-2,-1){0.3}}
	 \put(1,0){\line(1,0){6}}
	 \put(1,-1){\line(1,0){6}}
	 \put(1,-1){\line(2,1){2}}
	 \put(1,0){\line(2,-1){2}}
	 \put(3,-1){\line(2,1){2}}
	 \put(3,0){\line(2,-1){2}}
	 \put(5,-1){\line(2,1){2}}
	 \put(5,0){\line(2,-1){2}}
	 \put(5,-1){\textcolor{red}{\circle*{.2}}}
	 \put(5,0){\textcolor{red}{\circle*{.2}}}
	 \put(3,-1){\circle*{.2}}
	 \put(3,0){\circle*{.2}}
	 \put(4.9,0.2){\small{$1$}}
	 \put(4.8,-1.5){\small{$-1$}}
\end{picture}
\caption{Two choices for $V_f$ are shown in red. The ``obvious'' choice (right) has a single cell flat band. The one on the left does not: any eigenvector of $\Gamma$ will live on more than one copy of $V_f$.}\label{fig:choicevf}
\end{center}
\end{figure}
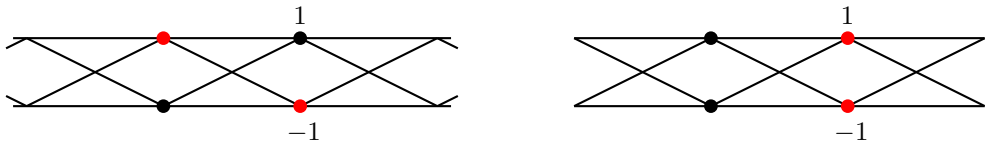
Despite this ``non-canonicity'', this notion is very useful to obtain some results, as we will see in this section.

We note that one can always choose $V_f$ to be a tree (in particular, connected). The argument can be found in \cite[Lem. 3.1]{KorSa2}. For example, in this convention, one chooses the $V_f$ on the left of Figure~\ref{fig:choicevf}. In general, the ``natural'' $V_f$ will not be a tree; one obtains a tree by removing edges, passing to a spanning tree. We will not follow this convention here: our $V_f$ is arbitrary, and can even be disconnected as in Figure~\ref{fig:choicevf} (right).
\end{rem}

Motivated by Figure~\ref{fig:boxnover}, we study the following neighborhood condition.

\begin{lem}\label{lem:equi2}
For any $\nu$ and $d$, we have
\[
\cN_{v_i}\setminus \{v_j\} = \cN_{v_j}\setminus \{v_i\} \iff I_{ii}=I_{jj}=I_{ji}\setminus\{0\}=I_{ij}\setminus\{0\} \ \text{and}\ I_{ir}=I_{jr} \ \forall r\neq i,j.
\]
This condition also implies that $\deg(v_i)=\deg(v_j)$.
\end{lem}
\begin{proof}
Since $\cN_{v_i}\setminus \{v_j\} = (\cI_{ij}\setminus \{v_j\})\cup (\cup_{k\neq j} \cI_{ik})$, with $\cI_{rs}$ of the form $\{v_s+n_\fa\}$, we see that equality holds iff $\cI_{ij}\setminus \{v_j\} = \cI_{jj}$, $\cI_{ji}\setminus \{v_i\} = \cI_{ii}$ and $\cI_{ik}=\cI_{jk}$ for $k\neq i,j$. This holds iff the statements holds for the $I_{rs}$, and we use \eqref{e:indicessym} to conclude.
\end{proof}

This simple geometric condition always generates the single-cell flat band $0$ or $-1$:

\begin{prp}\label{prp:neighgenflat}
For any $\nu$ and $d$, if $\cN_{v_i}\setminus\{v_j\}=\cN_{v_j}\setminus\{v_i\}$ for some $i\neq j$, then $\cA_\Gamma$ has a single-cell flat band and $I_{ij}$ is symmetric ($-I_{ij}=I_{ij}$).

This flat band is $\lambda=0$ if $v_i\not\sim v_j$ and $-1$ if $v_i\sim v_j$. The corresponding eigenvector localized on $V_f$ can be chosen to have values $+1$ on $v_i$, $-1$ on $v_j$ and zero elsewhere.
\end{prp}
\begin{proof}
Up to rearranging the indices, we may assume $i=1$ and $j=2$. Then \eqref{e:hthe} takes the form $A(\theta) = \begin{pmatrix} h_{11}(\theta) & h_{11}(\theta)+\epsilon& g_{\theta}\\ h_{11}(\theta)+\epsilon& h_{11}(\theta)& g_{\theta}\\ g_\theta^\ast& g_\theta^\ast &H'(\theta)\end{pmatrix}$ by Lemma~\ref{lem:equi2} for $g_{\theta}=(h_{13}(\theta),\dots,h_{1\nu}(\theta))$ and some $(\nu-2)\times (\nu-2)$ matrix $H'(\theta)$. Here $\epsilon=0$ if $v_i\not\sim v_j$ and $\epsilon=1$ if $v_i\sim v_j$.

Clearly, if $f = (1,-1,0,\dots,0)^\intercal$, then $H(\theta) f = -\epsilon f$. As this holds for all $\theta$, with $f$ independent of $\theta$, we have shown that $\Gamma$ has a single-cell flat band $-\epsilon$.

The set $I_{ij}$ is symmetric since $I_{ij}\setminus \{0\}=I_{ii}$ is symmetric.
\end{proof}

Figure~\ref{fig:boxnover} and Figure~\ref{fig:nonsymvar} illustrate Proposition~\ref{prp:neighgenflat}.

\begin{figure}[h!]
\begin{center}
\setlength{\unitlength}{0.8cm}
\thicklines
\begin{picture}(2,2)(-1,-2)
	 \put(-7,-1){\line(1,0){6}}
	 \put(-6,-1){\line(0,1){1}}
	 \put(-4,-1){\line(0,1){1}}
	 \put(-2,-1){\line(0,1){1}}
	 \put(-6,-2){\line(0,1){1}}
	 \put(-4,-2){\line(0,1){1}}
	 \put(-2,-2){\line(0,1){1}}
	 \put(-4,-1){\circle*{.2}}
	 \put(-4,0){\circle*{.2}}
	 \put(-6,-1){\circle*{.2}}
	 \put(-6,0){\circle*{.2}}
	 \put(-2,-1){\circle*{.2}}
 	 \put(-6,-2){\circle*{.2}}
	 \put(-4,-2){\circle*{.2}}
	 \put(-2,-2){\circle*{.2}}
	 \put(-2,0){\circle*{.2}}
	 \put(-3.8,-0.1){\small{$-1$}}
	 \put(-3.8,-0.9){\small{$0$}}
	 \put(-3.8,-2){\small{$1$}}
	 \put(1,-1){\line(1,0){6}}
	 \put(2,-1){\line(0,1){1}}
	 \put(4,-1){\line(0,1){1}}
	 \put(6,-1){\line(0,1){1}}
	 \put(2,-2){\line(0,1){1}}
	 \put(4,-2){\line(0,1){1}}
	 \put(6,-2){\line(0,1){1}}
	 \put(4,-1){\circle*{.2}}
	 \put(4,0){\circle*{.2}}
	 \put(2,-1){\circle*{.2}}
	 \put(2,0){\circle*{.2}}
	 \put(6,-1){\circle*{.2}}
 	 \put(2,-2){\circle*{.2}}
	 \put(4,-2){\circle*{.2}}
	 \put(6,-2){\circle*{.2}}
	 \put(6,0){\circle*{.2}}
	 \put(4.2,-0.1){\small{$-1$}}
	 \put(4.2,-0.7){\small{$0$}}
	 \put(4.2,-2){\small{$1$}}
	 \put(4,0){\line(-2,-1){2}}
	 \put(4,-2){\line(-2,1){2}}
	 \put(6,0){\line(-2,-1){2}}
	 \put(6,-2){\line(-2,1){2}}
	 \put(2,0){\line(-2,-1){1}}
	 \put(2,-2){\line(-2,1){1}}
	 \put(6,-1){\line(2,1){1}}
	 \put(6,-1){\line(2,-1){1}}
\end{picture}
\caption{Graphs with flat band $\lambda=0$, with all $I_{rs}$ symmetric (left) and $I_{13},I_{23}$ non-symmetric (right).}\label{fig:nonsymvar}
\end{center}
\end{figure}
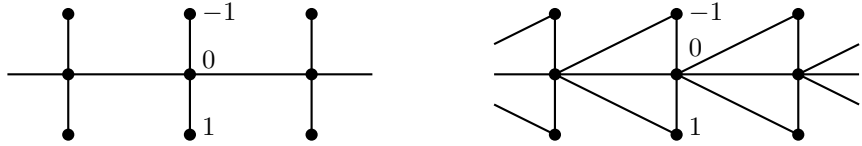

\begin{prp}\label{prp:nu2single}
For $\nu=2$, $\cA_\Gamma$ has a single-cell flat band iff $\cN_{v_1}\setminus \{v_2\} = \cN_{v_2}\setminus\{v_1\}$. Moreover, $\cF_2^s=\{0,-1\}$. 
\end{prp}
\begin{proof}
We already showed the converse in Proposition~\ref{prp:neighgenflat}. So suppose that $\Gamma$ has an eigenvector supported in a single cell. By Proposition~\ref{prp:trigiffcomp}, this means there are some $a,b$ independent of $\theta$ such that $\begin{pmatrix} h_{11}(\theta)&h_{12}(\theta)\\\overline{h_{12}(\theta)}&h_{22}(\theta)\end{pmatrix}\begin{pmatrix}a\\b\end{pmatrix}=\lambda\begin{pmatrix}a\\b\end{pmatrix}$. Let $z_j=\ee^{2\pi\ii\theta_j}$. Then as Laurent polynomials,
\begin{equation}\label{e:2eqs}
 a\sum_{k\in I_{11}} z^k+ b\sum_{p\in I_{12}} z^p = a\lambda \qquad \text{and} \qquad a\sum_{p\in I_{12}} z^{-p} +b\sum_{k'\in I_{22}} z^{k'} = b\lambda \,.
\end{equation}

Recall that $0\notin I_{ii}$. Equating the coefficients of $z^0$ on both sides of each equation, we deduce that $b\epsilon=a\lambda$, $a\epsilon=b\lambda$, where $\epsilon=0$ if $v_1\not\sim v_2$ and $\epsilon=1$ if $v_1\sim v_2$.

If $I_{12}\subseteq\{0\}$, then \eqref{e:2eqs} implies $I_{11}=I_{22}=\emptyset$. So $A(\theta)=\begin{pmatrix} 0&\epsilon\\ \epsilon&0\end{pmatrix}$ induces a disconnected $\Gamma$. Hence, $I_{12}$ must contain nonzero terms.

Then \eqref{e:2eqs} implies ($I_{12}\setminus\{0\}=I_{11}$ and $b=-a$) and $(I_{12}\setminus \{0\} = -I_{22}$), respectively, by looking at the coefficients of $z^p$. This implies $\lambda=-\epsilon$, and $I_{11}=I_{22}=I_{12}\setminus\{0\}$, where we used \eqref{e:indicessym}. Hence, $\cN_{v_1}\setminus \{v_2\} = \cN_{v_2}\setminus\{v_1\}$ by Lemma~\ref{lem:equi2}.
\end{proof}

We next give a simple but useful observation.

\begin{lem}\label{lem:gensincell}
Fix $(\Gamma,V_f)$. If $\lambda$ is a single-cell flat band for $\cA_\Gamma$, then $\lambda\in \sigma(\cA_{V_f})$. If $V_f$ is connected, then its top eigenvalue cannot be a single-cell flat band for $\Gamma$.
\end{lem}
\begin{proof}
If $\psi$ is an eigenvector entirely supported in $V_f$, then
\begin{equation}\label{e:vfeigen}
\lambda \psi_i(0) = (\cA_\Gamma \psi)_i(0) = \sum_{j=1}^\nu \sum_{k\in I_{ij}} \psi_j(k) = \sum_{j=1}^\nu \epsilon_{ij} \psi_j(0),
\end{equation}
because $\psi_j(k)=0$ for $k\neq 0$, and we denoted $\epsilon_{ij}=1$ if $v_i\sim v_j$ and $0$ otherwise. This means $\psi(0)$ is an eigenvector of the matrix $(\epsilon_{ij})_{i,j=1}^\nu=\cA_{V_f}$.

The second part is a Perron-Frobenius argument as in Theorem~\ref{thm:notflat}, Step 2.
\end{proof}

Before we move on to $\nu>2$, it will be helpful to have more criteria to eliminate eigenvalues of finite graphs instead of studying case by case. We have the following general result, which strengthens Lemma~\ref{lem:gensincell}.

\begin{thm}\label{thm:singlecellcriter}
Let $|V_f|=\nu$. If $\lambda$ is a single-cell flat band for $\cA_\Gamma$, with eigenvector $\psi$, then $\lambda\in \sigma(\cA_{V_f})$, and there exist $\delta_i\in \{0,1\}$, $i=1,\dots,\nu$, not all zero, such that $\sum_{i=1}^\nu \delta_i \psi_i=0$.

Conversely, if $\lambda\in \sigma(\cA_{V_f})$ and there exist $\delta_i\in \{0,1\}$, $i=1,\dots,\nu$, not all zero such that $\sum_{i=1}^\nu \delta_i \psi_i=0$, then $\lambda\in \cF_\nu^s$, as long as $V_f$ is connected.
\end{thm}
The connectedness of $V_f$ is not necessary, a weaker assumption suffices, and we will need this in Corollary~\ref{cor:increaflat}. The connectedness condition cannot be completely dropped however: there \emph{are} $\lambda\in \sigma(\cA_{V_f})$ that satisfy $\sum_{i=1}^\nu \delta_i \psi_i=0$ without being flat (take e.g. $V_f=G_4$ in Example~\ref{exa:singlenu3} below, $\lambda=1$, $\delta=(0,0,1)$).

The condition $\sum_{i=1}^\nu \delta_i\psi_i=0$ greatly generalizes some old procedures of constructing flat bands by starting from a finite connected graph $G_F$ and an eigenfunction which has a zero in $G_F$. It is clear in that case that one can periodize $G_F$ along this vanishing vertex and obtain a flat band by extending $\psi$ by zero. This corresponds to the very special case where $\delta_i=0$ for all $i$ except the vanishing vertex $i_0$, where $\delta_{i_0}=1$.
\begin{proof}
Suppose $\lambda$ is a single-cell flat band with eigenvector $\psi$. We saw in Lemma~\ref{lem:gensincell} that $\cA_{V_f}\psi=\lambda\psi$. It follows from Proposition~\ref{prp:trigiffcomp} that $\psi$ is an eigenvector of $A(\theta)$ for all $\theta$. Let $z_j = \ee^{2\pi\ii \theta_j}$. We write $A(z) = \cA_{V_f} + B(z)$, where $B(z)$ is the matrix of bridges from $V_f$ to its neighboring copies (i.e. in \eqref{e:hthe}, $h_{ij}(z) = \epsilon_{ij}+\sum_{k\in I_{ij}\setminus \{0\}} z^k$, the first term is the entry in $\cA_{V_f}$, the second in $B(z)$). 

Now $\cA(z)\psi=\lambda\psi$ for all $z$, so $(\cA_{V_f}+B(z))\psi=\lambda\psi$ for all $z$. Since $\cA_{V_f}\psi = \lambda\psi$, it follows that $B(z)\psi=0$ for all $z$.

Denote the entries of $B(z)$ by $b_{ij}(z)$. Then for all $i$, $\sum_{j=1}^\nu b_{ij}(z)\psi_j=0$. If $B(z)=0$, then $A(z)=\cA_{V_f}$ and $\Gamma$ is a disconnected direct sum of copies of $V_f$. So there exist $i,r$ such that $b_{ir}(z)\neq 0$. Say $0\neq k\in I_{ir}$. Fixing $i,k$, let $\delta_j$ be the coefficient of $z^k$ in $b_{ij}(z)$. Then $\delta_r=1$. On the other hand, since $\sum_{j=1}^\nu b_{ij}(z)\psi_j=0$, then as a Laurent polynomial, the coefficient of $z^k$ must vanish, hence $\sum_{j=1}^\nu \delta_j \psi_j=0$. This proves the first part.

Conversely, suppose that we found $\delta_j\in \{0,1\}$ not all zero such that $\sum_{j=1}^\nu \delta_j \psi_j =0$. Consider the matrix $B  = (\delta_i\delta_j)$ and define $B(z) = (z+z^{-1}) B$. Here $d=1$, i.e. we construct $\Gamma$ to be $\Z$-periodic (this can be generalized to any $d$, but is not needed). Then $B(z)\psi =0$, and $A(\theta) = \cA_{V_f}+B(\theta)$ induces a connected $\Gamma$, because $V_f$ is connected, and $B(\theta)$ gives bridges to $(V_f \pm 1_\fa)$ by our choice of taking the weight $z+z^{-1}=2\cos 2\pi\theta$ at each nonzero entry of $B$, which imply connection to nearest neighbors.
\end{proof}


\begin{cor}\label{cor:increaflat}
The sets $\cF_\nu^s$ are increasing: $\cF_\nu^s\subseteq \cF_{\nu+1}^s$ for all $\nu$.
\end{cor}
\begin{proof}
Let $\lambda\in \cF_\nu^s$, with $\Gamma$-eigenvector $\psi$ on $V_f$. By Theorem~\ref{thm:singlecellcriter}, we know there are $\delta_j$, $j=1,\dots,\nu$, with $\sum_{j=1}^\nu \delta_j\psi_j=0$. Consider the disjoint union $V_f' = V_f \sqcup \{v_{\nu+1}\}$. Let $\tilde{\psi}_j = \psi_j$ on $V_f$ and $\tilde{\psi}_{\nu+1} = 0$. Then $\cA_{V_f'}\tilde{\psi} = \lambda\tilde{\psi}$. Define $\delta'$ by $\delta'_j = \delta_j$ for $v_j\in V_f$ and $\delta_{\nu+1}' = 1$. Then $\sum_{j=1}^{\nu+1}\delta_j'\tilde{\psi}_j=0$. Here $V_f'$ is disconnected so we cannot directly apply the converse in Theorem~\ref{thm:singlecellcriter}. However, this is why we chose $\delta_{\nu+1}'=1$. In fact, we know we started from a connected $\Gamma$. To show that the graph $\Gamma'$ induced by $\cA(z) = \cA_{V_f'}+B(z)$ is connected, where $B(z)$ is constructed as in Theorem~\ref{thm:singlecellcriter} using $\delta'$, we only need to show that $v_{\nu+1}$ is connected to $V_f$ and $v_{\nu+1}+1_\fa$. By construction of $B(z)$ (and the choice $\delta_{\nu+1}'=1$) we know there is a bridge from $v_{\nu+1}$ to some vertex in $V_f\pm 1_\fa$, hence to any vertex in $\Gamma$. This also means there is a bridge from $v_{\nu+1}\pm 1_\fa$ to $\Gamma$, hence from $v_{\nu+1}$ to $v_{\nu+1}+1_\fa$. Thus, $\Gamma'$ is connected and $\lambda$ is flat for $\Gamma'$, which has $\nu+1$ vertices in its $V_f'$.
\end{proof}

Proposition~\ref{prp:nu2single} characterizes $\cF_2^s$ as those satisfying the neighborhood condition. Interestingly, this is no longer true for $\nu=3$, although we get the same flat bands.

\begin{exa}\label{exa:singlenu3}
Let $\nu=3$. We apply our results so far to show that
\begin{enumerate}
\item $\cF_3^s = \{0,-1\}$.
\item Each of these single-cell bands can appear with multiplicity one or two, but they cannot appear together.
\item Single-cell flat bands may exist even if $\cN_{v_i}\setminus\{v_j\}\neq \cN_{v_j}\setminus\{v_i\}$ for all $i,j$.
\end{enumerate}

To see that $0$ can appear with multiplicity $2$ consider $A(\theta) = (a_{ij}(\theta))$ with $a_{ij}(\theta)=2\cos 2\pi\theta$ for all $i,j$. Similarly, $-1$ appears with multiplicity $2$ if $A(\theta) = (a_{ij}(\theta))$ with $a_{ii}(\theta)=2\cos 2\pi\theta$ and $a_{ij}(\theta) = 1+2\cos 2\pi\theta$ for all $i\neq j$. Both cases have the eigenvectors $(-1,1,0)^\intercal,(-1,0,1)^\intercal$ independent of $\theta$ and yield connected $\Gamma$.

Examples where $0$ and $-1$ appear as flat bands of multiplicity one for $A(\theta)$ are given in Figure~\ref{fig:nonsymvar} (left) and \ref{fig:deco} (right), respectively.

We now show there are no more bands in $\cF_3^s$. By Lemma~\ref{lem:gensincell}, if $\lambda\in \cF_3^s$, then it is an eigenvalue of one of the graphs in Figure~\ref{fig:3graphs}.

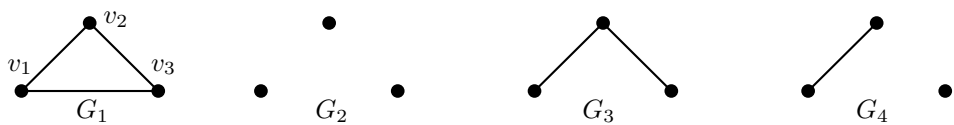
\begin{figure}[h!]
\begin{center}
\setlength{\unitlength}{0.9cm}
\thicklines
\begin{picture}(1.1,1.1)(-0.9,-1.1)
	 \put(-7,-1){\line(1,0){2}}
	 \put(-7,-1){\line(1,1){1}}
	 \put(-5,-1){\line(-1,1){1}}
	   \put(0.5,-1){\line(1,1){1}}
	  \put(2.5,-1){\line(-1,1){1}}
	    \put(4.5,-1){\line(1,1){1}}
	    \put(-7,-1){\circle*{.2}}
	 \put(-6,0){\circle*{.2}}
	 \put(-5,-1){\circle*{.2}}
	    \put(-3.5,-1){\circle*{.2}}
	 \put(-2.5,0){\circle*{.2}}
	 \put(-1.5,-1){\circle*{.2}}
	    \put(0.5,-1){\circle*{.2}}
	 \put(1.5,0){\circle*{.2}}
	 \put(2.5,-1){\circle*{.2}}
	     \put(4.5,-1){\circle*{.2}}
	 \put(5.5,0){\circle*{.2}}
	 \put(6.5,-1){\circle*{.2}}
	    \put(-7.2,-0.7){\small{$v_1$}}
	     \put(-5.1,-0.7){\small{$v_3$}}
	      \put(-5.8,0){\small{$v_2$}}
		\put(-6.2,-1.4){\small{$G_1$}}
		\put(-2.7,-1.4){\small{$G_2$}}
		\put(1.2,-1.4){\small{$G_3$}}
		\put(5.2,-1.4){\small{$G_4$}}
\end{picture}
\caption{All graphs on $3$ vertices.}\label{fig:3graphs}
\end{center}
\end{figure}

The corresponding eigenvalues and eigenvectors are
\begin{small}
\begin{alignat*}{3}
& G_1   \qquad && 2,-1,-1            \qquad && (1,1,1),(-1,0,1),(-1,1,0)\\
& G_2   \qquad && 0,0,0            \qquad && (1,0,0),(0,1,0),(0,0,1)\\
& G_3   \qquad && \sqrt{2},0,-\sqrt{2}            \qquad && (1,\sqrt{2},1),(-1,0,1),(1,-\sqrt{2},1)\\
& G_4   \qquad && 1,0,-1            \qquad && (1,1,0),(0,0,1),(-1,1,0)\,.
\end{alignat*}
\end{small}

The eigenvalues $2$ and $\sqrt{2}$ only arise as top eigenvalues of the connected $G_1,G_3$, respectively, so by Lemma~\ref{lem:gensincell} they cannot be flat. Theorem~\ref{thm:singlecellcriter} implies that $-\sqrt{2}\notin \cF_3^s$.

The eigenvalue $1\notin \cF_3^s$. If it was, we would have $V_f=G_4$. Since $1$ is simple for $\cA_{V_f}$, if $1\in \cF_3^s$, then $(1,1,0)^\intercal$ would have to be an eigenvector of $A(\theta)$ for all $\theta$ with eigenvalue $1$. This would imply $\overline{h_{13}}+\overline{h_{23}}=0$ as polynomials, so $h_{13}=h_{23}=0$. But this would make $\Gamma$ disconnected, with $v_3$ not connected to any $v_i+n_\fa$ for $i=1,2$. Thus, $1\notin \cF_3^s$.

Next, the only possibility to have both $0$ and $-1$ appear as single cell flat bands is if $V_f=G_4$. But $0$ cannot be flat in this case. In fact $0$ is simple for $\cA_{V_f}$, so the corresponding eigenvector would be a multiple of $(0,0,1)^\intercal$. This would imply $h_{13}=0$, $h_{23}=0$, $h_{33}=0$, yielding again a disconnected $\Gamma$. This shows that $0$ and $-1$ cannot appear together.

Finally, $A(\theta)= \begin{pmatrix} 2\cos 2\pi\theta&0&2\cos 2\pi\theta\\ 0&2\cos 4\pi\theta& 2\cos 4\pi\theta\\ 2\cos 2\pi\theta& 2\cos 4\pi\theta& 2\cos 2\pi\theta + 2\cos 4\pi\theta\end{pmatrix}$ induces a connected $\Gamma$, with eigenvector $(1,1,-1)^\intercal$ localized on a single cell, but $\cN_{v_i}\setminus\{v_j\}\neq \cN_{v_j}\setminus\{v_i\}$, $\forall i,j$.
\end{exa}

We can refine Corollary~\ref{cor:increaflat} by describing more precisely how $\cF_{\nu+1}^s$ differs from $\cF_\nu^s$.

\begin{thm}\label{thm:moredetailsonfnu+1}
We may write $\cF_{\nu+1}^s$ as a disjoint union
\[
\cF_{\nu+1}^s = \cF_\nu^s \sqcup C_{\nu+1}\,,
\]
where $\lambda\in C_{\nu+1}$ implies that $\lambda\in \sigma(\cA_{V_f})$ for a \emph{connected} graph on $\nu+1$ vertices, or a disconnected graph $V_f = V_f^1\sqcup V_f^2$ such that $V_f^i$ are connected, $|V_f^1| = |V_f^2| = \frac{\nu+1}{2}$ and $\lambda\in \sigma(\cA_{V_f^1})\cap\sigma(\cA_{V_f^2})$.

The set $C_{\nu+1}$ can be empty.

In particular, if $\nu+1$ is odd, only connected graphs on $\nu+1$ vertices may contribute new flat bands in $\cF_{\nu+1}^s$.
\end{thm}
Besides its theoretical input, the purpose of Theorem~\ref{thm:moredetailsonfnu+1} is to eliminate the disconnected graphs from consideration. For example, for $\nu=5$, this reduces the analysis from $34$ graphs to $21$ graphs. Let us clarify that disconnected graphs on $\nu+1$ vertices \emph{do} contribute flat bands. We are just saying that the flat bands they produce are old ones in $\cF_{\nu}^s$, with the exception of having $V_f$ precisely partitioned into two equal parts. This last scenario is in harmony with our recipe, Theorem~\ref{thm:eigenchar}.

\begin{proof}[Proof of Theorem~\ref{thm:moredetailsonfnu+1}]
Let $\lambda\in \cF_{\nu+1}^s\setminus \cF_\nu^s$. Then $\lambda\in \sigma(\cA_{V_f})$ for some graph $V_f$ on $\nu+1$ vertices. Suppose $V_f$ is disconnected, $V_f = \sqcup_{i=1}^r V^i$, for some connected $V^i$, so $|V^i|\le \nu$. Then $\sigma(\cA_{V_f}) = \cup_{i=1}^r \sigma(\cA_{V^i})$. We claim that if $|V^i|<\frac{\nu+1}{2}$, then $\lambda\notin\sigma(\cA_{V^i})$. Indeed if $\lambda\in \sigma(\cA_{V^i})$, then Theorem~\ref{thm:eigenchar} and Corollary~\ref{cor:increaflat} would imply $\lambda\in \cF_{2|V^i|}^s \subseteq \cF_\nu^s$, a contradiction. 

If $|V^i|<\frac{\nu+1}{2}$ for all $i$, we get a contradiction. If not, we may assume $|V^1|\ge \frac{\nu+1}{2}$. Then either $|V^i|<\frac{\nu+1}{2}$ for all other $i$, or $V_f = V^1\sqcup V^2$ with $|V^1|=|V^2|=\frac{\nu+1}{2}$. In the latter case, if $\lambda\in \sigma(\cA_{V^1})\cap \sigma(\cA_{V^2})$ then we are done. So suppose $\lambda\notin \sigma(\cA_{V^2})$. Then in all cases, we are in a situation where $\lambda\in \sigma(\cA_{V^1})$ and $\lambda\notin \sigma(\cA_{V^i})$ for $i>1$.

Now, since $\lambda\in \cF_{\nu+1}^s$, we know by Theorem~\ref{thm:singlecellcriter} that for the corresponding $\psi$, we may find $\delta_i\in \{0,1\}$ not all zero, such that $\sum_{i=1}^{\nu+1}\delta_i \psi_i=0$. Denote $\psi = (\psi^{(1)},\dots,\psi^{(r)})$ for the components of $\psi$ over $V^i$. Then for $v\in V^k$, we have $\cA\psi^{(k)}(v) = \cA\psi(v) = \lambda\psi(v) = \lambda\psi^{(k)}(v)$. Since $\lambda\notin\sigma(\cA_{V^k})$ for $k>1$, we must have $\psi^{(k)}=0$ for $k>1$.

Thus, $\psi = (\psi^{(1)},0,\dots,0)$. So $\lambda\in \sigma(\cA_{V^1})$, $0=\sum_{i=1}^{\nu+1}\delta_i \psi_i = \sum_{v\in V^1} \delta_v \psi^{(1)}_v$, and $V^1$ is connected. By Theorem~\ref{thm:singlecellcriter} and Corollary~\ref{cor:increaflat}, this implies $\lambda \in \cF_{|V^1|}^s \subseteq \cF_\nu^s$, again a contradiction. This completes the proof of the statement.
\end{proof}

So far, all flat bands have been $\{0,-1\}$, so integers and nonpositive. Interestingly, if $\nu=4$, the flat band can be positive and irrational.

\begin{exa}\label{exa:nu4single}
The set of single-cell flat bands for $\nu=4$ and $\nu=5$ are the same, with
\begin{equation}\label{e:f45}
\cF_4^s= \Big\{0,1, -1,-2,\frac{-1+\sqrt{5}}{2},\frac{-1-\sqrt{5}}{2}\Big\} = \cF_5^s\,.
\end{equation}

The fact that $\cF_5^s$ does not bring any new flat bands is quite remarkable, as there is a total of $34$ graphs on $5$ vertices which in principle could provide some.

To prove \eqref{e:f45}, we simply apply Theorem~\ref{thm:singlecellcriter} to the subset of graphs that Theorem~\ref{thm:moredetailsonfnu+1} tells us to consider. This however requires the tables of eigenvalues and eigenvectors of graphs on $4,5$ vertices; we provide this in Appendix~\ref{app:45}.
\end{exa}

We showed that
\begin{equation}\label{e:fscroi}
\emptyset = \cF_1^s\subset \cF_2^s = \cF_3^s \subset \cF_4^s = \cF_5^s
\end{equation}
and $\cup_\nu \cF_\nu^s$ is the set of all totally real algebraic integers. It would be interesting to know if \eqref{e:fscroi} is a general fact or a coincidence, i.e. do we always have $\cF_{2k+1}^s = \cF_{2k}^s$~? This would avoid the whole discussion in \S\ref{app:5ver}. In view of Theorem~\ref{thm:moredetailsonfnu+1}, we may state the question equivalently as follows: is it true that the connected graphs on $2k+1$ vertices only offer flat bands in $\cF_{2k}^s$~? In other words, the eigenvectors of such graphs corresponding to ``new'' eigenvalues never satisfy $\sum_{j=1}^{\nu+1} \delta_j \psi_j=0$, $\delta_j\in \{0,1\}$, except when $\delta_j=0$ for all $j$~?

To conclude, let us mention that there exist flat bands beyond single-cells even in the nearest-neighbor hopping case, as Figure~\ref{fig:curious} shows.

\begin{figure}[h!]
\begin{center}
\setlength{\unitlength}{1cm}
\thicklines
\begin{picture}(2,2)(-1,-2)
	 \put(-7,0){\line(1,0){6}}
	 \put(-7,-2){\line(1,0){6}}
	 \put(-6,0){\line(1,-1){2}}
	 \put(-4,0){\line(1,-1){2}}
	 \put(-4,-2){\line(1,1){2}}
	 \put(-6,-2){\line(1,1){2}}
	 \put(-2,-2){\line(1,1){1}}
	 \put(-2,0){\line(1,-1){1}}
	 \put(-3,-1){\circle*{.2}}
	 \put(-4,0){\circle*{.2}}
	 \put(-5,-1){\circle*{.2}}
	 \put(-6,0){\circle*{.2}}
	 \put(-1,-1){\circle*{.2}}
 	 \put(-6,-2){\circle*{.2}}
	 \put(-4,-2){\circle*{.2}}
	 \put(-2,-2){\circle*{.2}}
	 \put(-2,0){\circle*{.2}}
	 \put(-3.9,0.1){\small{$1$}}
	 \put(-3.6,-1.1){\small{$-1$}}
	 \put(-4.8,-1.1){\small{$-1$}}
	 \put(-3.9,-2.3){\small{$1$}}
	%
	 \put(2,0){\line(1,-1){1}}
	 \put(4,0){\line(1,-1){1}}
	 \put(6,0){\line(1,-1){1}}
	 \put(3,-1){\line(1,1){1}}
	 \put(5,-1){\line(1,1){1}}
	 \put(3,0){\line(1,-1){1}}
	 \put(5,0){\line(1,-1){1}}
	 \put(2,-1){\line(1,1){1}}
	 \put(4,-1){\line(1,1){1}}
	 \put(6,-1){\line(1,1){1}}
	 \put(2,-2){\line(0,1){2}}
	 \put(3,-2){\line(0,1){2}}
	 \put(4,-2){\line(0,1){2}}
	 \put(5,-2){\line(0,1){2}}
	 \put(6,-2){\line(0,1){2}}
	 \put(7,-2){\line(0,1){2}}
	 \put(5,-1){\circle*{.2}}
	 \put(4,0){\circle*{.2}}
	 \put(3,-1){\circle*{.2}}
	 \put(2,0){\circle*{.2}}
	 \put(7,-1){\circle*{.2}}
 	 \put(2,-2){\circle*{.2}}
	 \put(4,-2){\circle*{.2}}
	 \put(6,-2){\circle*{.2}}
	 \put(6,0){\circle*{.2}}
	 \put(7,0){\circle*{.2}}
	 \put(5,0){\circle*{.2}}
	 \put(3,0){\circle*{.2}}
	 \put(6,-1){\circle*{.2}}
	 \put(4,-1){\circle*{.2}}
	 \put(2,-1){\circle*{.2}}
	 \put(3,-2){\circle*{.2}}
	 \put(5,-2){\circle*{.2}}
	 \put(7,-2){\circle*{.2}}
	 \put(4.1,0.1){\small{$1$}}
	 \put(3.1,0.1){\small{$0$}}
	 \put(5.1,0.1){\small{$0$}}
	 \put(4.2,-1.1){\small{$0$}}
	 \put(3.2,-1.1){\small{$0$}}
	 \put(5.2,-1.1){\small{$0$}}
	 \put(3.9,-2.3){\small{$-1$}}
	 \put(4.9,-2.3){\small{$-1$}}
	 \put(2.9,-2.3){\small{$-1$}}
\end{picture}
\caption{Nearest hopping, $\nu=3$. Flat bands $\lambda=-2$ (left), $\lambda=0$ (right).}\label{fig:curious}
\end{center}
\end{figure}
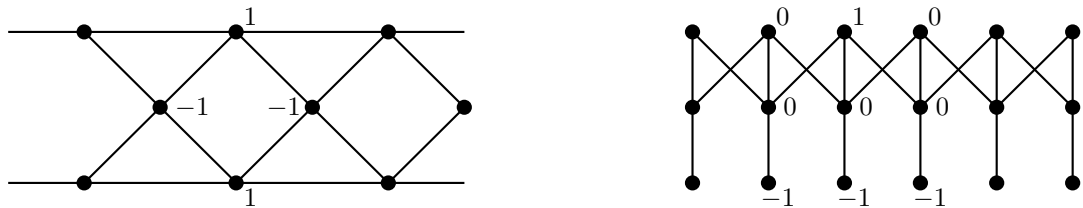

For the graph on Figure~\ref{fig:curious} (left), $-2$ is not even an eigenvalue of a graph on $3$ vertices. This is because $I_{12}$ is not symmetric. When all hopping are symmetric like Figure~\ref{fig:curious} (right), then the flat band is necessarily an eigenvalue of $V_f$ if $d=\fh_1=1$. This is because the Floquet matrix takes the form $A(\theta) = \cA_{V_f} + 2\cos 2\pi\theta B$ for some matrix $B$ encoding the first-hopping, so if $\lambda$ is an eigenvalue for all $\theta$, it is in particular for $\theta=\frac{1}{4}$.

\section{Case of two vertices}\label{sec:2ver}

In the following, $d$ is arbitrary, $\nu=2$, $Q\equiv 0$, $w\equiv 1$. Proposition~\ref{prp:nu2single} characterizes single-cell flat bands. Here we study existence beyond single-cells. We also show that flat bands are always integers if $\nu=2$, and we give two formulas for them in terms of $I_{ij}$.

Let $z_j=\ee^{2\pi\ii \theta_j}$. We denote \eqref{e:hthe} as $A(z) = \begin{pmatrix} f_1(z)& g(z)\\ g(z^{-1})& f_2(z)\end{pmatrix}$, with $f_i(z) = \sum_{k\in I_{ii}}z^k$ and $g(z)=\sum_{p\in I_{12}}z^p$. Then $\lambda$ is flat iff $(f_1(z)-\lambda)(f_2(z)-\lambda)=g(z)g(z^{-1})$ as Laurent polynomials, by Lemma~\ref{lem:triv}.

We first note that we cannot have $|I_{12}|=0$, otherwise $A(z) = \begin{pmatrix}  f_1(z)& 0\\ 0& f_2(z)\end{pmatrix}$ induces a disconnected graph $\Gamma$ in which $v_1$ is not connected to any $v_2+n_\fa$. So
\begin{equation}\label{e:i12ge1}
|I_{12}|\ge 1
\end{equation}
for any periodic connected graph with $\nu=2$.

\subsection{Absence of flat bands}
Most arguments here are based on degree analysis of Laurent polynomials.
Let $I_{ij}(r)=\{k_r\in \Z:k\in I_{ij}\}$ for $1\le r\le d$ and	
\[
\begin{cases}
k_{\max}(r) = \max I_{11}(r)\,, \qquad k_{\max}'(r) = \max I_{22}(r)\,,\\
p_{\max}(r) = \max I_{12}(r)\,, \,\qquad p_{\min}(r) = \min I_{12}(r)\,. 
\end{cases}
\]

If $\Gamma$ has a flat band then for all $r$,
\begin{equation}\label{e:lar}
k_{\max}(r)+k_{\max}'(r)=p_{\max}(r)-p_{\min}(r)\,.
\end{equation}
In fact, we have $(\sum_{I_{11}} z^k-\lambda)(\sum_{I_{22}}z^{k'}-\lambda) = (\sum_p z^p)(\sum_p z^{-p})$. Evaluate at $z=(z_r,1)$, so $z_i=1$ for $i\neq r$. As Laurent polynomials in $z_r$, the largest degrees on both sides are $k_{\max}(r)+k_{\max}'(r)$ and $p_{\max}(r)-p_{\min}(r)$, respectively, so they coincide.

It follows in particular that if $\Gamma$ has a flat band, then
\begin{equation}\label{e:i1122ge1}
|I_{11}|+|I_{22}|\ge 2 \,.
\end{equation}
In fact, if $I_{11}=I_{22}=\emptyset$, then $k_{\max}(r)+k_{\max}'(r) = 0 = p_{\max}(r)-p_{\min}(r)$ for each $r$, so $I_{12}=\{p\}$ or $I_{12}=\emptyset$, i.e. $A(z) = \begin{pmatrix} 0&\epsilon z^p\\ \epsilon z^{-p}&0\end{pmatrix}$ for $\epsilon\in\{0,1\}$ and $\Gamma$ is disconnected. So $I_{11}$ or $I_{22}$ is nonempty, and since each is symmetric without $0$, each has even cardinality.

Expand the characteristic polynomial,
\begin{equation}\label{e:char2}
p(z;\lambda)=\lambda^2 - \Big(\sum_{k\in I_{11}} z^k + \sum_{k'\in I_{22}} z^{k'}\Big)\lambda + \sum_{\substack{k\in I_{11}\\k'\in I_{22}}} z^{k+k'} - \sum_{p,p'\in I_{12}} z^{p-p'}\,.
\end{equation}

\begin{lem}\label{lem:vap}
If $\cA_\Gamma$ has a flat band $\lambda$, then $\lambda = \pm \sqrt{|I_{12}|-|I_{11}\cap I_{22}|} \in \Z$.
\end{lem}
\begin{proof}
Let $\lambda$ be a flat band. By Lemma~\ref{lem:triv}, the coefficient of $z^0$ in \eqref{e:char2} vanishes identically. Since $0\notin I_{ii}$, this coefficient is precisely $\lambda^2 + |I_{11}\cap I_{22}|-|I_{12}|$. Indeed, it comes from taking, for each $k\in I_{11}$, a corresponding $k'\in I_{22}$ with $k'=-k$, and taking for each $p\in I_{12}$ a corresponding $p'=p$. The latter has cardinality $|I_{12}|$. For the former, if $k'=-k$, then since $I_{11}$ is symmetric, we get $k'\in I_{11}$, so $k,k'\in I_{11}\cap I_{22}$. Conversely, for each $k\in I_{11}\cap I_{22}$ there is a corresponding $k'=-k$ contributing to the coefficient.

We showed that $\lambda = \pm \sqrt{|I_{12}|-|I_{11}\cap I_{22}|}$. Finally, since $p(z;\lambda)\equiv 0$ and the double sums in \eqref{e:char2} have coefficients $\pm 1$, then $\lambda$ must be an integer.
\end{proof}

\begin{lem}\label{lem:othervap}
If $\lambda$ is a flat band, we also have $\lambda = \frac{|I_{11}|+|I_{22}|-\sqrt{(|I_{11}|-|I_{22}|)^2+4|I_{12}|^2}}{2}$.
\end{lem}
\begin{proof}
Consider $z=1$ in \eqref{e:char2}. This yields $\lambda^2-\lambda(|I_{11}|+|I_{22}|)+|I_{11}||I_{22}|-|I_{12}|^2=0$. So $\lambda = \frac{|I_{11}|+|I_{22}|\pm \sqrt{(|I_{11}|+|I_{22}|)^2-4|I_{11}||I_{22}|+4|I_{12}|^2}}{2} = \frac{|I_{11}|+|I_{22}|\pm \sqrt{(|I_{11}|-|I_{22}|)^2+4|I_{12}|^2}}{2}$.

Using \eqref{e:i1122ge1}, $|I_{11}|+|I_{22}|\neq 0$. If the $+$ sign occurs in $\lambda$, then $\lambda > \frac{\sqrt{(|I_{11}|-|I_{22}|)^2+4|I_{12}|^2}}{2}\ge |I_{12}|$. This contradicts Lemma~\ref{lem:vap}, so the sign must be negative.
\end{proof}

\begin{lem}\label{lem:nu2nfb}
If $\nu=2$, then $\cA_\Gamma$ has no flat bands if any of the following holds:
\begin{enumerate}[\rm(i)]
\item $k_{\max}(r)+k_{\max}'(r)\neq p_{\max}(r)-p_{\min}(r)$ for some $r$.
\item $|I_{12}|=1$.
\item $I_{11}=\emptyset$ or $I_{22}=\emptyset$.
\end{enumerate}
\end{lem}
For example, the honeycomb lattice has $I_{ii}=\emptyset$ and no flat bands. Note that this is not true for $\nu>2$. The graph in Figure~\ref{fig:deco} (right) has $I_{11}=I_{22}=\emptyset$ but has a flat band. The graph in Figure~\ref{fig:pyro2} has $I_{ii}=\emptyset$ for all $i$, but has flat bands.
\begin{proof}
We proved (i) in \eqref{e:lar}. For (ii), if $\Gamma$ has a flat band and $|I_{12}|=1$, then by \eqref{e:lar}, $k_{\max}(r)+k_{\min}(r)=0$, so $I_{11}=I_{22}=\emptyset$, contradicting \eqref{e:i1122ge1}. 
Thus, $\Gamma$ has no flat bands.

For (iii), suppose on the contrary that $\lambda$ is flat. By \eqref{e:i1122ge1}, $I_{11}$ and $I_{22}$ cannot both be empty, so we may assume, by symmetry, that $I_{11}\neq \emptyset$ and $I_{22}=\emptyset$.

Then $\lambda^2 - \lambda\sum_{k\in I_{11}} z^k = \sum_{p,p'\in I_{1,2}} z^{p-p'}$ for all $z$. On the RHS, if $d=1$, then the highest power $p_{\max}-p_{\min}$ appears only once. In general, the maximum power in lexicographic order appears only once. More precisely, we choose the power $n$ with highest $p_{\max}(1)-p_{\min}(1)$. This will in general be a set of powers. Among these elements, we choose the one with highest $p_{\max}(2)-p_{\min}(2)$, and so on. Then the coefficient of $z^n$ is $1$, and we must have $\lambda=-1$ on the LHS.

But by Lemma~\ref{lem:vap}, $\lambda^2 = |I_{12}|$, so $|I_{12}|=1$. This contradicts (ii).
\end{proof}

\subsection{Existence beyond single-cells}\label{sec:2bsc}
It is easy to generate flat bands which do not live on a single cell by taking $\Gamma$ such that $I_{11}=I_{22} = (I_{12}\setminus\{\ell\})-\ell$, for some integer $\ell$. For example, consider $I_{11}=I_{22}=\{\pm 1\}$ and $I_{12}=\{0,2\}$. This gives the graph in Figure~\ref{fig:exnonsy}. Here $\ell=1$, $A(\theta) = \begin{pmatrix} 2\cos 2\pi\theta& 1+\ee^{4\pi\ii\theta}\\ 1+\ee^{-4\pi\ii\theta}& 2\cos 2\pi\theta\end{pmatrix}$, with eigenvector $\begin{pmatrix} 1\\-\ee^{-2\pi\ii\theta}\end{pmatrix}$ for $\lambda=0$. We may construct an eigenvector for $\Gamma$ according to Proposition~\ref{prp:trigiffcomp}, here $\alpha_1(0)=1$, $\alpha_2(1)=-1$, $\alpha_k(m)=0$ otherwise.

\begin{figure}[h!]
\begin{center}
\setlength{\unitlength}{1cm}
\thicklines
\begin{picture}(1.3,1.2)(-1.1,-1.05)
   \put(-5,0){\line(1,0){10}}
	 \put(-5,-1){\line(1,0){10}}
	 \put(-5,0){\line(2,-1){2}}
	 \put(-3,-1){\line(0,1){1}}
	 \put(-3,0){\line(2,-1){2}}
	 \put(-1,-1){\line(0,1){1}}
	 \put(-1,0){\line(2,-1){2}}
	 \put(1,-1){\line(0,1){1}}
	 \put(1,0){\line(2,-1){2}}
	 \put(3,-1){\line(0,1){1}}
	 \put(3,0){\line(2,-1){2}}
	 \put(-1,-1){\circle*{.2}}
	 \put(-1,0){\circle*{.2}}
	 \put(-3,-1){\circle*{.2}}
	 \put(-3,0){\circle*{.2}}
	 \put(1,-1){\circle*{.2}}
	 \put(1,0){\circle*{.2}}
	 \put(3,-1){\circle*{.2}}
	 \put(3,0){\circle*{.2}}
	 	 \put(-2,-1){\circle*{.2}}
	 \put(-2,0){\circle*{.2}}
	 \put(-4,-1){\circle*{.2}}
	 \put(-4,0){\circle*{.2}}
	 \put(0,-1){\circle*{.2}}
	 \put(0,0){\circle*{.2}}
	 \put(2,-1){\circle*{.2}}
	 \put(2,0){\circle*{.2}}
	 \put(0.9,0.2){\small{$1$}}
	 \put(1.8,-1.5){\small{$-1$}}
	 \put(1.9,0.2){\small{$0$}}
	 \put(0.8,-1.5){\small{$0$}}
	 \put(-2,0){\line(2,-1){2}}
	 \put(0,0){\line(2,-1){2}}
	 \put(2,0){\line(2,-1){2}}
	 \put(-4,0){\line(2,-1){2}}
	  \put(-2,-1){\line(0,1){1}}
	 \put(0,-1){\line(0,1){1}}
	 \put(2,-1){\line(0,1){1}}
	 \put(-4,-1){\line(0,1){1}}
\end{picture}

\caption{Flat band $\lambda=0$ with $I_{12}$ non-symmetric and localized $\psi$.}\label{fig:exnonsy}
\end{center}
\end{figure}
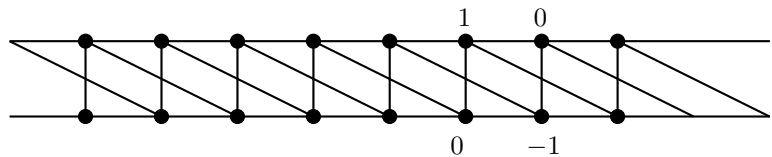

We notice however that such a graph is isomorphic to a graph $\Gamma'$ on which the flat band lives on a single cell; simply take $\phi:(v_1+k_\fa)\mapsto (v_1+\ell_\fa+k_\fa)$, $\phi:(v_2+k_\fa)\mapsto (v_2+k_\fa)$ and for edges, $\phi:\{u,u'\}\mapsto \{\phi(u),\phi(u')\}$. For the graph in Figure~\ref{fig:exnonsy}, this corresponds to a shearing of the top layer; the isomorphic graph $\Gamma'$ is the one of Figure~\ref{fig:boxnover} (right).

We may wonder if these are the only non-single-cell flat bands. Note that such $I_{12}$ are not symmetric if $\ell\neq 0$. Therefore, we take the special case where $I_{12}$ is symmetric and ask whether all flat bands are in $\cF_2^s$.

The answer is no. In fact, take $I_{11}=\{\pm 3\}$, $I_{22}=\{\pm 1\}$, $I_{12}=\{\pm 1,\pm 2\}$, so $A(z) = \begin{pmatrix} z^3+z^{-3}& z+z^{-1}+z^2+z^{-2}\\ z+z^{-1}+z^2+z^{-2}& z+z^{-1}\end{pmatrix}$. Then we have a flat band $\lambda=-2$ corresponding to the eigenvector $\begin{pmatrix} -2-z-z^{-1}\\ z+z^{-1}+z^2+z^{-2}\end{pmatrix}$. Using the recipe in Proposition~\ref{prp:trigiffcomp}, we construct a corresponding eigenvector for $\Gamma$ in Figure~\ref{fig:counterconj}.

\begin{figure}[h!]
\begin{center}
\setlength{\unitlength}{1cm}
\thicklines
\begin{picture}(1.3,1.6)(-1.1,-1.1)
	 \put(-7.5,-1){\line(1,0){13}}
	 \put(-5,0){\line(2,-1){2}}
	 \put(-5,0){\line(-2,-1){2}}
	 \put(-3,0){\line(2,-1){2}}
	 \put(-3,0){\line(-2,-1){2}}
	 \put(-1,0){\line(2,-1){2}}
	 \put(-1,0){\line(-2,-1){2}}
	 \put(1,0){\line(2,-1){2}}
	 \put(1,0){\line(-2,-1){2}}
	 \put(3,0){\line(2,-1){2}}
	 \put(3,0){\line(-2,-1){2}}
	 \put(-1,-1){\textcolor{red}{\circle*{.2}}}
	 \put(-1,0){\textcolor{red}{\circle*{.2}}}
	 \put(-3,-1){\circle*{.2}}
	 \put(-3,0){\circle*{.2}}
	 \put(1,-1){\textcolor{red}{\circle*{.2}}}
	 \put(1,0){\textcolor{red}{\circle*{.2}}}
	 \put(3,-1){\circle*{.2}}
	 \put(3,0){\circle*{.2}}
	 	 \put(-2,-1){\textcolor{red}{\circle*{.2}}}
	 \put(-2,0){\textcolor{red}{\circle*{.2}}}
	 \put(-4,-1){\circle*{.2}}
	 \put(-4,0){\circle*{.2}}
	 \put(0,-1){\textcolor{red}{\circle*{.2}}}
	 \put(0,0){\textcolor{red}{\circle*{.2}}}
	 \put(2,-1){\textcolor{red}{\circle*{.2}}}
	 \put(2,0){\textcolor{red}{\circle*{.2}}}
	    \put(4,-1){\circle*{.2}}
	 \put(4,0){\circle*{.2}}
	 \put(-5,-1){\circle*{.2}}
	 \put(-5,0){\circle*{.2}}
	 \put(2.9,0.2){\small{$0$}}
	 \put(1.9,0.2){\small{$0$}}
	 \put(0.7,0.2){\small{$-1$}}
	 \put(-0.3,0.2){\small{$-2$}}
	 \put(-1.3,0.2){\small{$-1$}}
	 \put(-2.1,0.2){\small{$0$}}
	 \put(-3.1,0.2){\small{$0$}}
	 \put(2.95,-1.5){\small{$0$}}
	 \put(1.95,-1.5){\small{$1$}}
	 \put(0.95,-1.5){\small{$1$}}
	 \put(-0.05,-1.5){\small{$0$}}
	 \put(-1.05,-1.5){\small{$1$}}
	 \put(-2.05,-1.5){\small{$1$}}
	 \put(-3.05,-1.5){\small{$0$}}
	 \put(-2,0){\line(2,-1){2}}
	 \put(0,0){\line(2,-1){2}}
	 \put(2,0){\line(2,-1){2}}
	 \put(-4,0){\line(2,-1){2}}
	 \put(-2,0){\line(-2,-1){2}}
	 \put(0,0){\line(-2,-1){2}}
	 \put(2,0){\line(-2,-1){2}}
	 \put(-4,0){\line(-2,-1){2}}
	     \put(-2,0){\line(1,-1){1}}
	 \put(0,0){\line(1,-1){1}}
	 \put(2,0){\line(1,-1){1}}
	 \put(-4,0){\line(1,-1){1}}
	 \put(-2,0){\line(-1,-1){1}}
	 \put(0,0){\line(-1,-1){1}}
	 \put(2,0){\line(-1,-1){1}}
	 \put(-4,0){\line(-1,-1){1}}
	 \put(3,0){\line(1,-1){1}}
	 \put(3,0){\line(-1,-1){1}}
	 \put(-5,0){\line(1,-1){1}}
	 \put(-5,0){\line(1,-1){1}}
	    \put(-3,0){\line(1,-1){1}}
	 \put(-3,0){\line(-1,-1){1}}
	 \put(-1,0){\line(1,-1){1}}
	 \put(-1,0){\line(-1,-1){1}}
	 \put(1,0){\line(1,-1){1}}
	 \put(1,0){\line(-1,-1){1}}
	 \put(-5,0){\line(-1,-1){1}}
	 \put(4,0){\line(-1,-1){1}}
	 \put(4,0){\line(-2,-1){2}}
	 \put(5,0){\line(-1,-1){1}}
	 \put(5,0){\line(-2,-1){2}}
	 \put(-6,0){\line(2,-1){2}}
	 \put(-6,0){\line(1,-1){1}}
	 \put(-7,0){\line(2,-1){2}}
	     \put(-5,0){\line(2,1){1}}
	     \put(-4,0.5){\line(1,0){1}}
	     \put(-3,0.5){\line(2,-1){1}}
	     \put(-4,0){\textcolor{blue}{\line(2,1){1.5}}}
	     \put(-2.5,0.75){\textcolor{blue}{\line(2,-1){1.5}}}
	  \put(-3,0){\line(2,1){1}}
	     \put(-2,0.5){\line(1,0){1}}
	     \put(-1,0.5){\line(2,-1){1}}
	       \put(-2,0){\textcolor{blue}{\line(2,1){1.5}}}
	       \put(-0.5,0.75){\textcolor{blue}{\line(2,-1){1.5}}}
	    \put(-1,0){\line(2,1){1}}
	     \put(0,0.5){\line(1,0){1}}
	     \put(1,0.5){\line(2,-1){1}}
	        \put(0,0){\textcolor{blue}{\line(2,1){1.5}}}
	        \put(1.5,0.75){\textcolor{blue}{\line(2,-1){1.5}}}
	      \put(1,0){\line(2,1){1}}
	     \put(2,0.5){\line(1,0){1}}
	     \put(3,0.5){\line(2,-1){1}}
	        \put(2,0){\textcolor{blue}{\line(2,1){1.5}}}
	        \put(3.5,0.75){\textcolor{blue}{\line(2,-1){1.5}}}
	  \put(-6,0){\textcolor{blue}{\line(2,1){1.5}}}
	        \put(-4.5,0.75){\textcolor{blue}{\line(2,-1){1.5}}}
\end{picture}

\caption{Flat band $\lambda=-2$ with localized eigenvector on $5$ cells (in red). The edge colors are only there to slightly improve visibility.}\label{fig:counterconj}
\end{center}
\end{figure}
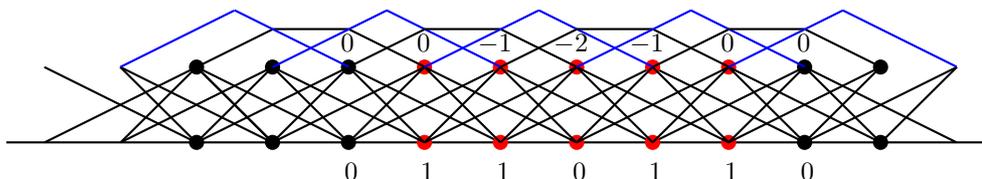

In contrast to the single-cell case, here $-2$ is not an eigenvalue of a graph on $2$ vertices.

It is then natural to ask: what are the flat bands that can be generated by $\nu=2$~? We know they are always integers, so not all flat bands can appear in this fashion. Are there other values of $\lambda$ besides $0,-1,-2$~? If not, this would induce some constraints on the possible choices of $I_{ij}$.
%

In this example, $|I_{12}|=4$. We know that if $|I_{12}|=1$, then there are no flat bands by Lemma~\ref{lem:nu2nfb}. If $|I_{12}|=2$ or $|I_{12}|=3$, then the only flat bands that can appear are those living on single-cells, up to the isomorphism explained before. This fact is not entirely obvious, but since we do not need it, let us simply mention that the main idea is to combine Lemma~\ref{lem:vap} and Lemma~\ref{lem:othervap} to deduce constraints on the sets $I_{ii}$, then analyze the degrees and coefficients of the possible characteristic polynomials.

\section{Flat bands from symmetry}\label{sec:symm}

The basic idea in \cite{RMS} is that the appearance of flat bands can sometimes be attributed to the presence of symmetries in $\Gamma$, such as in Figure~\ref{fig:boxnover}. We should be a bit careful by what we mean by symmetry, since the infinite ladder in Figure~\ref{fig:carte} is very symmetric yet features no flat bands. We explain this in some detail in this section and compare our results with \cite{RMS} at the end.

Given an infinite $\Gamma$ and its $V_f$, the symmetries that interest us are those permuting some vertices of $V_f$ while keeping those outside of $V_f$ fixed. Importantly, when we permute vertices, we exchange \emph{the whole star attached to them}. See Figures~\ref{fig:pyro2} and~\ref{fig:ladnot} for illustration. We are interested in the case when $\Gamma$ remains the same after such permutation.

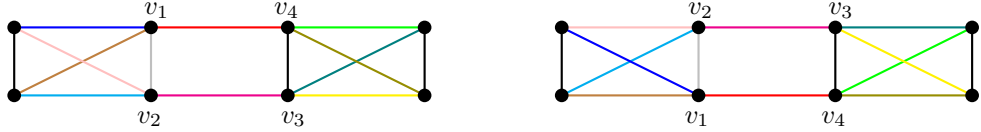
\begin{figure}[h!]
\begin{center}
\setlength{\unitlength}{0.9cm}
\thicklines
\begin{picture}(1.3,1.3)(-1.1,-1.1)
   \put(-7,0){\textcolor{blue}{\line(1,0){2}}}
   \put(-5,0){\textcolor{red}{\line(1,0){2}}}
   \put(-3,0){\textcolor{green}{\line(1,0){2}}}
   \put(-7,-1){\textcolor{cyan}{\line(1,0){2}}}
   \put(-5,-1){\textcolor{magenta}{\line(1,0){2}}}
   \put(-3,-1){\textcolor{yellow}{\line(1,0){2}}}
	 \put(-7,-1){\line(0,1){1}}
	 \put(7,-1){\line(0,1){1}}
	 \put(-7,-1){\textcolor{brown}{\line(2,1){2}}}
	 \put(-7,0){\textcolor{pink}{\line(2,-1){2}}}
	 \put(-5,-1){\textcolor{lightgray}{\line(0,1){1}}}
	 \put(-3,-1){\line(0,1){1}}
	 \put(-3,-1){\textcolor{teal}{\line(2,1){2}}}
	 \put(-3,0){\textcolor{olive}{\line(2,-1){2}}}
	 \put(-1,-1){\line(0,1){1}}
	 \put(1,-1){\line(0,1){1}}
		\put(1,-1){\textcolor{cyan}{\line(2,1){2}}}
	 \put(1,0){\textcolor{blue}{\line(2,-1){2}}}
	 \put(3,-1){\textcolor{lightgray}{\line(0,1){1}}}
	    \put(5,-1){\line(0,1){1}}
	    \put(5,-1){\textcolor{green}{\line(2,1){2}}}
	    \put(5,0){\textcolor{yellow}{\line(2,-1){2}}}
    \put(1,0){\textcolor{pink}{\line(1,0){2}}}
    \put(3,0){\textcolor{magenta}{\line(1,0){2}}}
    \put(5,0){\textcolor{teal}{\line(1,0){2}}}
    \put(1,-1){\textcolor{brown}{\line(1,0){2}}}
    \put(3,-1){\textcolor{red}{\line(1,0){2}}}
    \put(5,-1){\textcolor{olive}{\line(1,0){2}}}
	 \put(-1,-1){\circle*{.2}}
	 \put(-1,0){\circle*{.2}}
	 \put(-3,-1){\circle*{.2}}
	 \put(-3,0){\circle*{.2}}
	 \put(-5,-1){\circle*{.2}}
	 \put(-5,0){\circle*{.2}}
	 \put(1,-1){\circle*{.2}}
	 \put(1,0){\circle*{.2}}
	 \put(3,-1){\circle*{.2}}
	 \put(3,0){\circle*{.2}}
	 \put(5,-1){\circle*{.2}}
	 \put(5,0){\circle*{.2}}
	 \put(-7,-1){\circle*{.2}}
	 \put(-7,0){\circle*{.2}}
	  \put(7,0){\circle*{.2}}
	   \put(7,-1){\circle*{.2}}
	      \put(-5.1,0.2){\small{$v_1$}}
	 \put(-3.2,0.2){\small{$v_4$}}
	 \put(-5.2,-1.4){\small{$v_2$}}
	 \put(-3.1,-1.4){\small{$v_3$}}
	     \put(2.85,0.2){\small{$v_2$}}
	 \put(4.9,0.2){\small{$v_3$}}
	 \put(2.8,-1.4){\small{$v_1$}}
	 \put(4.8,-1.4){\small{$v_4$}}
\end{picture}
\caption{The $1d$ pyrochlore is preserved by the symmetry $v_1\leftrightarrow v_2$, $v_3\leftrightarrow v_4$. Vertices outside $V_f$ are kept fixed. We colored the edges to show their position after performing the symmetry, but they carry no weight.}\label{fig:pyro2}
\end{center}
\end{figure}

\begin{figure}[h!]
\begin{center}
\setlength{\unitlength}{0.9cm}
\thicklines
\begin{picture}(1.3,1.3)(-0.8,-1.1)
   \put(-7,0){\line(1,0){1}}
   \put(-6,0){\textcolor{blue}{\line(1,0){2}}}
   \put(-4,0){\textcolor{red}{\line(1,0){2}}}
   \put(-2,0){\line(1,0){1}}
     \put(-7,-1){\line(1,0){1}}
	 \put(-6,-1){\textcolor{cyan}{\line(1,0){2}}}
	 \put(-4,-1){\textcolor{green}{\line(1,0){2}}}
	 \put(-2,-1){\line(1,0){1}}
	 \put(-6,-1){\line(0,1){1}}
	 \put(-4,-1){\textcolor{brown}{\line(0,1){1}}}
     \put(-2,-1){\line(0,1){1}}
     \put(6,-1){\line(1,0){1}}
     \put(6,0){\line(1,0){1}}
     \put(4,-1){\textcolor{brown}{\line(0,1){1}}}
     \put(4,-1){\textcolor{red}{\line(2,1){2}}}
     \put(4,-1){\textcolor{blue}{\line(-2,1){2}}}
     \put(4,0){\textcolor{green}{\line(2,-1){2}}}
     \put(4,0){\textcolor{cyan}{\line(-2,-1){2}}}
       \put(6,-1){\line(0,1){1}}
       \put(2,-1){\line(0,1){1}}
	 \put(-4,-1){\circle*{.2}}
	 \put(-4,0){\circle*{.2}}
	 \put(-6,-1){\circle*{.2}}
	 \put(-6,0){\circle*{.2}}
	 \put(-2,-1){\circle*{.2}}
	 \put(-2,0){\circle*{.2}}
	 \put(-4.1,0.2){\small{$v_1$}}
	 \put(-4.2,-1.45){\small{$v_2$}}
	 \put(1,0){\line(1,0){1}}
	 \put(1,-1){\line(1,0){1}}
	 \put(4,-1){\circle*{.2}}
	 \put(4,0){\circle*{.2}}
	 \put(2,-1){\circle*{.2}}
	 \put(2,0){\circle*{.2}}
	 \put(6,-1){\circle*{.2}}
	 \put(6,0){\circle*{.2}}
	 \put(3.9,0.2){\small{$v_2$}}
	 \put(3.8,-1.5){\small{$v_1$}}
\end{picture}
\caption{The ladder is not preserved by exchanging $v_1$ and $v_2$.}\label{fig:ladnot}
\end{center}
\end{figure}
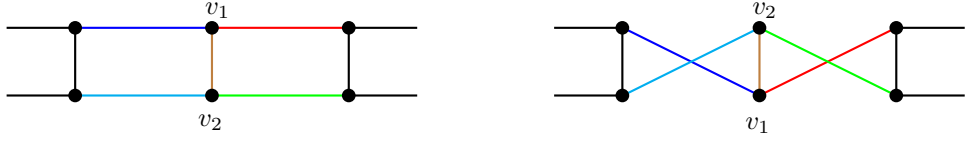

The figures motivate the following definition. Let $\epsilon_{i,j}=1$ if $v_i\sim v_j$ and $\epsilon_{i,j}=0$ otherwise. We say that $\Gamma$ has a \emph{local symmetry} if there exists a nontrivial permutation $\phi$ of $V_f$ such that for any $i,j=1,\dots,\nu$,
\begin{equation}\label{e:locsym}
I_{i,j}\setminus\{0\}=I_{\phi(i),j}\setminus\{0\} \qquad \text{and} \qquad  \epsilon_{i,j}=\epsilon_{\phi(i),\phi(j)}\,.
\end{equation}

If $\Gamma$ has a symmetry preserving it as in the figures, then \eqref{e:locsym} must be satisfied, because if $k\in I_{ij}\setminus\{0\}$, then $v_i\sim v_j+k_\fa$, so after the symmetry, since $v_j+k_\fa$ stays fixed, we must have $v_{\phi(i)}\sim v_j+k_\fa$, i.e. $k\in I_{\phi(i),j}\setminus\{0\}$. Similarly $I_{\phi(i),j}\setminus\{0\}\subseteq I_{ij}\setminus\{0\}$ since $\phi^m=\id$ for some $m$. On the other hand, if $\epsilon_{i,j}=1$, then $v_i\sim v_j$, so after performing the symmetry, we should have $v_{\phi(i)}\sim v_{\phi(j)}$, and vice versa.

In particular, the ladder does not satisfy this condition under $v_1\leftrightarrow v_2$ since $I_{11}=I_{22}=\{\pm 1\}$ while $I_{12}=\{0\} = I_{21}$. However, the pyrochlore satisfies it under $\phi(v_1)=v_2$, $\phi(v_2)=v_1$, $\phi(v_3)=v_4$, $\phi(v_4)=v_3$. Here $I_{ii}=\emptyset$ for all $i$, $I_{12}=I_{34}=\{0\}$, $I_{13}=I_{24}=\{-1\}$, $I_{14}=I_{23}=\{0,-1\}$, and $I_{ji}=-I_{ij}$ for the rest as usual. Also, $\epsilon_{12}=\epsilon_{14}=\epsilon_{21}=\epsilon_{23}=\epsilon_{32}=\epsilon_{34}=\epsilon_{43}=1$, the rest are zero. One checks that indeed \eqref{e:locsym} is satisfied.

We have shown that if $\Gamma$ is preserved by a symmetry, then it is necessary for \eqref{e:locsym} to be satisfied. We now show that \eqref{e:locsym} is enough to get some flat bands. Actually, a weaker condition on $\epsilon_{ij}$ suffices, see Remark~\ref{rem:epsicirc}.

We first note that this symmetry concept is a generalization of the neighborhood condition that we discussed before:

\begin{lem}\label{lem:implies2}
We have $\cN_{v_i}\setminus\{v_j\} = \cN_{v_j}\setminus\{v_i\}$ iff \eqref{e:locsym} is satisfied for the permutation $\phi(v_i)=v_j$, $\phi(v_j)=v_i$, $\phi(v_r)=v_r$ for $r\neq i,j$.
\end{lem}
\begin{proof}
This follows from Lemma~\ref{lem:equi2}.
\end{proof}

We next make the following observation.

\begin{lem}\label{lem:athetablocs}
Suppose \eqref{e:locsym} is satisfied for a permutation $\phi$. Write $\phi$ as a product of $r$ cycles of sizes $k_j$ (we allow $k_j=1$ i.e. fixed elements and allow $k_i=k_j$). Then up to renumbering the vertices of $V_f$, the Floquet matrix $A(\theta)$ acquires an $r\times r$ block structure, with blocks $D_{st}(\theta)$ of size $k_s\times k_t$. Moreover, $D_{st}(\theta) = b_{st}(\theta)J_{st} + C_{st}$, where $J_{st}$ is the all-one matrix of size $k_s\times k_t$ and $C_{st}$ is circulant, independent of $\theta$, for $s,t\le r$.
\end{lem}
\begin{proof}
Let $\phi=(i_1^{(1)}i_2^{(1)}\cdots i_{k_1}^{(1)})(i_1^{(2)}\cdots i_{k_2}^{(2)})\cdots(i_1^{(r)}\cdots i_{k_r}^{(r)})$, and $k_0:=0$.

By renumbering the vertices of $V_f$, we can replace this cumbersome permutation by $(1,2,\dots, k_1)(k_1+1,\dots,k_2)\cdots(k_{r-1}+1,\dots,k_{r})$. The renumbering is simply $v_m'= v_{i_m^{(1)}}$ and $v_{k_n+m}' = v_{i_m^{(n+1)}}$ for $n\geq 1$. These cycles induce the block structure of $A(\theta)$.

Let $b_{ij}(\theta) = \sum_{I_{ij}\setminus \{0\}} \ee^{2\pi\ii k\cdot \theta}$. Then $h_{ij}(\theta) = \epsilon_{ij}+b_{ij}(\theta)$. Condition \eqref{e:locsym} entails that the column of each block $D_{st}(\theta)$ has a fixed dependence on $\theta$. Indeed, $I_{ij}\setminus \{0\} = I_{\phi(i),j}\setminus\{0\} = I_{i+1,j}\setminus \{0\}$ by the cycle structure, where $i+1 := k_{s-1}+1$ if $i=k_s$. So $b_{ij}(\theta) = b_{i+1,j}(\theta)$ as asserted, for $i=k_{s-1}+1,\dots,k_s$ and $j=k_{t-1}+1,\dots,k_t$. The same holds for $D_{ts}(\theta)$, so we get $b_{ji}(\theta) = b_{j+1,i}(\theta)$ for $j=k_{t-1}+1,\dots,k_t$ and $i = k_{s-1}+1,\dots,k_s$. But $h_{ji} = \overline{h_{ij}}$, so $b_{ji} = \overline{b_{ij}}$. Thus, $b_{ij} = \overline{b_{ji}} = \overline{b_{j+1,i}} = b_{i,j+1}$. This shows that the rows have the same $\theta$ dependence. Thus, all entries of $D_{st}(\theta)$ take the form $b_{i_0j_0}(\theta)+\epsilon_{ij}$ for some fixed $i_0,j_0$.

It remains to show that $C_{st} = (\epsilon_{ij})$ is circulant. Condition \eqref{e:locsym} says that $\epsilon_{ij} = \epsilon_{i+1,j+1}$ in each block, that is, diagonals are constant and we have a Toeplitz structure. But $C_{st}$ is in fact circulant using $\epsilon_{i,k_t} = \epsilon_{i+1,k_{t-1}+1}$ and $\epsilon_{k_s,j} = \epsilon_{k_{s-1}+1,j+1}$.
\end{proof}

\begin{thm}\label{thm:symflat}
Under the assumptions of Lemma~\ref{lem:athetablocs}, $\Gamma$ has at least $\nu-r$ single-cell flat bands.
\end{thm}
\begin{proof}
Let $A(\theta) = (D_{ij}(\theta))_{i,j\le r}$ be in the block structure provided by Lemma~\ref{lem:athetablocs}. Consider the partition $\{X_1,\dots,X_r\}$ of $V_f$ induced by $\phi$, where $X_j$ contains the elements of the $j$-th cycle. It follows from Lemma~\ref{lem:athetablocs} that
\begin{equation}\label{e:cols}
\mathrm{Col}_k(A(\theta)) = \sum_{i=1}^r b_{ik}(\theta)\mathbf{1}_{X_i} + C_k \,,
\end{equation}
where $C_k$ is a column vector independent of $\theta$ and $b_{ik}=b_{ik'}$ if $k,k'$ are in the same block.

Let $Y_i = \frac{1}{\sqrt{|X_i|}}\mathbf{1}_{X_i}$ for $i=1,\dots,r$. Clearly the $Y_i$'s are orthonormal. Complete this with $Y_i$, $i=r+1,\dots,\nu$, as an orthonormal basis of $\ell^2(V_f)$. We now write $A(\theta)$ in the basis $\{Y_i\}_{i=1}^\nu$, in other words, the matrix elements become $a_{ij} = \langle Y_i, A(\theta) Y_j\rangle$. 

Let $W$ be the subspace generated by $Y_1,\dots,Y_r$. Then $W$ is invariant under $A(\theta)$. In fact, if $v\in X_i$, then $(A(\theta)\mathbf{1}_{X_j})(v) = \sum_{p\in X_j} A(\theta)_{v,p}$ is the $v$-th row sum of the block $D_{ij}(\theta)$, which is equal to some $d_{ij}(\theta)$ independent of $v$ since $D_{ij}(\theta)$ is circulant.
Thus, $A(\theta)\mathbf{1}_{X_j} = \sum_{i=1}^r d_{ij}(\theta)\mathbf{1}_{X_i}\in W$ as required.

It follows that $W^\bot$ is invariant as well. So $\langle Y_i,A(\theta)Y_j\rangle = 0$ for ($i\le r\wedge j>r$) or ($i>r\wedge j\le r$) i.e. $A(\theta)$ takes a block-diagonal form of sizes $r$ and $\nu-r$, respectively in the basis $\{Y_k\}_{k\le\nu}$ and $\sigma(A(\theta))$ is the union of the respective spectra. The eigenvectors $\phi_p$ of the submatrices are also eigenvectors $\tilde{\phi}_p$ for $A(\theta)$ if we extend $\phi_p$ by zero.

Since $Mf = \sum_{k=1}^m f(k)\mathrm{Col}_k(M)$ for an $m\times m$ matrix $M$, then for $i,j>r$ we have by \eqref{e:cols}, $\langle Y_i, A(\theta)Y_j\rangle = \sum_{k=1}^\nu Y_j(k) \langle Y_i, \mathrm{Col}_k(A(\theta))\rangle = \sum_{k=1}^\nu Y_j(k)\langle Y_i,C_k\rangle = \langle Y_i, CY_j\rangle$.
As the submatrix $(\langle Y_i,CY_j\rangle)_{i,j>r}$ is independent of $\theta$, so are its eigenvalues and eigenvectors $\{\mu_p,\phi_p\}_{p=1}^{\nu-r}$. Now the matrix $A(\theta)$ in the basis $\{Y_k\}$ is $\tilde{A}(\theta) = Y^\ast A(\theta)Y$, where $Y$ is the unitary matrix of columns $Y_k$. We showed that $\tilde{A}(\theta)\tilde{\phi}_p=\mu_p\tilde{\phi}_p$, so $A(\theta)Y \tilde{\phi}_p = \mu_p Y \tilde{\phi}_p$ for all $\theta$, and $Y \tilde{\phi}_p=\sum_{k>r}\tilde{\phi}_p(k)Y_k$ is independent of $\theta$. So $\mu_p$ is of single-cell type.
\end{proof}

\begin{exa}\label{exa:pyro}
The graph in Figure~\ref{fig:pyro2} has $\phi=(12)(34)$, $Y_1=\frac{1}{\sqrt{2}}(1,1,0,0)$, $Y_2=\frac{1}{\sqrt{2}}(0,0,1,1)$, $Y_3=\frac{1}{\sqrt{2}}(1,-1,0,0)$, $Y_4=\frac{1}{\sqrt{2}}(0,0,1,-1)$, so 
\begin{scriptsize}
\[
A(\theta) = \begin{pmatrix} 0&1&\ee^{-2\pi\ii\theta}&1+\ee^{-2\pi\ii\theta}\\ 1&0&1+\ee^{-2\pi\ii\theta}&\ee^{-2\pi\ii\theta}\\ \ee^{2\pi\ii\theta}&1+\ee^{2\pi\ii\theta}&0&1\\ 1+\ee^{2\pi\ii\theta}&\ee^{2\pi\ii\theta}&1&0\end{pmatrix}, \quad \tilde{A}(\theta) = \begin{pmatrix} 1&1+2\ee^{-2\pi\ii\theta}&0&0\\ 1+2\ee^{2\pi\ii\theta}&1&0&0\\ 0&0&-1&-1\\ 0&0&-1&-1\end{pmatrix}
\]
\end{scriptsize}
and we get the flat bands $\{-2,0\}$. The eigenvectors are given by $Y\tilde{\phi}_1 = Y_3+ Y_4$, $Y\tilde{\phi}_2=Y_3-Y_4$, for $\phi_1 = \binom{1}{1}$, $\phi_2=\binom{1}{-1}$. The estimate $\nu-r=2$ is exact here.
\end{exa}

\begin{exa}
The graph corresponding to the Floquet matrix at the end of Example~\ref{exa:singlenu3} has a single-cell flat band $\lambda=0$ which is not induced by any symmetry discussed in this section. In fact, since we showed it does not satisfy the neighborhood condition, this excludes the permutations of two vertices fixing the third. The remaining permutation would be $(123)$, but this is also not valid either since $A(\theta)$ would then have the form of a single block $h(\theta)J + C$ by Lemma~\ref{lem:athetablocs}, with $J$ all one, which is not the case here.
\end{exa}

\begin{rem}\label{rem:epsicirc}
Condition \eqref{e:locsym} on $\epsilon_{ij}$ means that $\phi:V_f\to V_f$ is a graph automorphism. We saw in Lemma~\ref{lem:athetablocs} that this implies the blocks are circulant. However in Theorem~\ref{thm:symflat}, we only used the weaker property that each block has a constant row sum. In particular, we can recover Proposition~\ref{prp:reggen} by considering $\phi$ to be cyclic on $G_F$ and fixing $o$, so that $\nu=|G_F|+1$ and $r=2$. In general, keeping the condition on $I_{ij}\setminus\{0\}$, for $\epsilon_{ij}$ we only need the cycles of $\phi$ to induce an \emph{equitable partition} on $V_f$. That is, each cell forms a regular subgraph ($C_{ss}$ have constant row sum) and all vertices in a given cell $s$ have the same number of edges to the cell $t$ (so $C_{st}$ have constant row sum). There can be no edges at all, whether within the cell or between different cells.
\end{rem}

Compared to \cite{RMS}, our cells can be of arbitrary sizes; we do not need $\phi$ to be ``basic of order $k$''. Also, \cite{RMS} need some sites in $V_f$ to be fixed by $\phi$, as the corresponding eigenvectors vanish on them, and this allows them to connect the copies of $V_f$ using these points \cite[eq. (19) and \S III.B]{RMS}. Our argument does not need any vertex to be fixed by $\phi$.

We saw in Lemma~\ref{lem:implies2} that the permutation invariance in this section generalizes the concept of having two vertices sharing the same neighbors, a criterion we had before. A different generalization is to continue to look at only two vertices, which can now be far apart, but have the same number of walks to a given subset of $V_f$. This idea is explored in \cite{MRPS} and is shown to yield flat bands as well.

\section{Destroying flat bands}\label{sec:destro}

The aim of this section is to take small steps towards answering Problem 2, that is, showing that flat bands are rare and disappear under generic perturbations. 

\begin{lem}\label{thm:enfin}
The set $\cP_\Gamma = \{Q\in \mathbb{R}^\nu : \mathcal{A}_\Gamma+Q \text{ has a flat band}\}$ is semialgebraic.
\end{lem}
\begin{proof}
Let $p(z;Q,\lambda) := \det(A(z)+Q-\lambda I)$. We have by Lemma~\ref{lem:triv},
\begin{align*}
\mathcal{P}_\Gamma &= \{Q\in \R^\nu : \exists \lambda\in \R,\ p(z;Q,\lambda)=0\quad \forall z\neq 0\}\\
&= \pi_1\{(Q,\lambda)\in \R^{\nu+1} : p(z;Q,\lambda)=0\quad \forall z\neq 0\},
\end{align*}
where $\pi_1$ is the projection $\pi_1(Q,\lambda)=Q$.

Now $p(z;Q,\lambda)$ is a Laurent polynomial in $z$ and it is also polynomial in $Q_1,\dots,Q_\nu,\lambda$. We expand $p(z;Q,\lambda) = \sum_{r\in \Lambda} p_r(Q,\lambda) z^r$, where $\Lambda$ is a finite subset of $\Z^d$ determined by $\nu$, $d$ and $I_{ij}$, $i,j\le \nu$. Then the Laurent polynomial $p(z;Q,\lambda)$ vanishes for all $z$ iff each coefficient of $z^r$ vanishes. Each $p_r(Q,\lambda)$ is a polynomial in $(Q_1,\dots,Q_\nu,\lambda)$. Hence,
\[
\mathcal{P}_\Gamma = \pi_1 \{(Q,\lambda)\in \R^{\nu+1}: p_r(Q,\lambda)=0\quad \forall r\in \Lambda\}
\]
is the projection of an algebraic variety and is thus semialgebraic.
\end{proof}

We now observe that 
\begin{enumerate}
\item There exists some $r\neq 0$ in $\Lambda$ such that $p_r(Q,\lambda)$ is not the trivial polynomial. Indeed, if $p_r(Q,\lambda)$ was trivial for all $r\neq 0$, we would have $p(z;Q,\lambda) = p_0(Q,\lambda)$ for all $z$, $Q$, $\lambda$. This would mean that the characteristic polynomial of $A(z)+Q$ is independent of $z$, and so are its roots, which would imply that all bands of $\mathcal{A}_\Gamma+Q$ are flat, contradicting Theorem~\ref{thm:notflat}.
\item The polynomial $p_0(Q,\lambda)$ is nontrivial and distinct from the polynomials $p_r(Q,\lambda)$ for $r\neq 0$. In fact, $p_0(Q,\lambda)$ contains the term $(-1)^\nu \lambda^\nu$ which does not appear in any other $p_r(Q,\lambda)$, since $p(z;Q,\lambda) = (-1)^\nu\lambda^\nu + q(z;Q,\lambda)$ with $q(z;Q,\lambda)$ a polynomial of lower order in $\lambda$
\end{enumerate}

It follows from (1--2) that $\{(Q,\lambda)\in \R^{\nu+1}:p_r(Q,\lambda)=0 \quad \forall r\}$ is given by the vanishing of at least two distinct nontrivial polynomials. We would like to conclude that the algebraic variety has dimension at most $\nu-1$, and so does its projection $\cP_\Gamma$. The problem is that $p_0$ and $p_r$ might share a common factor $q$, i.e. $p_r=qf$ and $p_0=qg$ for some polynomials $q,f,g$. In that case, $\{(Q,\lambda):p_0(Q,\lambda)=p_r(Q,\lambda)=0\} \supseteq \{(Q,\lambda):q(Q,\lambda)=0\}$ could have dimension $\nu$.
Let us show this does not happen for $\nu=2$.

\begin{lem}\label{lem:nu=2}
If $\nu=2$, then $\cP_\Gamma$ has dimension at most one.
\end{lem}
\begin{proof}
Here
\begin{multline*}
p(z;Q,\lambda) = \sum_{k\in I_{11}}\sum_{k'\in I_{22}} z^{k+k'} + (Q_1-\lambda) \sum_{k'\in I_{22}} z^{k'} + (Q_2-\lambda)\sum_{k\in I_{11}} z^{k} \\
+ (Q_1-\lambda)(Q_2-\lambda) - \sum_{p,p'\in I_{12}} z^{p-p'}\,,
\end{multline*}
so  $p_0(Q,\lambda) = n +(Q_1-\lambda)(Q_2-\lambda)$ for $n=|I_{11}\cap I_{22}| - |I_{12}|$. The nontrivial $p_r(Q,\lambda)$ must have the form $p_r(Q,\lambda) = m + \delta_1(Q_1-\lambda) + \delta_2(Q_2-\lambda)$ for some $m\in \Z$ and $\delta_i\in \{0,1\}$, the coefficient $m$ coming from the double sums.

If $\delta_1=\delta_2=0$, then either $p_r=m\neq 0$ never vanishes and so $\cP_\Gamma=\emptyset$, or $m=0$ and $p_r$ is trivial, contradicting the assumption. So let $\delta_1+\delta_2\neq 0$. Then $p_r=0$ implies $\lambda = \frac{m+\delta_1Q_1 + \delta_2 Q_2}{\delta_1+\delta_2}$. Substituting into $p_0$, we get $p_0(Q,\lambda) = n+\frac{(\delta_2(Q_1-Q_2)-m)(\delta_1(Q_2-Q_1)-m)}{(\delta_1+\delta_2)^2}$. This reduces to $p_0(Q,\lambda) = n-\frac{(Q_1-Q_2)^2-m^2}{4}$ or $p_0(Q,\lambda) = n+m^2\pm m(Q_1-Q_2)$. Either way, for $p_0$ to vanish, we must have $Q_2 = Q_1 + c$ for some $c\in \R$. This proves the result.
\end{proof}

\begin{rem}
It is worthwhile to note that if there exists a coefficient $p_r(Q,\lambda) = m\neq 0$ independent of $(Q,\lambda)$, then $\cP_\Gamma=\emptyset$. This was mentioned in the previous proof, and holds for any $\nu$. A simple example is the honeycomb lattice $H(z) = \begin{pmatrix} Q_1& 1+z_1^{-1}+z_2^{-1}\\ 1+z_1+z_2& Q_2\end{pmatrix}$.
\end{rem}

Let us turn back to general $\nu>1$ and prove a couple of results.

\begin{lem}\label{lem:pgclo}
For any $d$, $\nu>1$, the set $\cP_{\Gamma}$ is closed.
\end{lem}
This does not follow from Lemma~\ref{thm:enfin} as the projection of an algebraic variety may not be closed, e.g. $\pi_1\{(x,y)\in \R^2:xy-1=0\} = \R\setminus\{0\}$.
\begin{proof}
Let $\{Q^{(n)}\}\subset \cP_\Gamma$, $Q^{(n)}\to Q\in \R^\nu$. Then for each $n$, $\cA_\Gamma+Q^{(n)}$ has a flat band. So there exists $\lambda_n=\lambda(\Gamma,Q^{(n)})$ such that $\det(A(\theta)+Q^{(n)}-\lambda_n)=0$ for all $\theta\in \T_\ast^d$. 

As $Q^{(n)}$ converges in $\R^\nu$, it is bounded, $\|Q^{(n)}\|_\infty \le c$, so $\cA_\Gamma+Q^{(n)}$ is uniformly bounded, in particular $\{\lambda_n\}$ is a bounded sequence. Let $(\lambda_{n_k})$ be a subsequence converging to some $\lambda_\ast$. Since $Q^{(n_k)}\to Q$ and $\lambda_{n_k}\to \lambda_\ast$, we have for any fixed $\theta$: $\det(A(\theta)+Q^{(n_k)}-\lambda_{n_k})\to \det(A(\theta)+Q-\lambda_\ast)$, because the determinant is a polynomial in $Q^{(n_k)}$ and $\lambda_{n_k}$. But we know $\det(A(\theta)+Q^{(n_k)}-\lambda_{n_k})=0$ for all $n_k$ and all $\theta$. So for any $\theta\in \T_\ast^d$, we have $\det(A(\theta)+Q-\lambda_\ast)=0$. This shows that $\cA_\Gamma+Q$ has a flat band, i.e. $Q\in \cP_{\Gamma}$.
\end{proof}


\begin{prp}\label{prp:pertukorsa}
For any $d$, $\nu>1$, there exists $\Gamma$ such that $\mathcal{P}_\Gamma \supseteq \{Q\in \R^\nu: Q_1=Q_2\}$, and thus $\dim \cP_\Gamma\ge \nu-1$.
\end{prp}


\begin{proof}
Take $\Gamma$ with $\cN_{v_1}\setminus\{v_2\} = \cN_{v_2}\setminus\{v_1\}$ and add $Q=(Q_2,Q_2,Q_3,\dots,Q_\nu)$ with $Q_i$ arbitrary for $i\ge 2$. Then \eqref{e:hthe} takes the form $H(\theta) = \begin{pmatrix} h_{11}(\theta)+Q_2 & h_{11}(\theta)+\epsilon & g_\theta \\ h_{11}(\theta)+\epsilon& h_{11}(\theta)+Q_2& g_\theta\\ g^\ast_\theta& g^\ast_\theta& H'(\theta)\end{pmatrix}$ for $g_\theta=(h_{13}(\theta),\dots,h_{1\nu}(\theta))$. Clearly $(1,-1,0,\dots,0)^\intercal$ is an eigenvector with eigenvalue $\lambda=Q_2-\epsilon$, for all $\theta$.
\end{proof}


It was conjectured in \cite[p. 590-591]{KorSa} that if we perturb $\cA_\Gamma$ by a potential $Q$ such that $Q_1<Q_2<\dots<Q_\nu$, then $\cA_\Gamma+Q$ has no flat bands. This is not true because such potentials can create flat bands if $\cA_{\Gamma}$ has no flat bands, as the example in Figure~\ref{fig:flbmd}, taken from \cite{FLBMD}, shows. Here $H(\theta) = \begin{pmatrix} -1& 1+\ee^{-2\pi\ii\theta}\\ 1+\ee^{2\pi\ii\theta}& 2\cos 2\pi\theta\end{pmatrix}$.
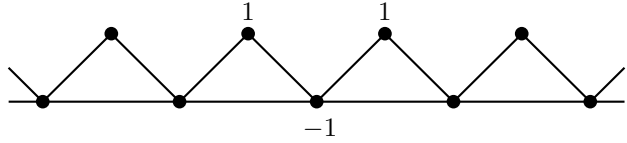
\begin{figure}[h!]
\begin{center}
\setlength{\unitlength}{0.9cm}
\thicklines
\begin{picture}(1.3,1.3)(-0.7,-0.3)
	 \put(-5,0){\line(1,0){9}}
	 \put(-4.5,0){\line(1,1){1}}
         \put(-2.5,0){\line(-1,1){1}}
	 \put(-2.5,0){\line(1,1){1}}
         \put(-0.5,0){\line(-1,1){1}}
         \put(-0.5,0){\line(1,1){1}}
         \put(1.5,0){\line(-1,1){1}}
         \put(1.5,0){\line(1,1){1}}
         \put(3.5,0){\line(-1,1){1}}
         \put(3.5,0){\line(1,1){0.5}}
         \put(-4.5,0){\line(-1,1){0.5}}
         \put(-3.5,1){\circle*{.2}}
	 \put(-4.5,0){\circle*{.2}}
	 \put(-2.5,0){\circle*{.2}}
	 \put(-1.5,1){\circle*{.2}}
	 \put(-0.5,0){\circle*{.2}}
	 \put(0.5,1){\circle*{.2}}
	 \put(1.5,0){\circle*{.2}}
	 \put(2.5,1){\circle*{.2}}
	 \put(3.5,0){\circle*{.2}}
	 \put(-0.7,-0.5){\small{$-1$}}
	 \put(0.4,1.2){\small{$1$}}
	 \put(-1.6,1.2){\small{$1$}}
\end{picture}
\caption{Here $\nu=2$, $\cA_{\Gamma}$ has no flat bands but $\cA_{\Gamma}+Q$ has the flat band $\lambda=-2$ if $Q=-1$ on the upper row and $Q=0$ on the lower row.}\label{fig:flbmd}
\end{center}
\end{figure} 

Still, this conjecture is almost true for $\nu=2$: we showed in Lemma~\ref{lem:nu=2} that for any $\Gamma$, there exists $c$ such that $\cP_\Gamma \subseteq \{Q\in \R^2:Q_1=Q_2+c\}$. But it is \emph{wrong} for higher $\nu$. For example, for $A(\theta) = \begin{pmatrix} c&c&c\\ c&c&c\\ c&c&0\end{pmatrix}$ with $c = 2\cos 2\pi\theta$, we have flat bands for all perturbations of the form $(Q_1,Q_2,\frac{Q_1+Q_2}{2})$. The same holds for the Laplacian $D - A(\theta)$, with $D=\mathrm{Diag}(6,6,4)$ the degree matrix, where flat bands occur for all $(Q_1,Q_2,\frac{Q_1+Q_2}{2}+2)$.

\appendix

\section{Eigenvalues and eigenvectors of small graphs}\label{app:45}

We give here the details for \eqref{e:f45}. The tables were produced by Wolfram Alpha.

\subsection{Case of 4 vertices}
Given $\cF_3^s$, by Theorem~\ref{thm:moredetailsonfnu+1}, the only $V_f$ that can produce more bands for $\cF_4^s$ are those in Figure~\ref{fig:4graphs}.

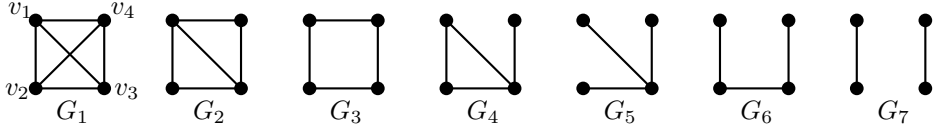
\begin{figure}[h!]
\begin{center}
\setlength{\unitlength}{0.9cm}
\thicklines
\begin{picture}(2.2,1.4)(-1.5,-1.3)
	 \put(-7,-1){\line(1,0){1}}
	 \put(-7,0){\line(1,0){1}}
	 \put(-7,-1){\line(0,1){1}}
	 \put(-6,-1){\line(0,1){1}}
	 \put(-7,-1){\line(1,1){1}}
	 \put(-7,0){\line(1,-1){1}}
	    \put(-5,-1){\line(1,0){1}}
	 \put(-5,0){\line(1,0){1}}
	 \put(-5,-1){\line(0,1){1}}
	 \put(-4,-1){\line(0,1){1}}
	 \put(-5,0){\line(1,-1){1}}
	    \put(-3,-1){\line(1,0){1}}
	 \put(-3,0){\line(1,0){1}}
	 \put(-3,-1){\line(0,1){1}}
	 \put(-2,-1){\line(0,1){1}}
	    \put(-1,-1){\line(1,0){1}}
	 \put(-1,-1){\line(0,1){1}}
	 \put(0,-1){\line(0,1){1}}
	 \put(-1,0){\line(1,-1){1}}
	    \put(1,-1){\line(1,0){1}}
	 \put(2,-1){\line(0,1){1}}
	 \put(1,0){\line(1,-1){1}}
	   \put(3,-1){\line(1,0){1}}
	 \put(4,-1){\line(0,1){1}}
	 \put(3,-1){\line(0,1){1}}
	    \put(5,-1){\line(0,1){1}}
	    \put(6,-1){\line(0,1){1}}
	    \put(-7,-1){\circle*{.2}}
	 \put(-7,0){\circle*{.2}}
	 \put(-6,-1){\circle*{.2}}   
	 \put(-6,0){\circle*{.2}}
	    \put(-5,-1){\circle*{.2}}
	 \put(-5,0){\circle*{.2}}
	 \put(-4,-1){\circle*{.2}}   
	 \put(-4,0){\circle*{.2}}
	     \put(-3,-1){\circle*{.2}}
	 \put(-3,0){\circle*{.2}}
	 \put(-2,-1){\circle*{.2}}   
	 \put(-2,0){\circle*{.2}}
	     \put(-1,-1){\circle*{.2}}
	 \put(-1,0){\circle*{.2}}
	 \put(0,-1){\circle*{.2}}   
	 \put(0,0){\circle*{.2}}
	    \put(1,-1){\circle*{.2}}
	 \put(1,0){\circle*{.2}}
	 \put(2,-1){\circle*{.2}}   
	 \put(2,0){\circle*{.2}}
	    \put(3,-1){\circle*{.2}}
	 \put(3,0){\circle*{.2}}
	 \put(4,-1){\circle*{.2}}   
	 \put(4,0){\circle*{.2}}
	    \put(5,-1){\circle*{.2}}
	 \put(5,0){\circle*{.2}}
	 \put(6,-1){\circle*{.2}}   
	 \put(6,0){\circle*{.2}}
	   \put(-7.4,0.1){\small{$v_1$}}
	   \put(-7.45,-1.1){\small{$v_2$}}
	   \put(-5.85,-1.1){\small{$v_3$}}
	   \put(-5.9,0.1){\small{$v_4$}}
	    \put(-6.7,-1.45){\small{$G_1$}}
	    \put(-4.7,-1.45){\small{$G_2$}}
	    \put(-2.7,-1.45){\small{$G_3$}}
	    \put(-0.7,-1.45){\small{$G_4$}}
	    \put(1.3,-1.45){\small{$G_5$}}
	    \put(3.3,-1.45){\small{$G_6$}}
	    \put(5.3,-1.45){\small{$G_7$}}
\end{picture}
\caption{All connected graphs and the relevant disconnected graph.}\label{fig:4graphs}
\end{center}
\end{figure}

We compute the eigenvalues and eigenvectors of these graphs in the following table.

\begin{small}
\begin{alignat*}{3}
& G_1   \quad && 3,-1,-1,-1            \quad && (1,1,1,1),(-1,0,0,1),(-1,0,1,0),(-1,1,0,0)\\
& G_2   \quad && \alpha,0,-1,\alpha'    \quad && \textstyle{\big(\frac{\alpha}{2},1,\frac{\alpha}{2},1\big),(0,-1,0,1),(-1,0,1,0),\big(\frac{\alpha'}{2},1,\frac{\alpha'}{2},1\big)}\\
& G_3   \quad && 2,0,0,-2 \quad && (1,1,1,1),(0,-1,0,1),(-1,0,1,0),(-1,1,-1,1)\\
& G_4 \quad && \beta_1,\beta_2,-1,\beta_3     \quad && (\kappa_1,\kappa_1,\beta_1,1),(\kappa_2,\kappa_2,\beta_2,1),(-1,1,0,0),(\kappa_3,\kappa_3,\beta_3,1)\\
& G_5   \quad && \sqrt{3},0,0,-\sqrt{3}            \quad && (1,1,\sqrt{3},1),(-1,0,0,1),(-1,1,0,0),(1,1,-\sqrt{3},1) \\
& G_6   \quad && \sigma_1,\sigma_2,\sigma_3,\sigma_4            \quad && (1,\sigma_1,\sigma_1,1), (-1,\sigma_3,\sigma_2,1),(1,\sigma_3,\sigma_3,1),(-1,\sigma_1,\sigma_4,1)\\
& G_7   \quad && 1,1,-1,-1            \quad && (0,0,1,1),(1,1,0,0),(0,0,-1,1),(-1,1,0,0)
\end{alignat*}
\end{small}

Here $\alpha = \frac{1+\sqrt{17}}{2}$, $\alpha'=\frac{1-\sqrt{17}}{2}$, $\beta_i$ are the roots of $\beta^3-\beta^2-3\beta+1=0$, with $\beta_1\approx 2.17009$, $\beta_2\approx 0.31111$, $\beta_3\approx -1.48119$ and $\kappa_1\approx 1.85464$, $\kappa_2\approx -0.451606$, $\kappa_3\approx 0.596968$. Finally $\sigma_1=\frac{1+\sqrt{5}}{2}$, $\sigma_2 = \frac{-1+\sqrt{5}}{2}$, $\sigma_3 = \frac{1-\sqrt{5}}{2}$, $\sigma_4=\frac{-1-\sqrt{5}}{2}$.

We know $\cF_3^s=\{0,-1\}\subseteq\cF_4^s$. By looking at the eigenvectors, we may deduce from Theorem~\ref{thm:singlecellcriter} that $-2,\sigma_2,\sigma_4\in \cF_4^s$ (because $\sigma_3=-\sigma_2$ and $\sigma_4=-\sigma_1$). We can also see this more explicitly from Figure~\ref{fig:irratflat4}, which shows that $\sigma_2=\frac{-1+\sqrt{5}}{2}$, $\sigma_4=\frac{-1-\sqrt{5}}{2}$ can appear together, and Example~\ref{exa:pyro}, which showed that $-2$ arises in $\cF_4^s$ from the $1d$ pyrochlore.

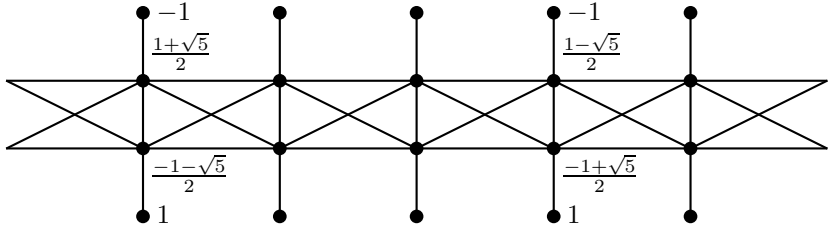
\begin{figure}[h!]
\begin{center}
\setlength{\unitlength}{0.9cm}
\thicklines
\begin{picture}(1.3,2.7)(-1.1,-1.95)
         \put(-7,0){\line(1,0){12}}
	 \put(-7,-1){\line(1,0){12}}
	 \put(-5,-1){\line(2,1){2}}
	 \put(-5,0){\line(2,-1){2}}
	 \put(-3,-1){\line(0,1){1}}
	 \put(-3,-1){\line(2,1){2}}
	 \put(-3,0){\line(2,-1){2}}
	 \put(-1,-1){\line(0,1){1}}
	 \put(-1,-1){\line(2,1){2}}
	 \put(-1,0){\line(2,-1){2}}
	 \put(1,-1){\line(0,1){1}}
		\put(1,-1){\line(2,1){2}}
	 \put(1,0){\line(2,-1){2}}
	 \put(3,-1){\line(0,1){1}}
		\put(3,-1){\line(2,1){2}}
	 \put(3,0){\line(2,-1){2}}
	 \put(-1,-1){\circle*{.2}}
	 \put(-1,0){\circle*{.2}}
	 \put(-3,-1){\circle*{.2}}
	 \put(-3,0){\circle*{.2}}
	 \put(1,-1){\circle*{.2}}
	 \put(1,0){\circle*{.2}}
	 \put(3,-1){\circle*{.2}}
	 \put(3,0){\circle*{.2}}
	\put(-3,-2){\line(0,1){1}}
	\put(-1,-2){\line(0,1){1}}
	\put(1,-2){\line(0,1){1}}
	\put(3,-2){\line(0,1){1}}
	\put(-3,0){\line(0,1){1}}
	\put(-1,0){\line(0,1){1}}
	\put(1,0){\line(0,1){1}}
	\put(3,0){\line(0,1){1}}
	\put(-5,-2){\line(0,1){1}}
	\put(-5,-1){\line(0,1){1}}
	\put(-5,0){\line(0,1){1}}
	\put(-5,0){\line(-2,-1){2}}
	\put(-5,-1){\line(-2,1){2}}
	\put(-5,-1){\circle*{.2}}
	\put(-5,0){\circle*{.2}}
	\put(-5,-2){\circle*{.2}}
	\put(-3,-2){\circle*{.2}}
	\put(-1,-2){\circle*{.2}}
	\put(1,-2){\circle*{.2}}
	\put(3,-2){\circle*{.2}}
	\put(-5,1){\circle*{.2}}
	\put(-3,1){\circle*{.2}}
	\put(-1,1){\circle*{.2}}
	\put(1,1){\circle*{.2}}
	\put(3,1){\circle*{.2}}
	\put(1.2,0.9){\small{$-1$}}
	\put(1.1,0.3){\small{$\frac{1-\sqrt{5}}{2}$}}
	\put(1.1,-1.5){\small{$\frac{-1+\sqrt{5}}{2}$}}
	 \put(1.2,-2.1){\small{$1$}}
	  \put(-4.8,0.9){\small{$-1$}}
	\put(-4.9,0.3){\small{$\frac{1+\sqrt{5}}{2}$}}
	\put(-4.9,-1.5){\small{$\frac{-1-\sqrt{5}}{2}$}}
	 \put(-4.8,-2.1){\small{$1$}}
	\end{picture}
\caption{Flat bands $\sigma_4=\frac{-1-\sqrt{5}}{2}$ (left) and $\sigma_2=\frac{\sqrt{5}-1}{2}$ (right).}\label{fig:irratflat4}
\end{center}
\end{figure}

If we remove the middle vertical links in Figure~\ref{fig:irratflat4} between $v_2$ and $v_3$, we obtain the flat bands $-1$ and $1$, corresponding to $(-1,1,-1,1)^\intercal$ and $(-1,-1,1,1)^\intercal$. Here $V_f=G_7$.

We now exclude the remaining eigenvalues. In view of Lemma~\ref{lem:gensincell}, we know we can remove the top eigenvalues of connected graphs $G_k$. Next, we exclude $\alpha'$ as there are no $\delta_i\in \{0,1\}$ such that $\frac{\alpha'}{2}(\delta_1+\delta_3)=-(\delta_2+\delta_4)$ except all $\delta_j=0$, the LHS being irrational. Similarly we see that $-\sqrt{3},\sigma_3\notin \cF_4^s$. From the values of $\beta_i,\kappa_i$, we see there are no $\delta_k\in\{0,1\}$ such that $\kappa_j(\delta_1+\delta_2)+\delta_3\beta_j+\delta_4=0$ except all $\delta_j=0$,  so $\beta_i\notin\cF_4^s$.

\subsection{Case of 5 vertices}\label{app:5ver}
By Corollary~\ref{cor:increaflat} and the case $\nu=4$, all the listed values are in $\cF_5^s$. We provide some explicit graphs in Figures~\ref{fig:deco2}  and \ref{fig:deco3} for illustration.

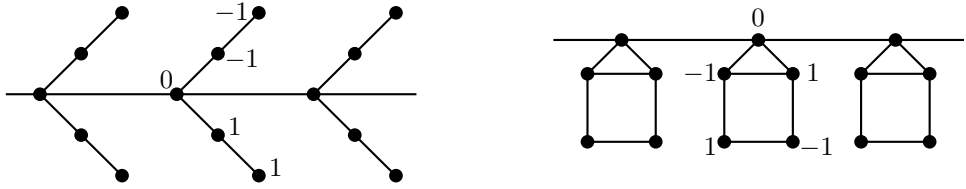
\begin{figure}[h!]
\begin{center}
\setlength{\unitlength}{0.9cm}
\thicklines
\begin{picture}(1.3,2)(-0.7,-1.8)
           \put(-7,-0.8){\line(1,0){6}}
           \put(-6.5,-0.8){\line(1,1){1.2}}
           \put(-6.5,-0.8){\line(1,-1){1.2}}
           \put(-4.5,-0.8){\line(1,1){1.2}}
           \put(-4.5,-0.8){\line(1,-1){1.2}}
           \put(-2.5,-0.8){\line(1,1){1.2}}
           \put(-2.5,-0.8){\line(1,-1){1.2}}
       \put(-6.5,-0.8){\circle*{.2}}
       \put(-4.5,-0.8){\circle*{.2}}
       \put(-2.5,-0.8){\circle*{.2}}
           \put(-5.9,-0.2){\circle*{.2}}
           \put(-3.9,-0.2){\circle*{.2}}
           \put(-1.9,-0.2){\circle*{.2}}
       \put(-5.3,0.4){\circle*{.2}}
       \put(-3.3,0.4){\circle*{.2}}
       \put(-1.3,0.4){\circle*{.2}}
           \put(-5.9,-1.4){\circle*{.2}}
           \put(-3.9,-1.4){\circle*{.2}}
           \put(-1.9,-1.4){\circle*{.2}}
       \put(-5.3,-2){\circle*{.2}}
       \put(-3.3,-2){\circle*{.2}}
       \put(-1.3,-2){\circle*{.2}}
           \put(-4.75,-0.7){\small{$0$}}
           \put(-3.8,-0.4){\small{$-1$}}
           \put(-3.95,0.3){\small{$-1$}}
           \put(-3.75,-1.4){\small{$1$}}
           \put(-3.15,-2){\small{$1$}}
   \put(1,0){\line(1,0){6}}
	    \put(1.5,-1.5){\line(1,0){1}}
	    \put(1.5,-1.5){\line(0,1){1}}
	    \put(2.5,-1.5){\line(0,1){1}}
	    \put(1.5,-0.5){\line(1,1){0.5}}
	    \put(2.5,-0.5){\line(-1,1){0.5}}
	    \put(1.5,-0.5){\line(1,0){1}}
	    \put(3.5,-0.5){\line(1,0){1}}
	    \put(5.5,-0.5){\line(1,0){1}}
	 \put(3.5,-1.5){\line(1,0){1}}
	    \put(3.5,-1.5){\line(0,1){1}}
	    \put(4.5,-1.5){\line(0,1){1}}
	    \put(3.5,-0.5){\line(1,1){0.5}}
	    \put(4.5,-0.5){\line(-1,1){0.5}}  
	 \put(5.5,-1.5){\line(1,0){1}}
	    \put(5.5,-1.5){\line(0,1){1}}
	    \put(6.5,-1.5){\line(0,1){1}}
	    \put(5.5,-0.5){\line(1,1){0.5}}
	    \put(6.5,-0.5){\line(-1,1){0.5}}    
	\put(2,0){\circle*{.2}}
	\put(4,0){\circle*{.2}}
	\put(6,0){\circle*{.2}}
	    \put(1.5,-0.5){\circle*{.2}}
	    \put(3.5,-0.5){\circle*{.2}}
	    \put(5.5,-0.5){\circle*{.2}}
	\put(2.5,-0.5){\circle*{.2}}
	\put(4.5,-0.5){\circle*{.2}}
	\put(6.5,-0.5){\circle*{.2}}
	     \put(1.5,-1.5){\circle*{.2}}
	     \put(3.5,-1.5){\circle*{.2}}
	     \put(5.5,-1.5){\circle*{.2}}
	 \put(2.5,-1.5){\circle*{.2}}
	 \put(4.5,-1.5){\circle*{.2}}
	 \put(6.5,-1.5){\circle*{.2}}
	 \put(3.2,-1.7){\small{$1$}}
	 \put(4.6,-1.7){\small{$-1$}}
	 \put(2.9,-0.6){\small{$-1$}}
	 \put(4.7,-0.6){\small{$1$}}
	 \put(3.9,0.2){\small{$0$}}
\end{picture}
\caption{Flat bands $\lambda=1$ (left) and $\lambda=-2$ (right).}\label{fig:deco2}
\end{center}
\end{figure}

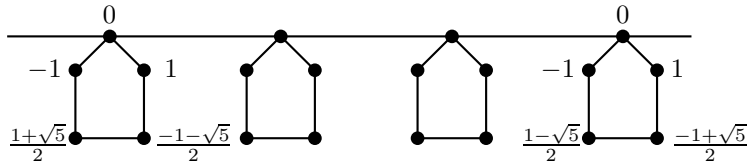
\begin{figure}[h!]
\begin{center}
\setlength{\unitlength}{0.9cm}
\thicklines
\begin{picture}(1.3,1.5)(-0.8,-1.5)
   \put(-6,0){\line(1,0){10}}
	 \put(-5,-1.5){\line(1,0){1}}
	 \put(-5,-1.5){\line(0,1){1}}
	 \put(-4,-1.5){\line(0,1){1}}
	 \put(-5,-0.5){\line(1,1){0.5}}
	 \put(-4,-0.5){\line(-1,1){0.5}}
	    \put(-2.5,-1.5){\line(1,0){1}}
	    \put(-2.5,-1.5){\line(0,1){1}}
	    \put(-1.5,-1.5){\line(0,1){1}}
	    \put(-2.5,-0.5){\line(1,1){0.5}}
	    \put(-1.5,-0.5){\line(-1,1){0.5}}
	 \put(0,-1.5){\line(1,0){1}}
	    \put(0,-1.5){\line(0,1){1}}
	    \put(1,-1.5){\line(0,1){1}}
	    \put(0,-0.5){\line(1,1){0.5}}
	    \put(1,-0.5){\line(-1,1){0.5}}  
	 \put(2.5,-1.5){\line(1,0){1}}
	    \put(2.5,-1.5){\line(0,1){1}}
	    \put(3.5,-1.5){\line(0,1){1}}
	    \put(2.5,-0.5){\line(1,1){0.5}}
	    \put(3.5,-0.5){\line(-1,1){0.5}}    
	\put(-4.5,0){\circle*{.2}}
	\put(-2,0){\circle*{.2}}
	\put(0.5,0){\circle*{.2}}
	\put(3,0){\circle*{.2}}
	    \put(-5,-0.5){\circle*{.2}}
	    \put(-2.5,-0.5){\circle*{.2}}
	    \put(0,-0.5){\circle*{.2}}
	    \put(2.5,-0.5){\circle*{.2}}
	\put(-4,-0.5){\circle*{.2}}
	\put(-1.5,-0.5){\circle*{.2}}
	\put(1,-0.5){\circle*{.2}}
	\put(3.5,-0.5){\circle*{.2}}
	     \put(-5,-1.5){\circle*{.2}}
	     \put(-2.5,-1.5){\circle*{.2}}
	     \put(0,-1.5){\circle*{.2}}
	     \put(2.5,-1.5){\circle*{.2}}
	 \put(-4,-1.5){\circle*{.2}}
	 \put(-1.5,-1.5){\circle*{.2}}
	 \put(1,-1.5){\circle*{.2}}
	 \put(3.5,-1.5){\circle*{.2}}
	 \put(-6,-1.7){\small{$\frac{1+\sqrt{5}}{2}$}}
	 \put(-3.85,-1.7){\small{$\frac{-1-\sqrt{5}}{2}$}}
	 \put(-5.7,-0.6){\small{$-1$}}
	 \put(-3.7,-0.6){\small{$1$}}
	 \put(-4.6,0.2){\small{$0$}}
	 \put(1.5,-1.7){\small{$\frac{1-\sqrt{5}}{2}$}}
	 \put(3.7,-1.7){\small{$\frac{-1+\sqrt{5}}{2}$}}
	 \put(1.8,-0.6){\small{$-1$}}
	 \put(3.7,-0.6){\small{$1$}}
	 \put(2.9,0.2){\small{$0$}}
\end{picture}
\caption{Flat bands $\lambda=\frac{-1-\sqrt{5}}{2}$ (left) and $\lambda=\frac{-1+\sqrt{5}}{2}$ (right).}\label{fig:deco3}
\end{center}
\end{figure}

We now show there are no more flat bands. By Theorem~\ref{thm:moredetailsonfnu+1} it suffices to investigate the $21$ connected graphs on $5$ vertices, listed in Figure~\ref{fig:5graphs}.
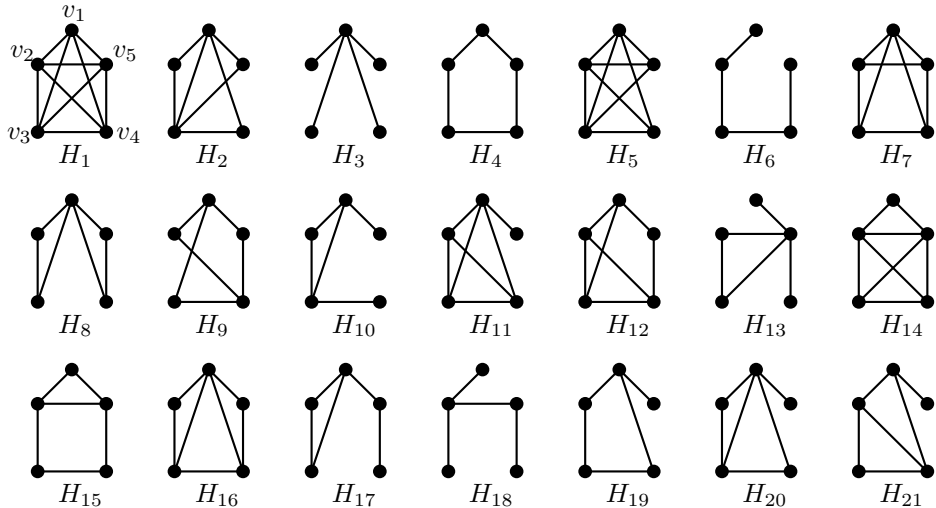
\begin{figure}[h!]
\begin{center}
\setlength{\unitlength}{0.9cm}
\thicklines
\begin{picture}(2.2,6.5)(-1.5,-6.2)
	 \put(-7,-1){\line(1,0){1}}
	 \put(-7,0){\line(1,0){1}}
	 \put(-7,-1){\line(0,1){1}}
	 \put(-6,-1){\line(0,1){1}}
	 \put(-7,-1){\line(1,1){1}}
	 \put(-7,0){\line(1,-1){1}}
	 \put(-7,0){\line(1,1){0.5}}
	 \put(-6,0){\line(-1,1){0.5}}
	 \put(-6.5,0.5){\line(-1,-3){0.5}}
	 \put(-6.5,0.5){\line(1,-3){0.5}}
	    \put(-5,-1){\line(1,0){1}}
	 \put(-5,-1){\line(0,1){1}}
	 \put(-4,0){\line(-1,-1){1}}
	 \put(-5,0){\line(1,1){0.5}}
	 \put(-4,0){\line(-1,1){0.5}}
	 \put(-4.5,0.5){\line(-1,-3){0.5}}
	 \put(-4.5,0.5){\line(1,-3){0.5}}
	 \put(-3,0){\line(1,1){0.5}}
	 \put(-2,0){\line(-1,1){0.5}}
	 \put(-2.5,0.5){\line(-1,-3){0.5}}
	 \put(-2.5,0.5){\line(1,-3){0.5}}
	      \put(-1,-1){\line(1,0){1}}
	 \put(0,-1){\line(0,1){1}}
	 \put(-1,-1){\line(0,1){1}}
	 \put(-1,0){\line(1,1){0.5}}
	 \put(0,0){\line(-1,1){0.5}}
	      \put(1,-1){\line(1,0){1}}
	 \put(1,0){\line(1,0){1}}
	 \put(1,-1){\line(0,1){1}}
	 \put(1,0){\line(1,-1){1}}
	 \put(1,0){\line(1,1){0.5}}
	 \put(2,0){\line(-1,1){0.5}}
	 \put(1,-1){\line(1,1){1}}
	 \put(1.5,0.5){\line(-1,-3){0.5}}
	 \put(1.5,0.5){\line(1,-3){0.5}}
	    \put(3,-1){\line(1,0){1}}
	 \put(4,-1){\line(0,1){1}}
	 \put(3,-1){\line(0,1){1}}
	 \put(3,0){\line(1,1){0.5}}
	    \put(5,-1){\line(1,0){1}}
	 \put(5,-1){\line(0,1){1}}
	 \put(6,-1){\line(0,1){1}}
	 \put(5,0){\line(1,0){1}}
	 \put(5,0){\line(1,1){0.5}}
	 \put(6,0){\line(-1,1){0.5}}
	 \put(5.5,0.5){\line(-1,-3){0.5}}
	 \put(5.5,0.5){\line(1,-3){0.5}}
	    \put(-7,-1){\circle*{.2}}
	 \put(-7,0){\circle*{.2}}
	 \put(-6.5,0.5){\circle*{.2}}
	 \put(-6,-1){\circle*{.2}}   
	 \put(-6,0){\circle*{.2}}
	    \put(-5,-1){\circle*{.2}}
	 \put(-5,0){\circle*{.2}}
	 \put(-4,-1){\circle*{.2}}   
	 \put(-4,0){\circle*{.2}}
	 \put(-4.5,0.5){\circle*{.2}}
	     \put(-3,-1){\circle*{.2}}
	 \put(-3,0){\circle*{.2}}
	 \put(-2,-1){\circle*{.2}}   
	 \put(-2,0){\circle*{.2}}
	 \put(-2.5,0.5){\circle*{.2}}
	     \put(-1,-1){\circle*{.2}}
	 \put(-1,0){\circle*{.2}}
	 \put(0,-1){\circle*{.2}}   
	 \put(0,0){\circle*{.2}}
	 \put(-0.5,0.5){\circle*{.2}}
	    \put(1,-1){\circle*{.2}}
	 \put(1,0){\circle*{.2}}
	 \put(2,-1){\circle*{.2}}   
	 \put(2,0){\circle*{.2}}
	 \put(1.5,0.5){\circle*{.2}}
	    \put(3,-1){\circle*{.2}}
	 \put(3,0){\circle*{.2}}
	 \put(4,-1){\circle*{.2}}   
	 \put(4,0){\circle*{.2}}
	 \put(3.5,0.5){\circle*{.2}}
	    \put(5,-1){\circle*{.2}}
	 \put(5,0){\circle*{.2}}
	 \put(6,-1){\circle*{.2}}   
	 \put(6,0){\circle*{.2}}
	 \put(5.5,0.5){\circle*{.2}}
	   \put(-6.65,0.7){\small{$v_1$}}
	   \put(-7.4,0.1){\small{$v_2$}}
	   \put(-7.45,-1.1){\small{$v_3$}}
	   \put(-5.85,-1.1){\small{$v_4$}}
	   \put(-5.9,0.1){\small{$v_5$}}
	 \put(-7,-3.5){\line(0,1){1}}
	 \put(-6,-3.5){\line(0,1){1}}
	 \put(-6.5,-2){\line(-1,-1){0.5}} 
	 \put(-6.5,-2){\line(1,-1){0.5}} 
	 \put(-6.5,-2){\line(-1,-3){0.5}}  
	 \put(-6.5,-2){\line(1,-3){0.5}} 
	    \put(-5,-3.5){\line(1,0){1}}
	    \put(-4,-3.5){\line(0,1){1}} 
	    \put(-4.5,-2){\line(-1,-1){0.5}}
	    \put(-4.5,-2){\line(1,-1){0.5}}
	    \put(-4.5,-2){\line(-1,-3){0.5}}
	    \put(-4,-3.5){\line(-1,1){1}}
	\put(-3,-3.5){\line(0,1){1}}
	\put(-3,-3.5){\line(1,0){1}}
	\put(-2.5,-2){\line(-1,-1){0.5}} 
	\put(-2.5,-2){\line(1,-1){0.5}} 
	\put(-2.5,-2){\line(-1,-3){0.5}}
	     \put(-1,-3.5){\line(0,1){1}}
	     \put(-1,-3.5){\line(1,0){1}}
	     \put(-1,-2.5){\line(1,-1){1}}
	     \put(-0.5,-2){\line(-1,-1){0.5}}   
	     \put(-0.5,-2){\line(1,-1){0.5}} 
	     \put(-0.5,-2){\line(-1,-3){0.5}} 
	     \put(-0.5,-2){\line(1,-3){0.5}} 
	\put(1,-3.5){\line(0,1){1}}
	\put(1,-3.5){\line(1,0){1}}
	\put(2,-3.5){\line(0,1){1}}
	\put(1,-2.5){\line(1,-1){1}}
	\put(1.5,-2){\line(-1,-1){0.5}}
	\put(1.5,-2){\line(1,-1){0.5}}
	\put(1.5,-2){\line(-1,-3){0.5}}
	    \put(3,-3.5){\line(0,1){1}}
	    \put(3,-2.5){\line(1,0){1}}
	    \put(3,-3.5){\line(1,1){1}}
	    \put(4,-3.5){\line(0,1){1}}
	    \put(3.5,-2){\line(1,-1){0.5}}
	\put(5,-3.5){\line(0,1){1}}
	\put(6,-3.5){\line(0,1){1}}
	\put(5,-3.5){\line(1,0){1}}
	\put(5,-2.5){\line(1,0){1}}
	\put(5,-3.5){\line(1,1){1}}
	\put(5,-2.5){\line(1,-1){1}}
	\put(5.5,-2){\line(-1,-1){0.5}}
	\put(5.5,-2){\line(1,-1){0.5}}
	     \put(-7,-6){\line(0,1){1}}
	     \put(-6,-6){\line(0,1){1}}
	     \put(-7,-6){\line(1,0){1}}
	     \put(-7,-5){\line(1,0){1}}
	     \put(-6.5,-4.5){\line(-1,-1){0.5}}
	     \put(-6.5,-4.5){\line(1,-1){0.5}}
	 \put(-5,-6){\line(0,1){1}}
	 \put(-4,-6){\line(0,1){1}}
	 \put(-5,-6){\line(1,0){1}}
	 \put(-4.5,-4.5){\line(-1,-1){0.5}}
	 \put(-4.5,-4.5){\line(1,-1){0.5}}
	 \put(-4.5,-4.5){\line(-1,-3){0.5}}
	 \put(-4.5,-4.5){\line(1,-3){0.5}}
	      \put(-3,-6){\line(0,1){1}}
	      \put(-2,-6){\line(0,1){1}}
	      \put(-2.5,-4.5){\line(-1,-1){0.5}}
	      \put(-2.5,-4.5){\line(1,-1){0.5}}
	      \put(-2.5,-4.5){\line(-1,-3){0.5}}
	 \put(-1,-6){\line(0,1){1}}
	 \put(0,-6){\line(0,1){1}}
	 \put(-1,-5){\line(1,0){1}}
	 \put(-0.5,-4.5){\line(-1,-1){0.5}}
	      \put(1,-6){\line(0,1){1}}
	      \put(1,-6){\line(1,0){1}}
	      \put(1.5,-4.5){\line(-1,-1){0.5}}
	       \put(1.5,-4.5){\line(1,-1){0.5}}
	        \put(1.5,-4.5){\line(1,-3){0.5}}
	  \put(3,-6){\line(0,1){1}}
	  \put(3,-6){\line(1,0){1}}
	  \put(3.5,-4.5){\line(-1,-1){0.5}}
	  \put(3.5,-4.5){\line(1,-1){0.5}}
	  \put(3.5,-4.5){\line(1,-3){0.5}}
	  \put(3.5,-4.5){\line(-1,-3){0.5}}
	       \put(5,-6){\line(0,1){1}} 
	       \put(5,-6){\line(1,0){1}}
	       \put(5,-5){\line(1,-1){1}}
	       \put(5.5,-4.5){\line(-1,-1){0.5}}
	       \put(5.5,-4.5){\line(1,-1){0.5}}
	       \put(5.5,-4.5){\line(1,-3){0.5}} 
	    \put(-7,-3.5){\circle*{.2}}
	 \put(-7,-2.5){\circle*{.2}}
	 \put(-6,-3.5){\circle*{.2}}   
	 \put(-6,-2.5){\circle*{.2}}
	 \put(-6.5,-2){\circle*{.2}}
	    \put(-5,-3.5){\circle*{.2}}
	 \put(-5,-2.5){\circle*{.2}}
	 \put(-4,-3.5){\circle*{.2}}   
	 \put(-4,-2.5){\circle*{.2}}
	 \put(-4.5,-2){\circle*{.2}}
	     \put(-3,-3.5){\circle*{.2}}
	 \put(-3,-2.5){\circle*{.2}}
	 \put(-2,-3.5){\circle*{.2}}   
	 \put(-2,-2.5){\circle*{.2}}
	 \put(-2.5,-2){\circle*{.2}}
	     \put(-1,-3.5){\circle*{.2}}
	 \put(-1,-2.5){\circle*{.2}}
	 \put(0,-3.5){\circle*{.2}}   
	 \put(0,-2.5){\circle*{.2}}
	 \put(-0.5,-2){\circle*{.2}}
	      \put(1,-3.5){\circle*{.2}}
	 \put(1,-2.5){\circle*{.2}}
	 \put(2,-3.5){\circle*{.2}}   
	 \put(2,-2.5){\circle*{.2}}
	 \put(1.5,-2){\circle*{.2}}
	       \put(3,-3.5){\circle*{.2}}
	 \put(3,-2.5){\circle*{.2}}
	 \put(4,-3.5){\circle*{.2}}   
	 \put(4,-2.5){\circle*{.2}}
	 \put(3.5,-2){\circle*{.2}}
	        \put(5,-3.5){\circle*{.2}}
	 \put(5,-2.5){\circle*{.2}}
	 \put(6,-3.5){\circle*{.2}}   
	 \put(6,-2.5){\circle*{.2}}
	 \put(5.5,-2){\circle*{.2}}
	    \put(-6.7,-1.45){\small{$H_1$}}
	    \put(-4.7,-1.45){\small{$H_2$}}
	    \put(-2.7,-1.45){\small{$H_3$}}
	    \put(-0.7,-1.45){\small{$H_4$}}
	    \put(1.3,-1.45){\small{$H_5$}}
	    \put(3.3,-1.45){\small{$H_6$}}
	    \put(5.3,-1.45){\small{$H_7$}}
	    \put(-6.7,-3.95){\small{$H_8$}}
	    \put(-4.7,-3.95){\small{$H_9$}}
	    \put(-2.7,-3.95){\small{$H_{10}$}}
	    \put(-0.7,-3.95){\small{$H_{11}$}}
	    \put(1.3,-3.95){\small{$H_{12}$}}
	    \put(3.3,-3.95){\small{$H_{13}$}}
	    \put(5.3,-3.95){\small{$H_{14}$}}
	\put(-7,-6){\circle*{.2}}
	 \put(-7,-5){\circle*{.2}}
	 \put(-6,-6){\circle*{.2}}   
	 \put(-6,-5){\circle*{.2}}
	 \put(-6.5,-4.5){\circle*{.2}}
	      \put(-5,-6){\circle*{.2}}
	 \put(-5,-5){\circle*{.2}}
	 \put(-4,-6){\circle*{.2}}   
	 \put(-4,-5){\circle*{.2}}
	 \put(-4.5,-4.5){\circle*{.2}}
	      \put(-3,-6){\circle*{.2}}
	 \put(-3,-5){\circle*{.2}}
	 \put(-2,-6){\circle*{.2}}   
	 \put(-2,-5){\circle*{.2}}
	 \put(-2.5,-4.5){\circle*{.2}}
	      \put(-1,-6){\circle*{.2}}
	 \put(-1,-5){\circle*{.2}}
	 \put(0,-6){\circle*{.2}}   
	 \put(0,-5){\circle*{.2}}
	 \put(-0.5,-4.5){\circle*{.2}}
	       \put(1,-6){\circle*{.2}}
	 \put(1,-5){\circle*{.2}}
	 \put(2,-6){\circle*{.2}}   
	 \put(2,-5){\circle*{.2}}
	 \put(1.5,-4.5){\circle*{.2}}
	      \put(3,-6){\circle*{.2}}
	 \put(3,-5){\circle*{.2}}
	 \put(4,-6){\circle*{.2}}   
	 \put(4,-5){\circle*{.2}}
	 \put(3.5,-4.5){\circle*{.2}}
	        \put(5,-6){\circle*{.2}}
	 \put(5,-5){\circle*{.2}}
	 \put(6,-6){\circle*{.2}}   
	 \put(6,-5){\circle*{.2}}
	 \put(5.5,-4.5){\circle*{.2}}
	       \put(-6.7,-6.45){\small{$H_{15}$}}
	    \put(-4.7,-6.45){\small{$H_{16}$}}
	    \put(-2.7,-6.45){\small{$H_{17}$}}
	    \put(-0.7,-6.45){\small{$H_{18}$}}
	    \put(1.3,-6.45){\small{$H_{19}$}}
	    \put(3.3,-6.45){\small{$H_{20}$}}
	    \put(5.3,-6.45){\small{$H_{21}$}}
\end{picture}
\caption{All \emph{connected} graphs on $5$ vertices.}\label{fig:5graphs}
\end{center}
\end{figure}
We compute the eigenvalues and eigenvectors in the following table.

\begin{scriptsize}
\begin{alignat*}{3}
& H_1   \quad && 4,-1,-1,-1,-1            \quad && (1,1,1,1,1),(-1,0,0,0,1),(-1,0,0,1,0),(-1,0,1,0,0),(-1,1,0,0,0)\\
& H_2   \quad && 3,0,0,-1,-2   \quad && (3,2,3,2,2),(0,-1,0,0,1),(0,-1,0,1,0),(-1,0,1,0,0),(-1,1,-1,1,1)\\
& H_3   \quad && 2,0,0,0,-2 \quad && (2,1,1,1,1),(0,-1,0,0,1),(0,-1,0,1,0),(0,-1,1,0,0),(-2,1,1,1,1)\\
& H_4   \quad && 2,\sigma_2,\sigma_2,\sigma_4,\sigma_4            \quad && (1,1,1,1,1),(\sigma_2,\sigma_3,-1,0,1),(-1,\sigma_3,\sigma_2,1,0),(-1,\sigma_1,\sigma_4,1,0),(\sigma_4,\sigma_1,-1,0,1) \\
& H_5 \quad && a,0,a',-1,-1    \quad && (\tilde{a},\tilde{a},\tilde{a},1,1),(0,0,0,-1,1),(\tilde{a}',\tilde{a}',\tilde{a}',1,1),(-1,0,1,0,0),(-1,1,0,0,0)\\
& H_6   \quad && \sqrt{3},1,0,-1,-\sqrt{3}   \quad && (1,\sqrt{3},2,\sqrt{3},1),(-1,-1,0,1,1),(1,0,-1,0,1),(-1,1,0,-1,1),(1,-\sqrt{3},2,-\sqrt{3},1)\\
& H_7   \quad && b,0,0,b',-2            \quad && (-b',1,1,1,1),(0,0,-1,0,1),(0,-1,0,1,0),(-b,1,1,1,1),(0,-1,1,-1,1)\\
& H_8   \quad && \alpha,1,-1,-1,\alpha'           \quad && (-\alpha',1,1,1,1),(0,-1,-1,1,1),(0,0,0,-1,1),(0,-1,1,0,0),(-\alpha,1,1,1,1)\\
& H_9   \quad && \sqrt{6},0,0,0,-\sqrt{6}            \quad && \textstyle{(\frac{\sqrt{6}}{2},1,1,\frac{\sqrt{6}}{2},1),(0,-1,0,0,1),(-1,0,0,1,0),(0,-1,1,0,0),(\frac{-\sqrt{6}}{2},1,1,\frac{-\sqrt{6}}{2},1)}\\
& H_{10}   \quad && c,0,\sigma_2,c',\sigma_4            \quad && (c,2,c,1,1),(0,-1,0,1,1),(\sigma_2,0,\sigma_3,-1,1),(c',2,c',1,1),(\sigma_4,0,\sigma_1,-1,1)\\
& H_{11}   \quad && d_1,d_2,-1,-1,d_3           \quad && (d_1,d_1',d_1',d_1',1),(d_2,d_2',d_2',d_2',1),(0,-1,0,1,0),(0,-1,1,0,0),(d_3,d_3',d_3',d_3',1)\\
& H_{12}   \quad && e_1,e_2,0,-1,e_3           \quad && (h_1,i_1,i_1,h_1,1),(h_2,i_2,i_2,h_2,1),(-1,0,0,1,0),(0,-1,1,0,0),(h_3,i_3,i_3,h_3,1)\\
& H_{13}   \quad && j_1,j_2,0,-1,j_3           \quad && (k_1,\ell_1,\ell_1,k_1,1),(k_2,\ell_2,\ell_2,k_2,1),(-1,0,0,1,0),(0,-1,1,0,0),(k_3,\ell_3,\ell_3,k_3,1)\\
& H_{14}   \quad && m_1,m_2,-1,-1,m_3           \quad && (n_1,1,p_1,p_1,1),(n_2,1,p_2,p_2,1),(0,-1,0,0,1),(0,0,-1,1,0),(n_3,1,p_3,p_3,1)\\
& H_{15}   \quad && q_1,q_2,0,q_3,-2           \quad && (r_1,1,s_1,s_1,1),(r_2,1,s_2,s_2,1),(0,-1,-1,1,1),(r_3,1,s_3,s_3,1),(0,-1,1,-1,1)\\
& H_{16}   \quad && t_1,\sigma_2,t_2,t_3,\sigma_4           \quad && (u_1,1,v_1,v_1,1),(0,-1,\sigma_3,\sigma_2,1),(u_2,1,v_2,v_2,1),(u_3,1,v_3,v_3,1),(0,-1,\sigma_1,\sigma_4,1)\\
& H_{17}   \quad && w_1,1,w_2,-1,w_3           \quad && (x_1,y_1,y_1,z_1,1),(0,-1,-1,2,2),(x_2,y_2,y_2,z_2,1),(0,-1,1,0,0),(x_3,y_3,y_3,z_3,1)\\
& H_{18}   \quad && \zeta_1,\zeta_2,0,\zeta_3,\zeta_4           \quad && (1,\zeta_1,1,\zeta_2,\sqrt{2}),(1,\zeta_2,1,\zeta_4,-\sqrt{2}),(-1,0,1,0,0),(1,\zeta_3,1,\zeta_1,-\sqrt{2}),(1,\zeta_4,1,\zeta_3,\sqrt{2})\\
& H_{19}   \quad && \eta_1,\eta_2,0,\eta_3,\eta_4           \quad && (\eta_1,\mu_1,\xi_1,\mu_1,1),(\eta_2,\mu_2,\xi_2,\mu_2,1),(0,-1,0,1,0),(\eta_3,\mu_2,\xi_3,\mu_2,1),(\eta_4,\mu_1,\xi_4,\mu_1,1)\\
& H_{20}   \quad && \varsigma_1,\varsigma_2,0,\varsigma_3,\varsigma_4           \quad && (\varsigma_1,\phi_1,\tau_1,\phi_1,1),(\varsigma_2,\phi_2,\tau_2,\phi_2,1),(0,-1,0,1,0),(\varsigma_3,\phi_3,\tau_1,\phi_3,1),(\varsigma_4,\phi_4,\tau_2,\phi_4,1)\\
& H_{21}   \quad && \chi_1,\chi_2,\chi_3,-1,\chi_4           \quad && (\chi_1,\psi_1,\omega_1,\psi_1,1),(\chi_2,\psi_2,\omega_2,\psi_2,1),(\chi_3,\psi_3,\omega_3,\psi_3,1),(0,-1,0,1,0),(\chi_4,\psi_4,\omega_4,\psi_4,1)
\end{alignat*}
\end{scriptsize}
where as before $\sigma_1=\frac{1+\sqrt{5}}{2}$, $\sigma_2 = \frac{-1+\sqrt{5}}{2}$, $\sigma_3 = \frac{1-\sqrt{5}}{2}$, $\sigma_4=\frac{-1-\sqrt{5}}{2}$, $\alpha = \frac{1+\sqrt{17}}{2}$, $\alpha'=\frac{1-\sqrt{17}}{2}$. Also, $a=1+\sqrt{7}$, $a'=1-\sqrt{7}$, $\tilde{a}=\frac{a}{3}$, $\tilde{a}'=\frac{a'}{3}$, $b=1+\sqrt{5}$, $b'=1-\sqrt{5}$, $c=\frac{1+\sqrt{13}}{2}$, $c'=\frac{1-\sqrt{13}}{2}$. Next, $\zeta_1=\sqrt{2+\sqrt{2}}$, $\zeta_2=\sqrt{2-\sqrt{2}}$, $\zeta_3=-\zeta_2$, $\zeta_4=-\zeta_1$. Next, $\eta_1=\sqrt{\frac{5+\sqrt{17}}{2}}$, $\eta_2=\sqrt{\frac{5-\sqrt{17}}{2}}$, $\eta_3=-\eta_2$, $\eta_4=-\eta_1$, $\mu_1=\frac{3+\sqrt{17}}{4}$, $\mu_2=\frac{3-\sqrt{17}}{4}$, $\xi_1=\frac{\sqrt{7+\sqrt{17}}}{2}$, $\xi_2=\frac{-\sqrt{7-\sqrt{17}}}{2}$, $\xi_3=-\xi_2$, $\xi_4=-\xi_1$. Next, $\varsigma_1=\frac{1+\sqrt{5+2\sqrt{2}}}{\sqrt{2}}$, $\varsigma_2=\frac{-1+\sqrt{5-2\sqrt{2}}}{\sqrt{2}}$, $\varsigma_3=\frac{1-\sqrt{5+2\sqrt{2}}}{\sqrt{2}}$, $\varsigma_4 = \frac{-1-\sqrt{5-2\sqrt{2}}}{\sqrt{2}}$, $\phi_1=\frac{\varsigma_1}{\sqrt{2}}$, $\phi_2=\frac{-\varsigma_2}{\sqrt{2}}$, $\phi_3=\frac{\sqrt{\varsigma_3}}{\sqrt{2}}$, $\phi_4=\frac{-\varsigma_4}{\sqrt{2}}$, $\tau_1=1+\sqrt{2}$, $\tau_2=1-\sqrt{2}$.

For the remaining values, we give the equations they solve and numerical approximations at the end of the section.

The only new eigenvalues that $H_1$--$H_4$ can offer are top ones, which cannot be in $\cF_5^s$. Ignoring the top eigenvalue, each of $H_5$--$H_{10}$ may only offer one possible flat band, namely $a',-\sqrt{3},b',\alpha',-\sqrt{6},c'$, respectively. By looking at the eigenvectors, we see that there are no $\delta_j\in \{0,1\}$ such that $\sum_{j=1}^5\delta_j\psi_j=0$, except all $\delta_j=0$. Hence, these values are not in $\cF_5^s$.

The same argument shows in fact that none of the graphs $H_j$, $j>5$ offers any flat band. Again, one looks at the eigenvalues outside the list in \eqref{e:f45}, considers the corresponding eigenvectors $\psi$, and checks by hand that $\sum_{j=1}^5\delta_j\psi_j=0$ for $\delta_j\in \{0,1\}$ implies $\delta_j=0$ for all $j$. For this one can use the numerical approximations of the different quantities.

\begin{footnotesize}
\begin{alignat*}{2}
&\qquad\quad \text{Equations}    \qquad &&\qquad\qquad\qquad \text{Approximate values}\\
& \bullet\ d^3-2d^2-4d+2=0,  \qquad && d_1\approx 3.08613,\, d_2\approx 0.428007,\, d_3\approx -1.51414,\\
& dd'=d+2d', d^2=3d'+1 \qquad && d_1'\approx 2.8414,\, d_2'\approx -0.27227,\, d_3'\approx 0.43087\\
& \bullet\ e^3-e^2-6e+2=0,  \qquad && e_1\approx 2.85577,\, e_2\approx 0.321637,\, e_3\approx -2.17741,\\
& \textstyle{h=\frac{e}{2},\ eh=2i+1,} \qquad && h_1\approx 1.42789,\, h_2\approx 0.160819,\,h_3\approx -1.0887  \\
& ei=e+i \qquad && i_1\approx 1.53886,\,i_2\approx -0.474137,\, i_3\approx 0.685278\\
& \bullet\ j^3-j^2-4j+2=0 \qquad && j_1\approx 2.34292,\, j_2\approx 0.470683,\, j_3\approx -1.81361,\\
& \textstyle{k= \frac{1}{j}, \ \ell+1=j\ell} \qquad&& k_1\approx 0.426817,\,k_2\approx 2.12457,\,k_3\approx -0.551388\\
& \qquad && \ell_1\approx 0.744644,\, \ell_2\approx -1.88923, \,\ell_3\approx -0.355416\\
& \bullet\ m^3-2m^2-5m+2=0 \qquad && m_1\approx 3.3234,\, m_2\approx 0.357926,\, m_3\approx -1.68133\\
& \textstyle{n = \frac{2}{m}, \ p+2=mp,} \qquad && n_1\approx 0.601793,\,n_2\approx 5.58774,\,n_3\approx -1.18953\\
& n+2p+1=m\qquad && p_1\approx 0.860806,\,p_2\approx -3.11491,\, p_3\approx -0.745898\\
& \bullet q^3-2q^2-2q+2=0 \qquad && q_1\approx 2.48119,\, q_2\approx 0.688892,\, q_3\approx -1.17009,\\
& \textstyle{r=\frac{2}{q}, \ qs=s+1} \qquad && r_1\approx 0.806063,\,r_2\approx 2.90321,\,r_3\approx -1.70928,\\
& q=r+s+1\qquad && s_1\approx 0.675131, s_2\approx -3.21432, s_3\approx -0.460811\\
& \bullet\ t^3-t^2-5t-2=0 \qquad && t_1\approx 2.93543,\, t_2\approx -0.462598,\, t_3\approx -1.47283\\
&t=u+v, \ tu=2+2v, \qquad && u_1\approx 1.59477,\, u_2\approx 0.699104,\, u_3\approx -1.79387\\
&  tv=u+v+1 \qquad && v_1\approx 1.34067, \, v_2\approx -1.1617,\, v_3\approx 0.321037\\
& \bullet\ w^3-4w-2=0\qquad && w_1\approx 2.21432,\, w_2\approx -0.539189,\, w_3\approx -1.67513\\
& \textstyle{z=\frac{1}{w}, \ wx=2y+1,} \qquad && x_1\approx 1.76271,\, x_2\approx 1.31545,\, x_3\approx -1.07816\\
& wy=x+y, \ w=x+z \qquad && y_1\approx 1.45161,\, y_2\approx -0.854638,\, y_3\approx 0.403032\\
& \qquad && z_1\approx 0.451606,\, z_2\approx -1.85464,\, z_3\approx -0.596968 \\ 
& \bullet\ \chi^4-\chi^3-5\chi^2+\chi+2=0, \qquad && \chi_1\approx 2.64119,\, \chi_2\approx 0.723742,\, \chi_3\approx -0.589216,\\
& \chi\omega=2\psi, \ \chi^2=2\psi+1 \qquad && \chi_4\approx -1.77571, \, \psi_1\approx 2.98793,\, \psi_2\approx -0.238099,\\
& \chi\psi=\chi+\psi+\omega \qquad && \psi_3\approx -0.326412,\, \psi_4\approx 1.07658,\, \omega_1\approx 2.26257,\\
& \qquad && \omega_2\approx -0.657965,\, \omega_3\approx 1.10796,\, \omega_4\approx -1.21256
\end{alignat*}
\end{footnotesize}

\providecommand{\bysame}{\leavevmode\hbox to3em{\hrulefill}\thinspace}
\providecommand{\MR}{\relax\ifhmode\unskip\space\fi MR }
\providecommand{\MRhref}[2]{%
  \href{http://www.ams.org/mathscinet-getitem?mr=#1}{#2}
}
\providecommand{\href}[2]{#2}

\end{document}